%% file: Main.tex
\documentclass[8pt]{report}
\usepackage[english]{babel}
\usepackage[utf8x]{inputenc}
\usepackage{amsmath, bbm, amssymb}
\usepackage{mathrsfs}
\usepackage{graphicx}
\usepackage{mathtools}
\usepackage{braket, tikz, algorithm, algorithmic}
\graphicspath{{Images/}}
\usepackage{hyperref}
\usepackage{amsthm}
\theoremstyle{definition}
\newtheorem{definition}{Definition}

\newtheorem{theorem}{Theorem}
\newtheorem{lemma}{Lemma}
\newtheorem{corollary}{Corollary}
\newtheorem{example}{Example}
\newtheorem{procedure}{Procedure}
\DeclareMathOperator{\Tr}{Tr}
\DeclareMathOperator{\diag}{diag}
\DeclareMathOperator{\bin}{bin}
\DeclareMathOperator{\col}{col}
\DeclareMathOperator{\Alt}{Alt}

\hypersetup{
    colorlinks=true,
    linkcolor=blue,
    filecolor=magenta,      
    urlcolor=cyan,
}
 
\begin{document}
\begin{titlepage}
\newcommand{\HRule}{\rule{\linewidth}{0.1mm}} 
\center 
 
\textsc{\Large }\\[0.5cm] 
\textsc{\Large }\\[0.5cm] 

\includegraphics[scale=.08]{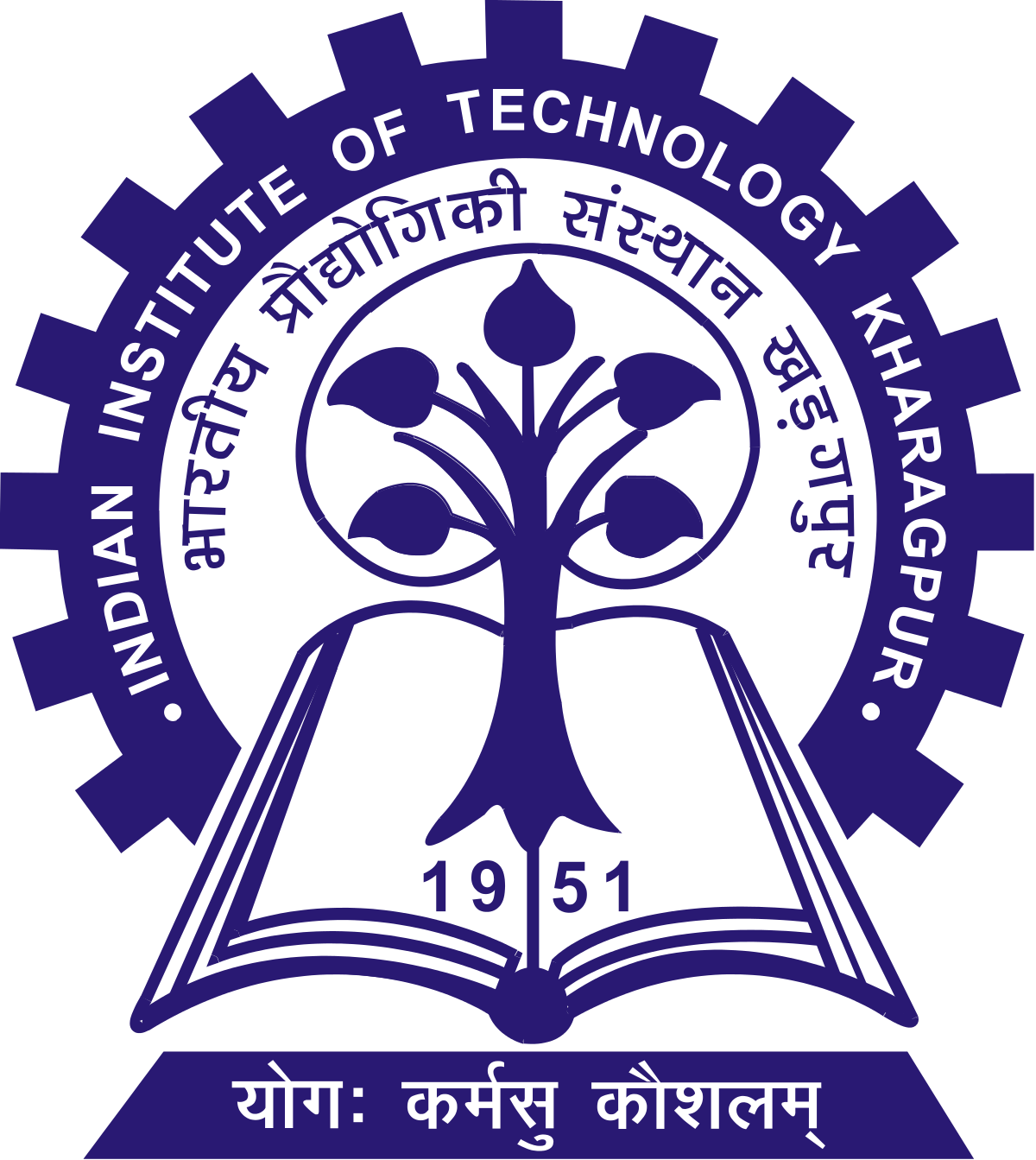}

\Large \textbf{Indian Institute of Technology Kharagpur\\
Department of Mathematics and Department of Physics}
\vfill
\HRule \\[0.6cm]
{ \Large \bfseries A switching approach for perfect state transfer over a scalable and routing enabled network architecture with superconducting qubits}\\[0.1cm] 
\HRule \\[2cm]
 

\Large Master Thesis Project (PH57002)\\
\vfill

\begin{minipage}{0.4\textwidth}
\begin{flushleft} \large
\emph{Student:}\\
\textbf{Siddhant Singh}
\\{\textit{\textbf{(15PH20030)}}}  
\end{flushleft}

\end{minipage}
\begin{minipage}{0.4\textwidth}
\begin{flushright} \large
\emph{Supervisors: \\}
\textbf{Bibhas Adhikari}\\ 
\textbf{Sonjoy Majumder}
\end{flushright}
\end{minipage}\\[1cm]
{\large \today}\\[1cm] 

\vfill 

\end{titlepage}
\section*{\begin{center}
    Abstract
\end{center}}
\input{Chapters/Abstract}
\newpage
\tableofcontents          
\listoffigures
\newpage


\chapter{Introduction to Quantum Perfect State Transfer (PST)}
\input{Chapters/Introduction}

\newpage
\chapter{Graph Theory}
\input{Chapters/Graph_Theory}
\newpage
\chapter{Perfect State Transfer under Corona product of signed graphs}
\input{Chapters/Corona}
\newpage
\chapter{Scalable and routing enabled network for Perfect State Transfer}
\input{Chapters/TwoHop}
\newpage
\chapter{Perfect State Transfer for qudit (d-level systems) networks}
\input{Chapters/QuditsPST}
\newpage
\chapter{Physical implementation for qubit networks with superconducting circuits}
\input{Chapters/PhysicalImplementation}
\newpage
\chapter{Conclusion and future work}
\input{Chapters/Future_Work}

\newpage
\bibliographystyle{unsrt}
\bibliography{Bibliography.bib}


\end{document}

%% file: Chapters/Abstract.tex
We propose a hypercube switching architecture for the perfect state transfer (PST) where we prove that it is always possible to find an induced hypercube in any given hypercube of any dimension such that PST can be performed between any two given vertices of the original hypercube. We then generalise this switching scheme over arbitrary number of qubits where also this routing feature of PST between any two vertices is possible. It is shown that this is optimal and scalable architecture for quantum computing with the feature of routing. This allows for a scalable and growing network of qubits. We demonstrate this switching scheme to be experimentally realizable using superconducting transmon qubits with tunable couplings. We also propose a PST assisted quantum computing model where we show the computational advantage of using PST against the conventional resource expensive quantum swap gates. In addition, we present the numerical study of signed graphs under Corona product of graphs and show few examples where PST is established, in contrast to pre-existing results in the literature for disproof of PST under Corona product. We also report an error in pre-existing research for qudit state transfer over Bosonic Hamiltonian where unitarity is violated. 

%% file: Chapters/Introduction.tex
\section{Introduction}

In quantum computation, it is often required to transfer an arbitrary quantum state from one site to another \cite{ref:48}. These two sites may belong to the same quantum processor or different processors. The latter one is not trivial for many quantum information processing (QIP) realizations, such as, solid state quantum computing and superconducting quantum computing \cite{ref:23}. This is because the state transmission channel is not a computational space for either processor and cannot involve manipulation. In large scale quantum computation, it is a very important task to be able to transfer a quantum state within a processor as well as between two physically distant QIP processors with robust transmission lines. It is also important to find the physical systems also which support this quantum information exchange between distant sites. For short distance communications (say between adjacent quantum processors), alternatives to interfacing different kinds of physical systems are highly desirable and have been proposed, for example, for ion traps \cite{ref:39}\cite{ref:34}, superconducting circuits \cite{ref:41}\cite{ref:42}, etc. The task of state transfer is thought with the intention of reducing the control required to communicate between distant qubits in a quantum computer \cite{ref:54}.
    
    Quantum state transfer with 100\% fidelity is known as perfect state transfer (PST) and this idea using interacting spin-1/2 particles was first proposed in \cite{ref:1} and established the connection between graph theoretic networks and actual quantum networks for quantum processors in the first excitation subspace of many-body qubit network \cite{ref:15}. For graph structures, usually the XY coupling Hamiltonian and the Heisenberg spin interaction are considered. Due to this connection, a quantum architecture can be designed purely in graph theoretic fashion (determining the qubits' mutual connectivity) and can be realized by physical systems. It was established that PST is possible in spin-1/2 systems and other bosonic networks without any additional action and manipulation from senders and receivers \cite{ref:49}. PST only requires access to two spins at each end of the spin network while all other spins in the network act like a channel for transfer and are not computational spins. In general, this involves mixed states of the network qubits \cite{ref:57}, however, showing PST for pure states in a graph suffices to prove the phenomenon. PST can be used in entanglement transfer, quantum communication, signal amplification, quantum information recovery and implementation of universal quantum computation \cite{ref:18}\cite{ref:49}\cite{ref:60}\cite{ref:61}.
    
    PST in graphs is a rare phenomenon and only very few graphs and class of graphs are known to exhibit the phenomenon of PST. For this reason, the idea of pretty good state transfer is also studied, where the fidelity is a little less than unity but offers a large number of graphs that support state transfer \cite{ref:66}\cite{ref:67}\cite{ref:68}. The task is to find graph structures which support PST for as many pair of vertices as possible and possibly grow under some operation (scalability of networks). It is important to find the class of graphs where it occurs and equally important to find graphs where it does not occur \cite{ref:58}\cite{ref:12}. Researchers aim to find class of graphs as well as as various products of graph to establish a growing scalable network supporting PST \cite{ref:4}\cite{ref:5}, and in general these graphs can be weighted \cite{kendon2011perfect}. More general graphs such as signed graphs \cite{ref:50} and oriented graphs \cite{ref:19} are also studied. In this way, we essentially define a quantum computig architecture. PST for qudits or higher dimensional spins over weighted graph is also classified for some networks \cite{ref:53}\cite{ref:21}. This was further developed for arbitrary states and large networks in \cite{ref:4} and \cite{ref:5}.

    PST scheme established in \cite{ref:4} and \cite{ref:5} allows PST over arbitrary long distances with the use of Cartesian product of one-link and two-link graphs which support PST under the XY as well as Heisenberg interaction of spins. It is known that one-link and two-link chain graphs exhibit PST between the end vertices \cite{ref:4}. And this feature carries over to the antipodal vertices of the resulting Cartesian product of these graphs with themselves, which become the pair of vertices exhibiting PST in the same time. First shortcoming in this model is that of the impossibility of routing \cite{ref:2}\cite{ref:16}. The second being that only a pair of antipodal vertices of the graph support PST which becomes less useful as the graph scales to a larger network. This would involve constructing a very large network just to enable PST between a pair of antipodal qubits of a network. Third being that this architecture scales the number of qubits  with the factor of 2 which will be very large gap for the larger dimensional hypercubes as the network scales up. This motivates for finding a quantum computing architecture which would allow us routing to any given vertex of the graph as well as enables arbitrary number of vertices while still preserving perfect fidelity and routing to any vertex starting from any other given vertex. A switching was proposed in \cite{ref:56} where in a complete graph $K_n$, switching off one link establishes PST in non-adjacent qubits. This enables PST for more vertices but still does not enable routing to different vertices and there is no scalability, the graph remains fixed. One attempt at switching and routing is proposed in \cite{ref:59} which involves creating new edges and coupling for qubits, however, is still not scalable. Routing in special regular graphs was proposed in \cite{ref:62}\cite{ref:64}. It leads us to the motivation that only quantum mechanical processes may not not be sufficient to fulfill our requirements, that is the perfect state transfer between any two vertices of a graph of arbitrary number of vertices. Therefore, we propose a hybrid of classical combinatorial and quantum information theoretic method, such that, a perfect quantum state transfer is possible between any two vertices of the graph.

    In this work, we propose a solution to both the problems of routing and scalability for quantum architecture where the network will be enabled with PST from all-to-all nodes for any arbitrary number of qubits, which fits with the idea of Noisy Intermediate-Scale Quantum (NISQ) processors \cite{ref:29}. And this task can be accomplished in just two quantum operations only. Our architecture also features the addition of any arbitrary number of qubits in an already constructed network according to our scheme while preserving both the properties. Thus, this is a possible and optimal solution for a scalable architecture for quantum information processing.  Our results in this work hold both for XY-coupling as well as the Laplacian interaction Hamiltonian. We also propose the idea of PST assisted quantum computation where PST can be used in contrast to large number of SWAP gates between any two distant qubits when the quantum circuit depth is very high and thereby reducing the complexity of a large quantum circuit. This involves PST over the computational qubits of a quantum processor. We also present analytically and simulate numerically the experimental implementation of our architecture using superconducting circuits thereby showing the implementation of PST in superconducting circuits for the very first time. Apart from these main results, we also report the PST in qudit systems and analysis of PST under Corona product of certain special graphs. Chapters 1 and 2 form the part of the preliminary literature and chapters 3,4,5 and 6 are the original contributions from this thesis.

\section{Spin Hamiltonian dynamics for Perfect State Transfer}
\label{sec:firstexcitation}
The idea for perfect state transfer of arbitrary states is to establish the connection between the graph theoretic approach and spin Hamiltonian. The system of spins can be translated into a corresponding graph where the dynamics can be explored by the structure of the graph governed by its adjacency marix and Laplacian which are in one-to-one correspondence with the connectivity of the spins in the physical picture. Arbitrary state of spin in a lattice is simply a qubit state. The principle problem is that the Hilbert space of a graph $G$ with $n$ vertices is given by $\mathbbm{C}^n$, and the Hilbert space of a spin-1/2 (or generally any many-body qubit network) particle attached to each vertex of the same graph $G$ is $\mathbbm{C}^{2^n}$, which is exponentially larger. Graphs and their products will be discussed in detail in chapter \ref{chap:Graph}. For spin-dynamics we have generally two kind of interaction Hamiltonians: The XY coupling adjacent interaction and the Heisenberg interaction model. We want to establish a connection between the dynamics of $n$ number of spin-1/2 particles interacting according to these two kind of interactions on a graph $G$ and the structure of $G$ itself. 

The central idea of this equivalence is that complicated physics of a system of distinguishable spin-1/2 particles interacting pairwise on a simple geometry given by an undirected simple graph $G$ are equivalent to the sometimes physics of a single free spinless particle hopping on a much complicated graph $\mathcal{G}$ (which is some disjoint union of the graphs $G_k$, which are related to $G$) \cite{ref:15}. This can be understood as direct application of a special graph product, called the \textit{wedge product} of graphs which is discussed in section \ref{sec:wedge}.

Consider $n=|V|$ \textit{distinguishable} spin-1/2 subsystems, each at one vertex of the graph $G(V,E)$, where $V(G)$ is the finite vertex set and $E(G)$ is the edge set (a two-element collection of vertices) of the graph. We say distinguishable spins because we are able to label them with the labeling of the graph vertices. We fix a labeling $i={1,2,...,n}$ of the vertices. This labeling induces an ordering of the vertices which we write as $v_i>v_j$ if $i>j$, where $v_i$ is the $i$th vertex of the graph $G$. The degree $\deg (v_i)$ of a vertex $v_i$ is equal to the number of edges which have $i$ as an endpoint. The adjacency matrix $A(G)$ for $G$ is the \{0,1\}-matrix of size $|V(G)|\times |V(G)|$ which has a 1 in the $(i,j)$ entry if there is an edge connecting $v_i$ and $v_j$. Define the Laplacian $L(G)$ of the graph as $D(G)-A(G)$ where $D(G)$ is the degree matrix for $G$ defined as $D(G)=\diag \{\deg(v_i)\}$. We also define the Hilbert $\mathcal{H}(G)$ space of the graph $G$ to be the vector space over $\mathbbm{C}$ generated by the orthonormal vectors $|i\rangle$, $\forall i \in V(G)$, with the canonical inner product $\langle i|j\rangle=\delta_{ij}$. For details in graph theory, refer to chapter \ref{chap:Graph}. Now we define the two mentioned interaction Hamiltonian for the pairwise interactions between the spins. The first is the XY model in two spatial degrees of freedom,
\begin{equation}
\label{eqn:XYmodel}
    H_{XY}=\sum_{(i,j)\in E(G)}J_{ij}\left(\sigma^x_i\sigma^x_j+\sigma^y_i\sigma^y_j\right)=\sum_{(i,j)\in E(G)}2J_{ij}\left(\sigma^+_i\sigma^-_j+\sigma^-_i\sigma^+_j\right),
\end{equation}
\begin{equation}
    =\frac{1}{2}\sum_{\langle i,j \rangle}\left(\sigma^x_i\sigma^x_j+\sigma^y_i\sigma^y_j\right) \quad (\text{for }J_{ij}=\frac{1}{2}\quad \forall \quad \langle i,j \rangle)
\end{equation}
where $\langle i,j \rangle$ denotes an adjacent pair of vertices on the graph (which have an edge between them) and $\sigma_i^{+,-}$ are the ladder operators for the $i$th spin (qubit) such that $\sigma_i^x=\sigma^+_i+\sigma^-_i$ and $\sigma_i^y=i\sigma^+_i-i\sigma^-_i$, with $\sigma^{x,y}_i$ as the Pauli matrices for spin-1/2 system at the $i$th vertex. The second connection is via the three-dimensional Heisenberg model,
\begin{equation}
    H_{Hei}=-\sum_{(i,j)\in E(G)}J_{ij}\Vec{\sigma}_i\cdot \Vec{\sigma}_j+\sum_jB_j\sigma^z_j
\end{equation}
\begin{equation}
    =-\frac{1}{2}\sum_{\langle i,j \rangle}\left( \Vec{\sigma}_i\cdot \Vec{\sigma}_j -I_iI_j\right) \quad (\text{for }J_{ij}=\frac{1}{2}\quad \forall \quad \langle i,j \rangle)
\end{equation}
where, $\Vec{\sigma}_i$ is Pauli matrix vector $\Vec{\sigma}_i=(\sigma^x_i,\sigma^y_i,\sigma^z_i)$ for the $i$th spin and $I_i$ is the identity operator for the $i$th vertex. We take $J_{ij}=1/2$ and choose local fields $B_j$ at the $j$th site such that $\sum_jB_j\sigma^z_j=\mathbb{I}$ (which makes the Hamiltonian coincide with the Laplacian of the graph)for the uniformly coupled system with edge weight unity for both the models.

There is one peculiar conservation property of both of these Hamiltonians that they conserve the total spin along the $z$-axis of the whole system. Formally, we define $S^z := \sum_{i\in V}\sigma^z_i$, and it can be verified that
\begin{equation}
     \left[ S^z, \sigma^x_i\sigma^x_j+\sigma^y_i\sigma^y_j\right]=0 \quad \forall \quad  i,j \in V
\end{equation}
and also
\begin{equation}
    \left[ S^z, \Vec{\sigma}_i\cdot \Vec{\sigma}_j\right]=0 \quad \forall \quad  i,j \in V.
\end{equation}
For Heisenberg Hamiltonian, it even commutes with total spin along $x$- or $y$-axis also, or generally along any axis with suitably defined spin operators along that axis. These commutation relations are enough to establish the idea of different excitation subspaces for the action of the Hamiltonian. For example, if the system had two spins excited and all other in ground state, then throughout the quantum evolution of the system under the spin Hamiltonian will conserve this total spin as two excitations. Therefore, the action of $H_{XY}$ and $H_{Hei}$ breaks the Hilbert space $\mathcal{H}\cong \mathbbm{C}^{2^n}$ into a direct sum 
\begin{equation}
    \mathcal{H}\cong \bigoplus_{k=0}^n \Gamma^k
\end{equation}
where the vector subspaces $\Gamma^k$ constitute the elements as follows
\begin{equation}
    \begin{split}
        & \Gamma^0=\{|00...0\rangle\},\\
        & \Gamma^1=\{|10...0\rangle,|01...0\rangle,...,|00...1\rangle \},\\
        & \Gamma^2=\{|110...0\rangle,|101...0\rangle,...,|000...11\rangle \},\\
        & \quad \quad \vdots \\
        & \Gamma^n=\{|11...1\rangle \}.\\
    \end{split}
\end{equation}
This implies that starting with a ket in $\Gamma^k$, the system will evolve in the linear combination of vectors in $\Gamma^k$ strictly. Projectors $P_k$ can be defined for each $\Gamma_k$, then the total spin for each $\Gamma_k$ is $\Tr(S^zP_k)=k$ (use standard basis) which is also justified by the physical argument for total spin in each subspace. The dimension of $\Gamma_k$ is
\begin{equation}
    \text{dim} \left(\Gamma_k \right)= \binom{n}{k}
\end{equation}
We also note that the subspace $\Gamma^1$ matches exactly with the vertex space of the graph $G$ itself. Dynamics in this subspace is simply the hopping between different vertices while at each point in time one of the vertices is excited. There is a deep relationship between the subspaces $\Gamma^k$ and the exterior vector spaces via the wedge product. The vector space $\Gamma^k$ is generated by the $k$th wedge product of $G$ denoted as $\wedge ^k (\mathcal{H}_G)$. For the first wedge product of $G$, we get $\Gamma^1$ and action of corresponding adjacency matrix $A(G)$ is identical to the action of the Hamiltonian $H_{XY}$. And the action of the graph Laplacian is identical to the action of the Heisenberg hamiltonian $H_{Hei}$. In the first excitation space, we denote these Hamiltonians as $H^1_{XY}$ and $H_{Hei}^1$ respectively. For general higher order wedge product we formulate it in section \ref{sec:wedge}.


\section{Perfect State Transfer in general networks for arbitrary states in Heisenberg model}
Simplest and fundamental system for the construction of network for perfects transfer are the ferromagnetic chains (or network) in Heisenberg model for spins.  Consider the general graph shown in figure \ref{fig:spinpst}, where the vertices are spins and the edges connect spins which interact. Say there are $n$ spins in the graph and these are labeled $1,2,...,n$. The Hamiltonian is given by
\begin{equation}
    H^1_{Hei}=-\sum_{\langle i,j\rangle}J_{ij}\Vec{\sigma}_i\cdot \Vec{\sigma}_j-\sum_{i=1}^nB_i\sigma^z_i
\end{equation}
as before where $\Vec{\sigma_i}=(\sigma^x_i,\sigma^y_i,\sigma_i^z)$ are the Pauli spin matrices for the $i$th spin, $B_i> 0$ are static magnetic fields and $J_{ij}>0$ are coupling strengths, and $\langle i,j\rangle$ represents pairs of adjacent spins which are coupled. $H_G$ describes an arbitrary ferromagnet with \textit{isotropic} Heisenberg interactions. We now assume that the state sender $A$ is located closest to the $s$th (sender) spin and the state receiver $B$ is located closest to the $r$th (receiver) spin (these spins are shown in figure \ref{fig:spinpst}. All the other spins will be called channel spins as they are involved in transferring the state of the qubit (spin) identical to a quantum channel. In the original idea \cite{ref:1}, it is also assumed that the sender and receiver spins are detachable from the chain. In order to transfer an unknown state to Bob, $A$ replaces the existing sender spin with a spin encoding the state to be transferred. After waiting for a specific amount of time, the unknown state placed by $A$ travels to the receiver spin with some fidelity. $B$ then picks up the receiver spin to obtain a state close to the the state Alice wanted to transfer. As individual access or individual modulation of the channel spins is never required in the process, they can be constituents of rigid 1D magnets (for instance).

Perfect transfer of a state in a many-qubit system modeled as a combinatorial graph in which the edges of the graph represents coupling of qubits, is defined by starting with a single qubit state $\rho_{\text{qubit}}^A$ on some vertex $A$, with $\rho_{\text{in}}$ in the state of the rest of the qubits, and after evolution for some time $t_0$ under a fixed Hamiltonian $H$, the output state
\begin{equation}
    \label{eqn:pstgeneral}
    e^{-iHt_0}(\rho_{\text{qubit}}^A\otimes \rho_{\text{in}})e^{iHt_0}=\rho_{\text{qubit}}^B\otimes \rho_{\text{out}}
\end{equation}
is produced, thereby transmitting the input qubit to another desired vertex $B$ of the graph. In general, $\rho_{\text{qubit}}$ is a density matrix, however, in this thesis we consider that it corresponds to a pure state (which is sufficient to demonstrate the idea of PST). The most simplified case for such realization is the one-dimensional chain of qubits.

We assume that initially the system is initially cooled
to its ground state $|0\rangle = |000...0\rangle$ where $|0\rangle$ denotes the spin down state (i.e., spin aligned along $−z$ direction) of a spin. This is shown for a 1D chain in the upper part of figure \ref{fig:spinpst}. We set the ground state energy $E_0 = 0$ (i.e., redefine $H_G$ as $E_0+H_G$). We also introduce the class of states $|j\rangle = |00...010....0\rangle$ (where $j = 1,2,..s,..r,..,n$, i.e., the vertex space of $G$) in which the spin at the $j$th site has been flipped to th $|1\rangle$ state. To start the protocol, $A$ places a spin in the unknown state $|\psi_{in}\rangle = \cos(\theta/2)|0\rangle + e^{i\phi} \sin(\theta/2)|1\rangle$ at the $s$th site in the spin chain. We can describe the state of the whole chain at this instant (time $t = 0$) as
\begin{equation}
    |\Psi(0)\rangle=\cos(\theta/2)|0\rangle+e^{i\phi}\sin(\theta/2)|s\rangle
\end{equation}
\begin{figure}[hbtp]
    \centering
    \includegraphics[scale=0.4]{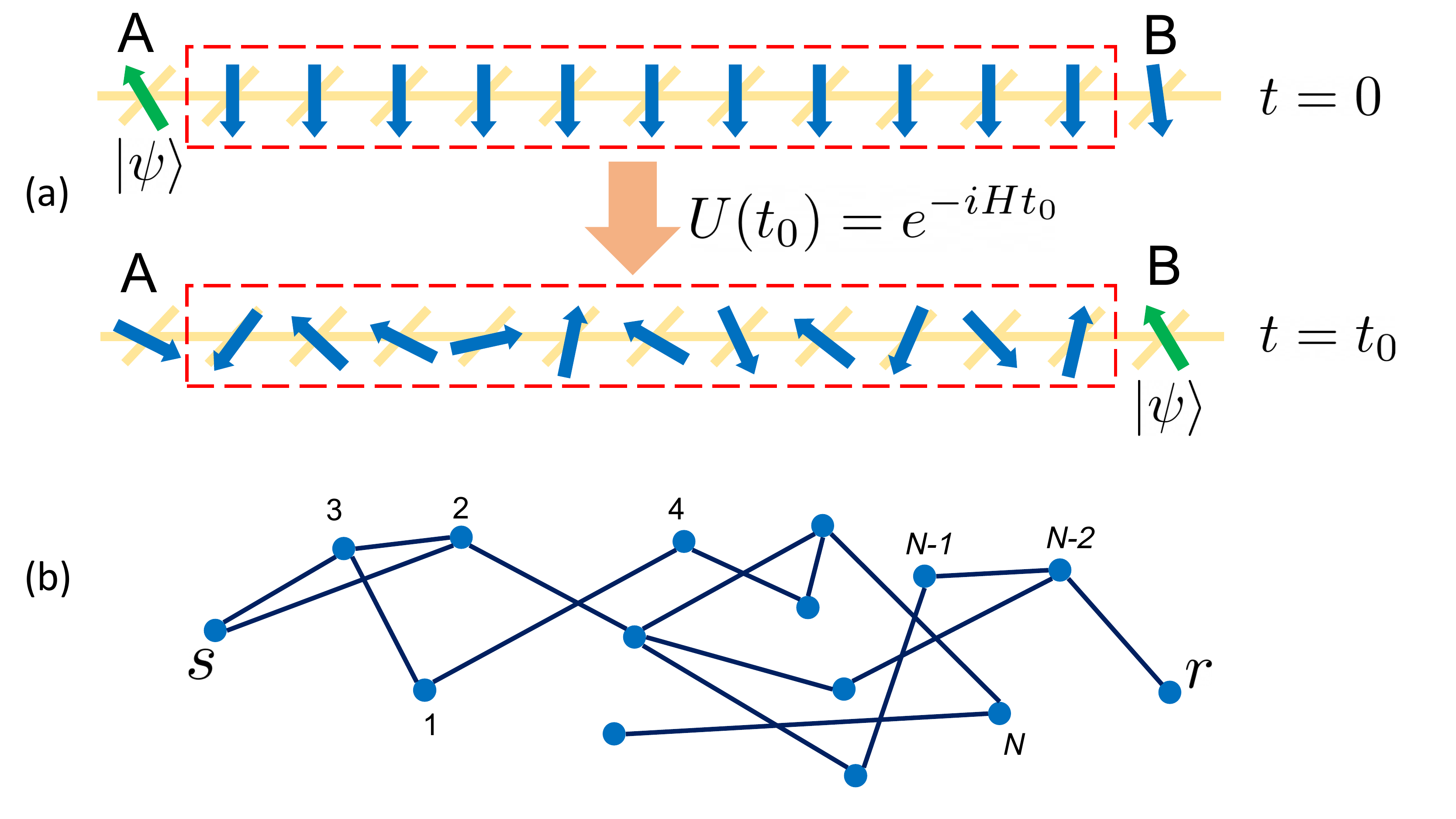}
    \caption{The part (a) of the figure shows the quantum communication protocol. Initially the spin chain is in its ground state in an external magnetic field. Alice and Bob are at opposite ends of the chain. Alice places the quantum state she wants to communicate on the spin nearest to her. After a while, Bob receives this state with some fidelity on the spin nearest to him. Part (b) shows an arbitrary graph of spins through which quantum communications may be accomplished using this protocol. The communication takes place from the sender spin $s$ to the receiver spin $r$.}
    \label{fig:spinpst}
\end{figure}
$B$ wants to retrieve this state, or a state as close to it as possible, from the $r$th site of the graph. Then he has to wait for a specific time till the initial state $|\Psi(0)\rangle$ evolves to a final state which is as close as possible to $\cos(\theta/2)|0\rangle+e^{i\phi}\sin(\theta/2)|r\rangle$. As $[H_G,\sum_i\sigma_i]=0$, the state $|s\rangle$ only evolves to states $|j\rangle$ and the evolution of the spin-graph (with $\hbar = 1$) is
\begin{equation}
\label{eqn:bose1evo}
        |\Psi(t)\rangle=\cos(\theta/2)|0\rangle+e^{i\phi}\sin(\theta/2)\sum_{j=1}^n\langle j|e^{-iH^1_{Hei}t}|s\rangle |j\rangle
\end{equation}
For the Perfect State Transfer (PST) to occur from the $s$th site to the $r$th site, we should have $\langle r|e^{-iH_Gt}|s\rangle:=f^n_{r,s}=e^{i\eta}$ (where the phase $e^{i\eta}$ can be compensated and corrected later by $B$) or simply $|f^n_{r,s}|=1$ for some finite $t=t_0$, is enough for pure states. The state of the $r$ spin will, in general, be a mixed state, and can be obtained by tracing off the states of all other spins from $|\Psi(t)\rangle$. This evolves with time as
\begin{equation}
    \rho_{\text{out}}(t)=P(t)|\Psi_{\text{out}}(t)\rangle \langle \Psi_{\text{out}}(t)|+(1-P(t))|0\rangle\langle 0|
\end{equation}
where
\begin{equation}
    |\Psi_{\text{out}}(t)\rangle=\frac{1}{\sqrt{P(t)}}\left(  \cos\frac{\theta}{2}|0\rangle+e^{i\phi}\sin\frac{\theta}{2}f^n_{r,s}(t)|1\rangle \right)
\end{equation}
with $P(t)=\cos^2(\theta/2)+\sin^2(\theta/2)|f^n_{r,s}|^2$ is the normalization of the state at any time $t$ and $f^n_{r,s}=\langle r|e^{-iH^1_{Hei}t}|s\rangle$. Note that $f^N_{r,s}(t)$ is just the transition amplitude of an excitation (the $|1\rangle$ state) from the $s$th to the $r$th site of a graph of $n$ spins. It is also equal to the fidelity between these states for the case of pure states. It is decided that $B$ will pick up the $r$th spin (and hence complete the communication protocol) at a predetermined time $t = t_0$. We show later in chapter \ref{chap:phyreal} that the phenomenon of perfect state transfer is not only restricted to spin Hamiltonians but is a general property of many Hamiltonians which are similar in action with the $H_{XY}$ or $H_{Hei}$ models. We show that superconducting transmon qubit network Hamiltonian also allows perfect state transfer along with the power of quantum computation.


\section{Review on fidelity and distance measures}
Fidelity is a measure of how apart two states are in the Hilbert space, or to say, simply the overlap of states (see chapter 9 in \cite{ref:55} for distance measures). Fidelity for two density matrices $\rho$ and $\sigma$ is defined as
\begin{equation}
    F(\rho,\sigma)=\sqrt{\sqrt{\sigma}\rho\sqrt{\sigma}}\in [0,1]
\end{equation}
which takes the form (when $\sigma=|\psi\rangle\langle\psi|$, that is, a pure state) of 
\begin{equation}
    F(\rho,\sigma)=\sqrt{\langle\psi|\rho|\psi\rangle}.
\end{equation}
We can treat $\sigma=|\psi\rangle\langle\psi|$ as the final desired (pure) state of the system and $\rho$ as the noisy obtained state (after an evolution) and check the fidelity of the two as a check for closeness.
Fidelity is related to the trace distance $D(\rho,\sigma)$ by the following mathematical relations
\begin{equation}
    1-F(\rho,\sigma)\leq D(\rho,\sigma)\leq \sqrt{1-F^2(\rho,\sigma)}
\end{equation}
where the trace distance is defined as $D(\rho,\sigma)=\frac{1}{2}\Tr|\rho-\sigma|$. They can be used interchangeably but we will stick with the measure of the fidelity for the purpose of perfect state transfer.


\section{Conditions on Perfect State Transfer for pure states}
First let us recall that for the general Hamiltonian (for the first excitation subspace)
\begin{equation}
    H^1=\sum_{i<j}J_{ij}|i\rangle\langle j|+J_{ij}^*|j\rangle\langle i|+\sum_{i=1}^nB_i|i\rangle\langle i|
\end{equation}
($J_{ij}\in\mathbbm{C}$ in general) where $i$ and $j$ are connected, has the associated space-vectors $|n\rangle$ (again, the vertex vector space of $G$) and the sender node $A$ and the receiver node $B$ are defined in the same manner. We can easily diagonalize any such given Heisenberg or XY-model Hamiltonian in the manner $H_1=\sum_{j=1}^n \lambda_j|\lambda_j\rangle\langle \lambda_j|$. We want to prove necessary and sufficient conditions for perfect state transfer in the first excitation subspace of a spin-preserving Hamiltonian $H^1$. These conditions can be expressed as the existence of a state transfer time $t_0$ and transfer phase $\phi$ in a condition on the eigenvectors,
\begin{equation}
    |\langle A|\lambda_j\rangle|=|\langle B|\lambda_j\rangle|
\end{equation}
for all $j$, and on the eigenvalues,
\begin{equation}
\label{eqn:phij}
    \lambda_j=-\phi -\varphi_j +2\pi m_j
\end{equation}
for all $j$ for which $\langle A|\lambda_j \rangle\neq 0$, where $m_j$ is an integer. The definition for $\varphi_j$ is
\begin{equation}
    \varphi_j=\arg \left( \frac{\langle\lambda_j|B\rangle}{\langle\lambda_j|A\rangle}\right)
\end{equation}
To prove necessity, we start from the definition of state transfer in the single excitation subspace, requiring that there exist a $t_0$ and $\phi$ such that
\begin{equation}
    e^{-iH^1t_0}|A\rangle=e^{i\phi}|B\rangle.
\end{equation}
This is equivalent to requiring that
\begin{equation}
\label{eqn:modcondition}
    |\langle B|U(t)|A\rangle|=1
\end{equation}
where $U(t)$ is the evolution operator. This is exactly the fidelity for the pure states.
By taking the overlap with an eigenvector,
\begin{equation}
    e^{-i\lambda_jt_0}\langle \lambda_j|A\rangle=e^{i\phi} \langle \lambda_j|B\rangle
\end{equation}
then the amplitude and phase matching gives the previously stated conditions precisely. Having proved necessity, we now prove sufficiency. Assume that a suitable finite $t_0$ and $\phi$ exist. So,
\begin{equation}
\label{eqn:spectralevo}
    e^{-iH^1t_0}|A\rangle=\sum_j|\lambda_j\rangle\langle\lambda_j|A\rangle e^{-i\lambda_jt_0}
\end{equation}
We can now impose the conditions on $\lambda_j$,
\begin{equation}
\begin{split}
    e^{-iH^1t_0}|A\rangle &= \sum_{\langle\lambda_j|A\rangle\neq 0}|\lambda_j\rangle\langle\lambda_j|A\rangle e^{i(2\pi m_j+\phi +\varphi_j)}\\
    &=\sum_{\langle\lambda_j|A\rangle\neq 0}|\lambda_j\rangle\langle\lambda_j|A\rangle e^{i\phi}\left( \frac{\langle\lambda_j|B\rangle}{\langle\lambda_j|A\rangle}\right)\\
    &=e^{i\phi} |B\rangle
\end{split}
\end{equation}
This yields a rather simple set of conditions which one use to verify that perfect transfer occurs in a network.

Alternatively, we start with the state $|A\rangle$ then evolve it according to the free Schrodinger quantum evolution as $e^{-iH^1t}|A\rangle$ such that for some finite $t_0$ we have $e^{-iH^1t}|A\rangle=e^{i\phi}|B\rangle$. Taking inner product with $\langle A|$ and modulus we have equation (\ref{eqn:modcondition}) which can be rewritten using equation (\ref{eqn:spectralevo}) as
\begin{equation}
\label{eqn:pstcond}
    \bigg| \sum_{j=1}^n e^{-i\lambda_j t_0} \langle B|\lambda_j\rangle\langle\lambda_j|A\rangle\bigg|=1.
\end{equation}
This condition is all we need to ensure for completing the task of perfect state transfer for pure states. In higher excitation subspaces generally more conditions are required to ensure perfect state transfer. However, we should carefully note that condition \ref{eqn:pstcond} signifies a state swap between $|A\rangle$ and $|B\rangle$ whereas the actual problem was to have
\begin{equation}
    \left( \cos \frac{\theta}{2}|0\rangle+e^{i\phi}\sin \frac{\theta}{2}|1\rangle \right)_A^{\otimes (n-1)}\longrightarrow \left( \cos \frac{\theta}{2}|0\rangle+e^{i\phi}\sin \frac{\theta}{2}|1\rangle \right)_B^{\otimes (n-1)}.
\end{equation}
This can be understood in the sense that the actual state of the qubit is encoded in $A$ which gets swapped with $B$ at exactly $t=t_0$ with fidelity of unity.


\section{Symmetries in network}
Symmetries are very important tool in understanding any
system. The construction of perfect state transfer chains originally relied heavily on an assumption of symmetry which was subsequently proven to be necessary \cite{ref:2}. We are interested in whether every perfect transfer Hamiltonian $H^1$ has a symmetry operator $S$ which satisfies $SH^1S^\dagger=H^1$ and also $S|A\rangle=|B\rangle$.

The existence of a symmetry can be proven by construction. By
defining a unitary rotation that is diagonal in the basis of the Hamiltonian, it will clearly satisfy the commutation property. Specifying the phases as
\begin{equation}
    S=\sum_{\langle\lambda_j|A\rangle\neq 0}e^{i\varphi_j}|\lambda_j\rangle\langle\lambda_j|+\sum_{\langle\lambda_j|A\rangle= 0}|\lambda_j\rangle\langle\lambda_j|
\end{equation}
allows us to verify the desired transformation
\begin{equation}
    S|A\rangle=\sum_{\langle\lambda_j|A\rangle\neq 0}e^{i\varphi_j}|\lambda_j\rangle\langle\lambda_j|A\rangle=|B\rangle.
\end{equation}
For a \textit{real} Hamiltonian $H^1$, $S^2 = \mathbb{I}$, so $S|B\rangle = |A\rangle$. It is worth observing that there is still continuous freedom in the definition of $S$ — the phases that are applied to the eigenvectors for which $\langle A|\lambda_n\rangle = 0$ — which gives away to see that $S$ is not necessarily a permutation (which cannot be continuous). If $S$ were a permutation, it would have to be the mirror symmetry operator. If one knows the symmetry operators of a system for some a priori reason, this identifies the values $\varphi_n$ (the eigenvalues of $S$) and associates them with specific eigenspaces. Hence, for systems where $S$ can be identified, and the eigenvalues can be modified while preserving the symmetry, we should be able to construct perfect transfer networks. This was the key insight for designing chains, and it can hopefully now be applied in other scenarios.

\subsection{Mirror symmetry}
With mirror symmetry in hands, we have periodicity implied as follows \cite{ref:2}. With a system capable of perfect state transfer, initialised in the state $|A\rangle$, at time $t_0$ we have the state
\begin{equation}
    e^{-iH^1t_0}|A\rangle=e^{i\phi}|B\rangle
\end{equation}
but by the definition of a symmetric system, $A$ and $B$ are entirely equivalent, and thus after another period of time $t_0$ , we have the state 
\begin{equation}
    e^{-iH^1\cdot 2t_0}|A\rangle=e^{-iH^1t_0}e^{i\phi}|B\rangle=e^{i2\phi}|A\rangle
\end{equation}
and thus the system is periodic, up to a phase $2\phi$, with
period $2t_0$ . Thus we conclude that \textit{a mirror symmetric system must be periodic if it is to allow perfect state transfer.} This may be written most simply as
\begin{equation}
    |\langle A|U(2t_0)|A\rangle|=1
\end{equation}
for some $t_0 \in (0,\infty)$. Let us examine the general state of a periodic system with period $2t_0$ . We can write
\begin{equation}
    |\psi (2t_0)\rangle=\sum_ja_je^{-i2\lambda_jt_0}|\lambda_j\rangle=e^{i2\phi} \sum_ja_j|\lambda_j\rangle
\end{equation}
for eigenstates $|\lambda_j\rangle$ of $H^1$ with corresponding eigenvalues $\lambda_j$, where $a_j=|\langle A|\lambda_j\rangle|^2$. Hence, for all of the stationary states $|j\rangle$, we have the condition
\begin{equation}
    2\lambda_jt_0-2\phi=2k_j\pi
\end{equation}
where the $k_j$’s are integers. Eliminating $\phi$ between two of these, we get that
\begin{equation}
    (\lambda_j-\lambda_l)2t_0=2\pi(k_j-k_l)
\end{equation}
and eliminating $t_0$ between any two of these $(\lambda_{j'}\neq \lambda_{l'})$ gives
\begin{equation}
\label{eqn:rationality}
    \frac{\lambda_j-\lambda_l}{\lambda_{j'}-\lambda_{l'}}=\frac{k_j-k_l}{k_{j'}-k_{l'}}\in \mathbb{Q}
\end{equation}
where $\mathbb{Q}$ denotes the set of rational numbers. As the $k_j$’s are integers, this implies that the ratio is rational. Hence, a symmetric system capable of perfect state transfer must be periodic, which is equivalent to the requirement that the ratios of the differences of the eigenvalues are rational. Both conditions are completely equivalent.

\subsection{Bi-partite graphs}
For a bipartite graph, its vertices can be divided into two disjoint sets $U$ and $W$ such that every edge connects a vertex in $U$ to one in $W$. Vertex sets $U$ and $W$ are called the parts of the graph. This construction is one special kind of partition of a graph. For real Hamiltonians ($J_{ij}\in \mathbbm{R}$), this implies that $B_j=0$ \cite{ref:2}. And also imposes that the transfer phase $e^{i\phi}$ is $\pm 1$ if the transfer distance is even and $\pm i$ if the transfer distance is odd. Transfer distance is number of edges between two vertices between which perfect state transfer is performed. Such a bi-partite graph can be understood as marking of vertices (see section \ref{subsec:graphproperties}). let us mark the vertices belonging to $U$ positive and those belonging to $W$ as negative. Then our symmetry operator $S$ can be written as
\begin{equation}
S=\sum_{j\in U}|j\rangle \langle j|-\sum_{j\in W}|j\rangle \langle j|.
\end{equation}
This bicoupling for marked vertices also has a conservation property for the Hamiltonian $H^1$ that
\begin{equation}
    \{ S,H^1 \}=0
\end{equation}
(notice the anti-commutator instead of a commutator) which means that for any eigenvector $|\lambda_j\rangle$ of $H^1$ with $\lambda_j\neq 0$, $S|\lambda_j\rangle$ must also be an eigenvector of $H^1$, but with a negative eigenvalue $-\lambda_j$. Let us now assume that the initial vertex $A\in U$ partition. Then, it can be re-expressed as
\begin{equation}
    |A\rangle =\sum_{\substack{\lambda_j>0\\ \langle \lambda_j|A\rangle\neq 0}}\langle \lambda_j|A\rangle\left( |\lambda_j+S|\lambda_j\rangle \right).
\end{equation}
We assumed there are no zero eigenvalues except at most one zero eigenvector with non-zero overlap with $|A\rangle$ which needs to be accounted for; without loss of generality. This eigenvector must satisfy $S|\lambda_0\rangle=|\lambda_0\rangle$. The quantum evolution of this is
\begin{equation}
    e^{-iH^1t_0}|A\rangle=\sum_{\substack{\lambda_j>0\\ \langle \lambda_j|A\rangle\neq 0}}\langle \lambda_j|A\rangle\left(e^{-i\lambda_jt_0} |\lambda_j+e^{+i\lambda_jt_0}S|\lambda_j\rangle \right).
\end{equation}
If we now calculate the overlap with some vertex $C$, such that $C\in U$, $S|C\rangle=|C\rangle$, then
\begin{equation}
    \langle C|e^{-iH^1t_0}|A\rangle=\sum_{\substack{\lambda_j>0\\ \langle \lambda_j|A\rangle\neq 0}}2\langle C|\lambda_j\rangle \langle\lambda_j|A\rangle \cos(\lambda_jt_0),
\end{equation}
such that the amplitude is always real. Because $C\in U$ was a positive vertex, it must be at an even distance from $A$. Otherwise, if $C\in W$ is a negative vertex, then $S|C\rangle=-|C\rangle$ and we have
\begin{equation}
    \langle C|e^{-iH^1t_0}|A\rangle=\sum_{\substack{\lambda_j>0\\ \langle \lambda_j|A\rangle\neq 0}}2i\langle C|\lambda_j\rangle \langle\lambda_j|A\rangle \sin(\lambda_jt_0),
\end{equation}
such that the amplitude is always imaginary. Because $C\in U$ was a negative vertex, it must be at an odd distance from $A$. there are more peculiar implications of basic symmetries. For our relevance, these two symmetries are enough. We conclude that the real nature of the Hamiltonian plays a very important role in determining these effects of symmetries.

\subsection{Symmetry of balanced graph}
\label{subsec:balanced}
For definition of balanced graph, see subsection \ref{subsec:graphproperties}. Signing of graph edges can be realised by a transformation $\Theta$, the new graph is $\Sigma=(G,\sigma)$. We assume that $\sigma$ is balanced or anti-balanced signing of $G$. Then due to lemma 1 in \cite{ref:50}, we have that if a graph $G$ has perfect state transfer, then so does the signed graph $\Sigma = (G, \sigma)$. Suppose $G$ has perfect state transfer from vertex $A$ to $B$. If $\sigma$ is a balanced or antibalanced signing of $G$, then there is a diagonal $\pm 1$ matrix $\Theta$ for which $A(G^\sigma
) = \pm \Theta^{-1}A(G)\Theta$.
Thus, we have 
\begin{equation}
    \langle A| e^{-iA(G^\sigma)t} |B|\rangle= \langle A| \Theta^{-1}e^{\pm iA(G)t} \Theta |B\rangle=\pm \langle A| e^{\pm iA(G)t}|B\rangle.
\end{equation}
Therefore, $G^\sigma$ has perfect state transfer from $A$ to $B$.



Directed or oriented graphs can be defined in similar manner. A directed graph (or digraph) is a graph that is made up of a set of vertices connected by edges, where the edges have a direction associated with them. A directed graph is an ordered pair $G = (V, A)$ where
\begin{itemize}
    \item $V$ denotes the vertex set, and
    \item $E$ is now a set of ordered pairs of distinct vertices, called directed edges.
\end{itemize}



For oriented graphs, the direction of edges is modeled with a sign. Hence, the adjacency matrix is skew symmetric. We take the adjacency matrix of an oriented graph to be the matrix $A$ with rows and columns indexed by the vertices of the graph, and $A_{ij}$ equal to 1 if the the edge $\{i,j\}$ is oriented from $i$ to $j$, equal to $−1$ if the the edge $\{i,j\}$ is oriented from $j$ to $i$, and equal to zero if $i$ and $j$ are not adjacent. Consequently $A(G)$ is skew symmetric \cite{ref:19}. If $A(G)$ is a skew symmetric matrix, then $iA(G)$ is Hermitian and so
\begin{equation}
U(t) = \exp(−it(iA(G))) = \exp(tA(G))
\end{equation}
is the transition matrix of a continuous quantum walk (a quantum evolution with adjacency or Laplacian matrix over the vertex space of a graph is simply a continuous quantum walk). We note that $U(t)$ is real and orthogonal. Similarly, a general weighted graph can be obtained if weight $w(i,j)$, some real number, is assigned to each edge $(i,j)$ for all possible edges. Weighted chain for perfect state transfer in last chapter was an example of weighted graph.


\section{Impossibility of routing and need for a custom architecture}
\label{sec:routingimpossible}
Routing of an initial state means the freedom to be able to choose which vertex we wish to perfectly transfer the state at the initial given vertex. The idea of routing is eventually related to the complex nature of the Hamiltonian involved \cite{ref:2}. If for a given graph $G$, perfect state transfer is possible between $A$ and $B$, and the minimum time for which this is possible is $t_{AB}$, then we have
\begin{equation}
    e^{-iH^1t_{AB}}|A\rangle=e^{i\phi} |B\rangle.
\end{equation}
Furthermore, since the Hamiltonian is real, all the $\varphi_j$ are 0 or $\pi$ (see equation (\ref{eqn:phij})). Then for twice the time, $2t_{AB}$, we will have perfect revivals as
\begin{equation}
    e^{-iH^12t_{AB}}|A\rangle=e^{2i\phi} |A\rangle.
\end{equation}
which is periodic dynamics. Now, if perfect routing is possible, this means that we must have a time $t_{AC}<t_{AB}$ such that
\begin{equation}
    e^{-iH^1t_{AC}}|A\rangle=e^{i\phi'} |C\rangle
\end{equation}
and similar revivals as
\begin{equation}
    e^{-iH^12t_{AC}}|A\rangle=e^{2i\phi'} |A\rangle.
\end{equation}
These equations together give
\begin{equation}
    e^{-iH^1(2t_{AC}-t_{AB})}|B\rangle=e^{i2\phi'-\phi} |A\rangle
\end{equation}
But this is simply the perfect state transfer between $B$ and $A$ with time $|2t_{AC}-t_{AB}|<t_{AB}$, which is impossible by the initial assumption that $t_{AB}$ was the shortest time for perfect state transfer between $A$ and $B$. Therefore, the transfer from $A$ to $C$ cannot exist. As a result, if there is perfect state transfer to one site from a given site, there cannot be a perfect state transfer from this given site to any other sites. However, this can be tackled as shown in \cite{ref:16} by considering the dynamics of complex Hamiltonians.

This means that for a real Hamiltonian, we can have only pair of vertices where perfect state transfer is possible. If the network is very large, then this becomes less useful because given a vertex only one other vertex can be reached out with the transfer. This calls for the possibility of an architecture which allows more freedom for routing. In chapter \ref{chap:twohop}, we propose such an architecture with additional conditions to this idea such that routing is possible to any arbitrary site in the network for real Hamiltonians. Such an architecture will be very useful for the era of large scalable quantum processors where a state needs to be perfectly transferred to any given qubit in the processor with maximum fidelity. Moreover, in chapter \ref{chap:phyreal} we also show that our architecture works with the conventional quantum computing architecture which can be used to perfectly transfer or swap arbitrary states between any two given qubits in just two steps. Our scheme greatly reduces the circuit depth if the swap has to be performed between distant qubits using the universal quantum gates.


\section{Bounds on transfer rate in chains and beyond}
Some weaker bounds for transfer rate for spin chains are known. Transfer rate has been discussed in \cite{ref:17}\cite{ref:18}. Consider assigning a second state at site $A$ before the first state has been moved to site $B$, and impose the condition that 
the first state should still arrive at $B$ perfectly. After the existence of a transfer Hamiltonian is guaranteed, the necessary condition for the ability to insert a second quantum state into the spin network to the same initial input qubit at some time $t$ without disturbing the first quantum state is that
\begin{equation}
    \langle A|e^{-iH^1t}|A\rangle=0.
\end{equation}
This condition is necessary and sufficient for chains. However, for general networks which are not chains, more conditions are required. More generally, consider inserting many different states at different times, but the condition for the chain remains the same for all possible times. This will certainly not be the case for the dynamically changing many-excitation states in an arbitrary network. Yet this is a necessary condition. Therefore, given $p$ unique time intervals $t_i<t_0$ at which $\langle A|e^{-iH^1t_i}|A\rangle=0$, perfect state transfer can occur to a site $|B\rangle$ at a distance of $d$ in time $t_0$. With $p$ time intervals, one can have $p$ unique time intervals $t_i$ by imposing fixed intervals. For each transfer distance $m=1,2,...,d-1$,
\begin{equation}
    \langle B|H^1_m|A\rangle=0.
\end{equation}
This can be expressed as
\begin{equation}
    \sum_{j=1}^ne^{-i\varphi_j}\lambda^m_ja_j=0
\end{equation}
with the same $a_j=|\langle A|\lambda_j\rangle|^2$ which can be re-expressed again after removing the degeneracies ($n$ reduces to $n'$ number of unique eigenvalues) from the system in the linear equation form as
\begin{equation}
    \left( \sum_{m=0}^{d-1}\sum_{j=1}^{n'}e^{-i\varphi_j}\lambda^m_ja_j|m\rangle \langle j| \right)\left( \sum_{j=1}^{n'}a_j|j\rangle \right)
\end{equation}
Each of the $d-1$ rows is linearly independent. Linearly expressed, the normalization condition says that
\begin{equation}
    \left( \sum_{j=1}^{n'}\langle j| \right)\left( \sum_{j=1}^{n'}a_j|j\rangle \right)
\end{equation}
We can now add conditions corresponding to $\langle A|e^{-iH^1t_i}|A\rangle=0$ and further divide these into real and imaginary components. The real part is
\begin{equation}
    \left( \sum_{i=1}^{p}\sum_{j=1}^{n'}\cos(\lambda_jt_i)|i\rangle \langle j| \right)\left( \sum_{j=1}^{n'}a_j|j\rangle \right)
\end{equation}
and the imaginary component gives
\begin{equation}
    \left( \sum_{i=1}^{p}\sum_{j=1}^{n'}\sin(\lambda_jt_i)|i\rangle \langle j| \right)\left( \sum_{j=1}^{n'}a_j|j\rangle \right)
\end{equation}
Given that all these times ti are less than $t_0$, the half period
of the system, all of these rows must be linearly independent from each other. (Since we are assuming the Hamiltonian is real and performs perfect transfer, the system is periodic with a period $2t_0$.) Hence, if a suitable set of an is to possibly exist, it must be the case that
\begin{equation}
    2p+d\leq n'\leq n
\end{equation}
Ideally, we want the maximum transfer distance, which would be $n − 1$ (a chain), imposing that $p = 0$. The only way to increase the perfect transfer rate is to reduce the transfer distance. However, you cannot also lower the state transfer time (as you would expect by shortening the transfer distance). This is because the Margolus-Levitin theorem \cite{ref:17} imposes a minimum time for evolving between two orthogonal states, such as a $|1_A\rangle$ as an input state and the $|0_A\rangle$ required for the next input. Hence the transfer time is bounded from below by $(p+1)\pi/(4\sum_jJ_{1j})$. In some sense, the “standard” perfect state transfer chains saturate the bound for a chain of $n$ qubits, any state $|j\rangle$ transfers a distance $D = n + 1 − 2j$, but there are $j − 1$ distinct times $t_i$ such that $\langle j|e^{-iH^1t_i}|j \rangle=0$. Unfortunately, however, these times are not equally spaced, so they are not useful for achieving a high rate of transfer.


\section{Limitations over the uniformly coupled chain}
\label{sec:chainlimitation}
It is desirable to maximise the distance over which communication is possible for a fixed number of qubits. The simplest and optimal arrangement, in this case, is just a linear chain of $n$ qubits, where $A$ and $B$ are the qubits at opposite ends of the chain.

Let us start with the XY chain of qubits, with uniform couplings $J_{i,i+1} = 1$ for all $1 ≤ i ≤ n − 1$. The
Hamiltonian reads
\begin{equation}
    H=\frac{1}{2}\sum_{i=1}^{n-1}\sigma_i^x\sigma_{i+1}^x+\sigma_i^y\sigma_{i+1}^y
\end{equation}
In this case, one can compute $f_{AB}(t)$ explicitly by diagonalizing the Hamiltonian or the corresponding adjacency
matrix. The eigenstates and the corresponding eigenvalues are given by
\begin{equation}
    |\Tilde{k}\rangle=\sqrt{\frac{2}{n+1}}\sum_{j=1}^n\sin \left(\frac{\pi kj}{n+1} \right)
\end{equation}
and 
\begin{equation}
    \lambda_k=E_k=-2\cos \frac{k\pi}{n+1}
\end{equation}
with $k = 1,...,n$. Thus
\begin{equation}
    f_{AB}(t)=\frac{2}{n+1}\sum_{k=1}^n\sin \left( \frac{\pi k}{n+1}\right)\sin \left( \frac{\pi kn}{n+1}\right)e^{-iE_kt}
\end{equation}
Perfect state transfer from one end of the chain to another is possible for $n = 2$ and $n = 3$, where we find that $f_{AB} (t) = −i\sin(t)$ and $f_{AB} (t) = -\sin^2(t/\sqrt{2})$ respectively. Hence, for perfect state transfer, that is, to have $|f_{AB}(t_0)|=1$, we have
\begin{equation}
\label{eqn:timeforPST}
    t_0=
    \begin{cases}
    & \pi/2,\text{ for n=2},\\
    & \pi/\sqrt{2}, \text{ for n=3},\\
    \end{cases}
\end{equation}
in the units of energy inverse. 
We have shown that perfect state transfer is possible for chains containing 2 or 3 qubits. It can be now shown that it is not possible to get perfect state transfer for $n \geq 4$. This work was originally done in \cite{ref:5}. A chain is symmetric about its centre. Hence the rationality for eigenvalues condition equation (\ref{eqn:rationality}) for perfect state transfer applies for all longer chains. If we pick specific values $l=2,$ $j=n-1,$ $l'=1$, and $j'=n$, then using the expression for eigenvalues for the chain, this condition becomes
\begin{equation}
    \frac{\cos \frac{2\pi}{n+1}}{\cos \frac{\pi}{n+1}} \in \mathbb{Q}
\end{equation}
to hold for the perfect state transfer. The concept of algebraic numbers can be used to find the value of $n$ for which the above condition holds. An algebraic number $x$ is a complex number that satisfies a polynomial equation of the form
\begin{equation}
    a_0x^m+a_1x^{m-1}+...+a_{m-1}x+a_m=0,
\end{equation}
with integral coefficients $a_i$. Every algebraic number $\alpha$ satisfies a unique polynomial equation of least degree. The degree of this polynomial is called the degree of $\alpha$. If $\alpha$ satisfies a polynomial of degree $m$, then it i called an \textit{algebraic integer} of degree $\alpha$. An algebraic integer of degree $m$ is also number of degree $m$. Rational numbers are algebraic numbers with degree $1$, and numbers with degree $\geq 2$ are necessarily irrational. If $n>1$, and gcd$(k,n+1)=1$ then $\cos [\pi k/(n+1)]$ is an algebraic integer of degree $\phi(2(n+1))/2$, where $\phi$ is the Euler phi function and we have that $\phi(2(n+1))/2\geq 3$ for $n\geq 6$. See \textit{Irrational numbers} by Lehmer (Mathematical Association of America, 1956).

It we assume that the expression of the form 
\begin{equation}
        \frac{\cos 2\theta}{\cos \theta}=\frac{p}{q} \in \mathbb{Q} \quad p,q\in \mathbbm{Z}
\end{equation}
(with $\theta=\pi/(n+1)$ is an algebraic number of degree $\geq 3$) is rational then using trigonometric identity $\cos 2\theta =2\cos^2\theta -1$ we have
\begin{equation}
    \cos^2 \theta -\frac{p}{2q}\cos \theta -\frac{1}{2}=0
\end{equation}
which has rational coefficients. This means that $\cos \theta$ is algebraic with degree $\leq 2$. Hence, this is a contradiction and $\cos 2\theta/\cos\theta$ must be irrational. Therefore, this strictly proves that for $n\geq 6$, perfect state transfer is impossible (between the end vertices of the chain) as $\deg (n)\geq 3$. Furthermore, for $n=4$ and $n=5$ similar calculations show that
\begin{equation}
    \frac{\cos 2\theta}{\cos \theta}=\frac{p}{q} \notin \mathbb{Q}.
\end{equation}
We, therefore, in conclusion, have that it is impossible to perform perfect state transfer in unmodulated chains of constant coupling for number of nodes $n\geq 4$.

However, modulated chains can allow perfect transfer over arbitrary long distances as we shall see in the next section. This has been explored in \cite{ref:5} as the column method for pseudo-hyperspin. The result is to select the coupling $J_i=\sqrt{i(n-i)}$ and the chain magically supports perfect state transfer. But such modulations over a considerable length of chain is very hard to engineer experimentally.


\section{Perfect State Transfer in long and weighted chains}
\label{sec:longchains}
The workaround to enable perfect state transfer for chains with length $n\geq 4$, the idea of projecting a \textit{hypercube} (see chapter \ref{chap:twohop} for detailed study on hypercubes) onto a spin chain was originally studied in \cite{ref:5}. The hypercube \textbf{Q}$_k$ resulting from the $k-$fold Cartesian product (see section \ref{sec:cartesian}) of one-link graph has the property that it can seen as arrangement of its vertices as columns such that there are no edges between the vertices within any column and edges only join vertices in different columns. And furthermore, each vertex in column $i$ must have the same number of incoming (from column $i-1$) and outgoing (to column $i+1$) edges as all other vertices in that column (it is simply due to the property of hypercubes that each vertex has the same number of adjacent vertices). 

Let \textbf{Q}$_k$ be arranged in $n_c$ columns, call the graph as $G$. The size of each column is $c_i:=|G_i|=^{n_c-1}$C$_{i-1}$ and label the vertices in each column as $G_{ij}$ with $j=\{1,2,...,c_i\}$. Start with a vertex $A$, then the $i$th column is $i-1$ edges away from the vertex $A$. From each column there are edges going backward to the previous column and edges going forward to the next column (except for the end columns). These are denoted as
\begin{equation}
    \begin{split}
        & C^{for}_i:=\{ (G_{ij},k):j\in \{1,...,c_i\},k \in \{1,...,f_i\} \}\\
        & C^{back}_i:=\{ (G_{ij},k):j\in \{1,...,c_i\},k \in \{1,...,b_i\} \}
    \end{split}
\end{equation}
where $f_i$ and $b_i$ denote the number of forward and backward edges, respectively, for the $i$th column. If all the edges are to have ends, then $|C^{for}_i|=|C^{back}_i|$. Since there is only one vertex (qubit) in the first column ($c_1=1$), each vertex in the second column has only a single edge going backward, implying $b_1=1$. Starting from this boundary condition, and that $f_i$ and $b_i$ must be integers for all $1\leq i \leq n_c$, we have the condition that
\begin{equation}
    c_if_i=c_{i+1}b_{i+1}
\end{equation}
which implies
\begin{equation}
    \frac{f_i}{b_{i+1}}=\frac{n_c-i}{i}.
\end{equation}
The solution for this is to choose $f_i=n_c-i$ and $b_i=i-1$. Therefore, we end up with a graph such that for every pair of numbers ($i,j$), $G_{ij}$ ic onnected with $n_c-i$ columns in $G_{i+1}$ and each vertex in $G_{i+1}$ is connected with $i-1$ vertices in $G_i$.

We define the vectors that span the column space $\mathcal{H}_c$ as,
\begin{equation}
    |\col i \rangle:=\frac{1}{\sqrt{c_i}}\sum_{j=1}^{c_i}|G_{}ij\rangle.
\end{equation}
Class of networks with this column representation have the special property that throughout the quantum evolution with the adjacency matrix the instantaneous state always remains in the column space $\mathcal{H}_c$. Thus, it can be seen as the problem of perfect state transfer from $G_{11}$ to $G_{n_c1}$, for instance. Also note that the \textit{antipodal} vertices of a hypercube constructed from one-link and two-link graphs admit perfect state transfer (see next section). The matrix elements of the adjacency matrix of $G$, restricted to the column space are given as 
\begin{equation}
\label{eqn:Jj}
    J_i:=\langle \col i|H_G|\col i+1\rangle=\sqrt{i(i-n_c)}.
\end{equation}
The matrix form is
\begin{equation}
    J=\begin{bmatrix}
        0 &J_1 &0 &0 &... &0\\
        J_1 &0 &J_2 &0 &... &0\\
        0 &J_2 &0 &J_3 &... &0\\
        0 &0 &J_3 &0 &... &0\\
        \vdots &\vdots &\vdots &\vdots &\ddots &J_{n_c-1}\\
        0 &0 &0 &0 &J_{n_c-1} &0
      \end{bmatrix}
\end{equation}
This is because
\begin{equation}
    \langle \col i|H_G|\col i+1\rangle=\frac{1}{\sqrt{c_ic_{i+1}}}\sum_{j=1}^{c_i}\sum_{j'=1}^{c_{i+1}}\langle G_{ij}|H_G|G_{i+1j'}\rangle
\end{equation}
\begin{equation}
    =\frac{1}{\sqrt{c_ic_{i+1}}}c_i(n_c-1)=\sqrt{i(n_c-i)}.
\end{equation}
Clearly, this is identical to the matrix form of the $XY$-model chain with just specially engineered coupling strengths $\{ J_i\}$ sch that the Hamiltonian is
\begin{equation}
\label{eqn:XYweighted}
    H_G \equiv H_{XY}=\frac{1}{2}\sum_{j=1}^{n_c-1}J_j\left( \sigma^x_j\sigma^x_{j+1}+ \sigma^y_j\sigma^y_{j+1} \right)
\end{equation}
where $J_j$ is given by equation (\ref{eqn:Jj}). Such a chain must allow perfect state transfer over any length $n_c$ (where we redefine $|A\rangle:=|\col 1\rangle$ and $|B\rangle:=|\col n_c \rangle$) because the hypercube does (see next section). Similar weighted chain can be realised with the Heisenberg model with local magnetic fields as
\begin{equation}
    H_G \equiv H_{Hei}=\frac{1}{2}\sum_{j=1}^{n_c-1}J_j\Vec{\sigma}_j\cdot \Vec{\sigma_{j+1}}+\sum_{j=1}^{n_c}B_j\sigma^z_j
\end{equation}
with $B_j = \frac{1}{2}(J_{j-1}+J_j)-\frac{1}{2(n_c-2)}\sum_{k=1}^{n_c-1}J_k$ which has been specially chosen to cancel the diagonal elements to bring the $XY$ and Heisenberg model on equal grounds. This model is now perfectly equivalent to the transfer dynamics of a weighted chain with $n_c$ vertices (qubits). With this idea of projecting a hypercube to a spin chain, we see that the chain is enabled for perfect state transfer with the difference being that it is now modulated with special coupling strengths which gives rise to a weighted chain graph.


\section{Perfect state transfer over greater distances}
\label{sec:hypercubetransfer}
Perfect state transfer over arbitrary distances is impossible
for a simple unmodulated spin chain (limited to $n=2$ and $n=3$ only!). Clearly it is desirable to find a class of graphs that allow state transfer over larger distances. One approach to achieved this apart from modulation of spin chains is to construct larger arbitrary graphs using the graph products of small blocks of $n=2$ or $n=3$ which serve as the fundamental building blocks for such construction. One well explored construction is through the Cartesian product of linear chains proposed in \cite{ref:5}. We examine the $d$-fold Cartesian product of one-link (two-vertex) and two-link (three-vertex) chain $G$. We denote this by $G^d:=\square^dG$ where the square denotes the Cartesian product of $G$ with itself. See section \ref{sec:cartesian} for details and construction of Cartesian product of graphs. Following the binary and ternary representation for the vertex labeling as in chapter \ref{chap:twohop}, consider two antipodal vertices $A(0,0,...,0)$ and $B(1,1,...,1)$ (labels of length $d$ each) for for one-link. Similarly, for two-link hypercube we have the antipodal points as $A(0,0,...,0)$ and $B(2,2,...,2)$ respectively. This can be proved that for \textit{any dimension} $d$ $|f_{AB}(t)|=1$ for $t=t_0=\pi/2$ and $t=t_0=\pi/\sqrt{2}$ respectively! This means that over this large hypercube the perfect transfer takes place in the same time as the one-link and two-link chain respectively. Hence, $t_0$ is the perfect state transfer time for transfer between antipodal vertices $A$ and $B$ for $G^d$ also.

The first sign of perfect state transfer for hypercubes can be seen due to equation (\ref{eqn:rationality}). For hypercubes from one-link and two-link seed graph $G$, the ratios of differences of all possible eigenvalues are rational, which permits perfect state transfer. Furthermore, it can be proved strongly by construction. As already established, the Hamiltonian dynamics of $XY$ interaction Hamiltonian is identical to the dynamics of the adjacency matrix in the first excitation subspace. This holds equally for the Cartesian product of $G$, by construction. Hence,
\begin{equation}
    H=A(G^d)=\sum_{j=0}^{d-1}\mathbbm{I}^{\otimes j}\otimes A(G)\otimes \mathbbm{I}^{\otimes d-j-1}
\end{equation}
and
\begin{equation}
    e^{-iHt}=(e^{-iA(G)t})^{\otimes d}
\end{equation}
Thus, if we evolve the system for time $t_0$, we get perfect state transfer along each dimension. Each term in the tensor product applies to a different element of the basis. We therefore achieve perfect state transfer between $A$ and $B$ as well as between any qubit and its mirror vertex qubit. The fidelity of the state transfer is simply the $d$th power of the fidelity for the original chain:
\begin{equation}
    F_{G^d}(t)=[F_G]^d(t)=
    \begin{cases}
    & \sin^d(t), \quad \text{for }2^d \text{ vertices}\\
    & \sin^{2d}(t/\sqrt{2}), \quad \text{for }3^d \text{ vertices}.
    \end{cases}
\end{equation}
This formalism also extends over to the Heisenberg couping Hamiltonian. This is because, in the case of a two-qubit chain, the Hamiltonian in the single excitation subspace is represented by a matrix with identical diagonal elements, and hence is the same as the Hamiltonian of an $XY$ model up to a constant energy shift, which just adds a global phase factor. Hence, the same hypercube transfer dynamics holds true for Heisenberg scheme. Thus, any quantum state can be perfectly transferred between the two antipodes of the one-link and two-link hypercubes of any dimensions in constant time.

\section{Contribution from this thesis on hypercubes and beyond}
The above discussion motivates the idea of finding other graph products or class of graphs which might support perfect transfer. Cartesian product is the simplest such product which has been explored. Other products are not so physically relevant as the Cartesian product. Basically, this defines a growing architecture scheme for connectivity in a quantum processor. However, for large $d$, the difference between $2^{d}$ and $2^{d+1}$ is very large and it makes little sense experimentally to add this many qubits in a system to allow for perfect state transfer. Moreover, for large $d$, the cost of adding so many edges (physically establishing precisely the same coupling strength) in the system just to establish perfect state transfer between only pair of antipodal nodes, is too high. In this thesis work, we propose our scheme which is based on the hypercube result for long distance transfer which primarily resolves these two challenges to the hypercube architecture. Our scheme in chapter \ref{chap:twohop} enables perfect state transfer from all-to-all nodes for \textit{arbitrary} number of qubits! Hennce, the problem of routing of states can be resolved. It also features that one qubit can be added each time individually in our architecture. This also complies with the current experimental challenges for the realization of quantum computing where a small number of qubits can be added into the processor for scalability.

%% file: Chapters/Graph_Theory.tex
\label{chap:Graph}
In this chapter we briefly discuss some fundamental concepts related to graphs, and matrices associated with graphs. We primarily focus on finite, simple graphs: those without loops or multiple edges. Further details can be found in \cite{west2001introduction}.  

\section{Graphs, Adjacency matrices and graph Laplacian matrices}
A graph is an ordered pair $G = (V, E)$, where
\begin{itemize}
    \item $V$ denotes the set of vertices \{$v_i$\}(also called nodes or points), and
    \item $E \subseteq \{\{x, y\} | (x, y) \in V\times V,  x \neq y\}$ denotes the set of edges (also called links or lines), which are unordered pairs of vertices (i.e., an edge is associated with two distinct vertices)
\end{itemize}


The adjacency matrix associated with a graph $G$ is denoted by $A(G)=[a_{ij}]$,where
        \[
        a_{ij}=
        \begin{cases}
        1, \text{\quad if }(i,j)\in E\\
        0, \quad \text{otherwise}
        \end{cases}
        \]
Let $G$ be a graph on $n$ vertices, that is, $|V|=n.$ Then obviously, $A(G)$ is a symmetric matrix of order $n\times n.$ Let $v_i\in V.$ Then the degree of $v_i$ is defined as $\mbox{deg}(v_i)=\sum_{j=1}^n a_{ij}.$


The degree matrix $D(G)$ of $G$ is a diagonal matrix where \begin{equation}
    D(G)\equiv [d_{ii}]=\deg (v_i)
\end{equation}
The graph Laplacian matrix $L(G)$ associated with the  graph $G$ is defined as
\begin{equation}
    L(G)=D(G)-A(G)
\end{equation}
which is equivalent to saying
\[
        L(G)\equiv[l_{ij}]=
        \begin{cases}
        \text{deg}(v_i), \quad  \text{if }i=j\\
        -1, \quad \text{if }i\neq j, \text{and } (i, j)\in E \\
        0, \quad \text{otherwise}.
        \end{cases}
\]

The \textit{signless} Laplacian matrix corresponding to $G$ is defined by 
\begin{equation}
    L^+(G)=D(G)+A(G).
\end{equation}

\section{More general properties of graphs}
\label{subsec:graphproperties}
Following are some more general definitions of the graphs which may play a role in PST for specific class of graphs. 
\begin{itemize}
\item \textbf{Signed graph:} A signed graph is an ordered tuple $G = (V,E,\sigma)$ where $V$ denotes the set of nodes, $E \subseteq V \times V$, the edge set, and $\sigma : E \longrightarrow \{+,−\}$ is called the signature function. This another degree added to the definition of a graph.  An obvious way to construct a signed graph from a marked graph is be defining the sign of an edge of the marked graph as the product of signs of its adjacent vertices. Thus, the sign of an edge is the product of its signs of vertices it connects.
    \item \textbf{Marked graph:} We can assign a marking $\{\pm\}$ to the nodes or vertices along with naming them. This adds more degree of freedom to the graph which can be captured by additional functions. A graph is called a marked graph if every node of the graph is marked by either a positive or negative sign. Thus a marked graph is a tuple $G = (V,E,\mu)$ where $V$ is the node set, $E$ the edge set and $\mu : V → {+,−}$ is called the marking function. There are various possible marking schemes. Let us look at the following two conventional ways of marking the vertices.
    \begin{itemize}
     
        \item \textbf{Canonical marking scheme:}  Defined from a signed (see next definition) graph $G = (V,E,\sigma)$ by defining the marking of a node $v \in V$ as
        \begin{equation}
            \mu(v)=\prod _{e\in E_v}\sigma(e)
        \end{equation}
        where $E_v$ is the set of signed edges adjacent at $v$.
        \item \textbf{Plurality marking scheme:} We define plurality marking of a node $v$ of a signed graph $G = (V,E,\sigma)$ as
        \[
        \mu(v)=
        \begin{cases}
        +, \quad \text{if }\max \{d^+(v),d^-(v)\}=d^+\\
        -, \quad \text{Otherwise}
        \end{cases}
        \]
        Hence a node is negatively marked in plurality marking scheme only when $d^−(v) > d^+ (v)$.
    \end{itemize}
   
    \item \textbf{Balanced graph:}  A signed network is balanced if and only if all its cycles are balanced. A signed cycle is called balanced if the number of negative edges in it is even. In other words, a graph is balanced if all its cycles (vacuum-loops) or cliques are balanced. More rigorously it can be stated as follows. Let $G$ be a graph with vertices $v_1 ,..,v_n$ and cliques $C_1 ,..,C_m$ . The (0,1) matrix $A = (a_{ij} )$ where $a_{ij}$ is 1 iff vertex $v_i$ belongs to clique $C_j$ is called a clique matrix of $G$. A (0,1)-matrix is balanced if it does not contain the vertex-edge incidence matrix of an odd-cycle as a submatrix (that is, it contains no square submatrix of odd order with exactly two 1s per row and per column). A graph is balanced if its clique matrix is balanced.
    \item \textbf{Regularity:} A regular graph is a graph where each vertex has the same number of neighbors; i.e. every vertex has the same degree or valency. A regular directed graph must also satisfy the stronger condition that the indegree and outdegree of each vertex are equal to each other. A regular graph with vertices of degree $k$ is called a $k$‑regular graph or regular graph of degree $k$. Also, from the handshaking lemma, a regular graph of odd degree will contain an even number of vertices. A signed regular graph can be defined in the sense that signed regularity (say, $d^+(v_i)-d^-(v_i)=d$, constant) for all vertices $v_i\in V(G)$.
\end{itemize}
Therefore, a graph generally is a 4-tuple $G(V,E,\sigma,\mu)$. To take products of two graphs we first start by defining a graph with signed edges or marked vertices and find the other using a scheme we wish to follow in accordance with the graph operation involved. Then, the new edge signing in the new product graph can be found using the same scheme(s).

\section{Product Graphs}
Two graphs can be operated with a defined operation that gives another resultant graph. A graph product is a binary operation on graphs. Specifically, it is an operation that takes two graphs $G_1$ and $G_2$ and produces a graph $H$ with the following properties:
\begin{itemize}
    \item The vertex set of $H$ is the Cartesian product $V(G_1) \times V(G_2)$, where $V(G_1)$ and $V(G_2)$ are the vertex sets of $G_1$ and $G_2$, respectively.
    \item Two vertices $(u_1, u_2)$ and $(v_1, v_2)$ of $H$ are connected by an edge if and only if the vertices $u_1$, $u_2$, $v_1$, and $v_2$ satisfy a condition that takes into account the edges of $G_1$ and $G_2$. The graph products differ in exactly which this condition is.
\end{itemize}
Graph product is a very important operation for this work as it defines a new 'larger' graph from the initial graphs. This is helpful in describing a growing network which multiplies according to some defined graph product rule. In this thesis, we are specifically concerned with Wedge product, Cartesian product and Corona product of graphs.

\subsection{Wedge product of graphs}
\label{sec:wedge}
This section follows the wedge product as proposed in \cite{ref:15} and describes the general action of coupling Hamiltonians in different excitation spaces 
\begin{definition}
We define the wedge product $\wedge^kG$ of a graph $G$ to be the graph with vertex set $V(\wedge ^kG):=\{ (v_0,v_1,...,v_{k-1})|v_j\in V(G), v_{k-1}>v_{k-2}>...>v_0 \}$. We write vertices of $\wedge^k G$ as $v_0\wedge v_1 \wedge ...\wedge v_{k-1}$. We connect two vertices $v_0\wedge v_1 \wedge ...\wedge v_{k-1}$ and $w_0\wedge w_1 \wedge ...\wedge w_{k-1}$ in $\wedge^kG$ with an edge if there is a permutation $\pi\in S_k$ ($S_k$ is a permutation group on $k$ distinguish entities) such that $w_j=v_{\pi(j)}$ for all $j=0,1,...,k-1$ except at one place $j=\pi(l)$ where $(v_l,w_{\pi(l)})\in E(G)$ is an edge in $G$. The Hilbert space $\mathcal{H}_{\wedge^kG}$ of the graph $\wedge^kG$ is isomorphic to $\wedge^k(\mathcal{H}_G)$. Exterior vector space $\wedge^kG$ is spanned by vectors $|v_0\wedge v_1\wedge ...\wedge v_{k-1}\rangle$ where no two vectors $|v_j\rangle,|v_k\rangle,\forall j\neq k$, are the same.
\end{definition}
Wedge product for corresponding Hilbert space $\wedge ^kG$ is defined as $\wedge:\mathcal{H}_G\times \mathcal{H}_G\times...\times \mathcal{H}_G\longrightarrow
\bigotimes_{l=0}^{k-1}\mathcal{H}_G$ with action on the vectors as
\begin{equation}
    |v_0\wedge v_1 \wedge ...\wedge v_{k-1}\rangle:=\frac{1}{k!}\sum_{\pi\in S_k}\epsilon(\pi)|v_{\pi(0)},v_{\pi(1)},...,v_{\pi(k-1)}\rangle
\end{equation}
where $\epsilon(\pi)$ is the sign of the permutation $\pi$. This defines the basis for $\wedge^kG$ as $\{ |v_0\wedge v_1 \wedge ... \wedge v_{k-1} \rangle \}$ with $v_j\in V(G)$ and $v_{k-1}>v_{k-2}>...>v_0$. Furthermore, $\Gamma^k$ (defined in section \ref{sec:firstexcitation}) and $\wedge^k$ have the same dimension and they are isomorphic vector spaces over $\mathbb{C}$. The correspondence can be assigned by identifying the state $|1_{v_0},1_{v_1},...,1_{v_{k-1}} \rangle \in \Gamma^k$ which has a 1 at positions or vertices $v_{k-1}>v_{k-2}>...>v_0$ and zeros elsewhere, with the basis vector $|v_0\wedge v_1 \wedge ... \wedge v_{k-1} \rangle  \in \wedge^k\mathcal{H}_G$.

We show how the structure of $\wedge^G$ is captured by the adjacency matrix for for $\wedge^G$. And then show its action is identical to $H_{XY}$ in that excitation space. Let $\mathcal{B}(\mathcal{H}_G)$ be the space of all bound operators on $\mathcal{H}_G$. And let $M\in \mathcal{B}(\mathcal{H}_G)$ be a linear operator from $\mathcal{H}_G$ to $\mathcal{H}_G$. Define the operation $\Delta^k:\mathcal{B}(\mathcal{H}_G) \longrightarrow \bigotimes_{j=0}^{k-1}\mathcal{B}(\mathcal{H}_G)$ as
\begin{equation}
    \Delta^k(M):=\sum_{j=0}^{k-1}I_{01...j-1}\otimes M_j\otimes I_{j+1...k-1}.
\end{equation}
Dimension of $\bigotimes^k\mathcal{H}_G$ is greater than that of $\wedge^k \mathcal{H}_G$. We define the projection $\Alt: \bigotimes  ^k\mathcal{H}_G\longrightarrow \wedge^k \mathcal{H}_G$ by
\begin{equation}
    \Alt | \phi_0,\phi_1,...,\phi_{k-1} \rangle:=\frac{1}{k!}\sum_{\pi\in S_k}\epsilon(\pi)\pi\left[ | \phi_0,\phi_1,...,\phi_{k-1} \rangle \right]
\end{equation}
\begin{equation}
    =\frac{1}{k!}\sum_{\pi\in S_k}\epsilon(\pi)| \phi_{\pi(0)},\phi_{\pi(1)},...,\phi_{\pi(k-1)} \rangle
\end{equation}
where the action of the symmetric group $S_k$ on the basis kets is evident and $\Alt ^\dagger =\Alt$ (because it is a projection). The analogous adjacency matrix (signed version) for $\wedge^G$ is generally  defined by
\begin{equation}
    C(\wedge^kG)=\Alt \Delta^k[A(G)]\Alt
\end{equation}
which generally contained negative entries also. It can be calculated that $\langle v_0\wedge v_1\wedge ... \wedge v_{k-1}|C(\wedge^kG)|v_0\wedge v_1\wedge ... \wedge v_{k-1}\rangle=\pm 1$ iff $v_j=w_j$ for all $j$ except at exactly one place $j=l$ where $(v_j,w_l)\in E(G)$. All other entries are exactly zero. The unsigned adjacency matrix $A(\wedge^kG)$ is the matrix obtained after replacing all instances of $-1$ by $+1$ in the obtained matrix $C(\wedge^kG)$. Using the spectral decomposition of $A(G)$, the spectral properties of $C(\wedge^kG)$ can be obtained as in \cite{ref:15}.

As evident from equation (\ref{eqn:XYmodel}), the Hamiltonian $H_{XY}$ action on $|\phi\rangle=|1_v,1_{v_1},...,1_{v_{k-1}}\rangle \in \Gamma^k$ is to move the 1 at position $i$ to $j$ if and only if there is no 1 in the $j$ place. In this way, we can say that the Hamiltonian $H_{XY}$ maps the state $|\phi\rangle$ to an equal superposition if all states which are identical to $|\phi\rangle$ at all indices except at one place. So, a 1 at a given place has been moved along an edge $e \in E(G)$ as long as there is no 1 at the endpoint of $e$. This is identical to the action of $A(\wedge^kG)$ on $\Gamma^k$. Similarly, an observation can be made for Laplacian matrix in $\Gamma^k$.

The action of $H_{XY}$ (and respectively $H_{Hei}$), when restricted to $\Gamma^k$ is the same as that of the adjacency matrix (and respectively Laplacian) of $\wedge^kG$.

\subsection{Cartesian product of graphs}
\label{sec:cartesian}
\begin{definition}
The Cartesian product of two graphs $G:=\{V(G),E(G)\}$ and $H:=\{V(H),E(H)\}$ is a graph $G \times H$ whose vertex is a set $V(G)\times V(H)$ and two of its vertices $(g,h)$ and $(g',h')$ are adjacent iff one of the following conditions hold
\begin{itemize}
    \item $g=g'$ and $\{h,h'\} \in E(H)$
    \item $h=h'$ and $\{g,g'\} \in E(G)$.
\end{itemize}
\end{definition}
Furthermore, if $|k\rangle$ is an eigenvector of $A(G)$ with corresponding eigenvalue $E_k$ and $|l\rangle$ is an eigenvector of $A(H)$ with corresponding eigenvalue $E_l$, then $|k\rangle \otimes |l\rangle$ is an eigenvector of $A(G\times H)$ with corresponding eigenvalue $E_k+E_l$. Here, $A(\cdot)$ is the adjacency matrix. This happens due to the underlying construction
\begin{equation}
    A(G\times H) =A(G) \otimes \mathbb{I}_{V(H)}+\mathbb{I}_{V(G)}\otimes A(H).
\end{equation}
This is exactly as forming the composite system out of two sub-systems in quantum theory. All the same construction applies.
\begin{figure}[hbtp]
    \centering
    \includegraphics[scale=0.35]{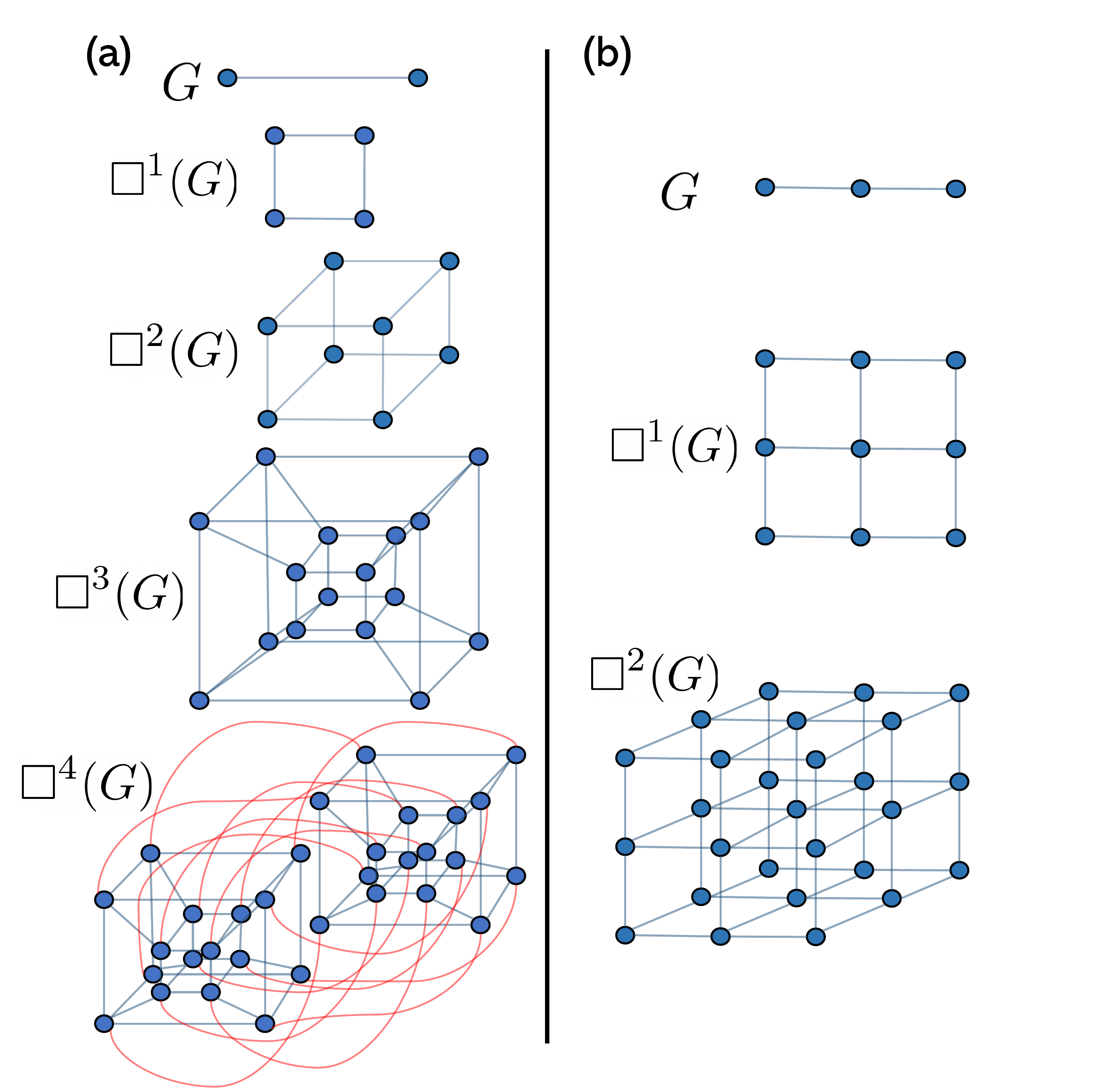}
    \caption{Part (a) shows the first four Cartesian product of $K_2$ or $Q_1$ with itself. These are the first five hypercubes. The new edges formed due to the fourth Cartesian product $\square^4 (G)$ is shown in red for clarity. Part (b) shows the first two Cartesian product of two-link graph $G$ with itself.}
    \label{fig:cartesian}
\end{figure}

\subsection{Corona product of graphs}
\label{sec:corona}
We state a relatively new and special kind of product called the Corona product. Corona product of graphs was introduced by Frucht and Harary in 1970 \cite{ref:6}\cite{ref:8}\cite{ref:9}. Given two unsigned and unmarked graphs $G$ and $H$, the corona product of $G$ and $H$ is a graph, we denote it by $G \circ H$, which is constructed by taking $n$ instances of $H$ and each such $H$ gets connected to each node of $G$, where $n$ is the number of nodes of $G$. Starting with a connected simple graph $G$, we define corona graphs which are obtained by taking corona product of G with itself iteratively. In this case, $G$ is called the seed graph for the corona graphs. Using a seed is exactly the same approach as the chain building blocks were considered previously. 

Now we state the definition of the signed Corona product of two graphs.
\begin{definition}
Let $G_1 = (V_1 ,E_1 ,\sigma_1 ,\mu_1 )$ and $G_1 = (V_2 ,E_2 ,\sigma_2 ,\mu_2 )$ be signed graphs on $n$ and $k$ nodes respectively. Then corona product $G_1 \circ G_2$ of $G_1,G_2$ is a signed graph by taking
one copy of $G_1$ and $n$ copies of $G_2$, and then forming a signed edge from $i$th node of $G_1$ to
every node of the $i$th copy of $G_2$ for all $i$. The sign of the new edge between $i$th node of $G_1$ , say $u$ and $j$th node in the $i$th copy of $G_2$ , say $v$ is given by $\mu_1 (u)\mu_ 2 (v)$ where $µ$ is a marking scheme defined by $\sigma_i ,i = 1,2$.
\end{definition}

For instance, the corona product $G_1 \circ G_2$ of signed graphs $G_1$ and $G_2$ is shown in figure \ref{fig:corona1}. Note that canonical and plurality marking are same for the graph $G_2$ . For $G_1$ the marking of the nodes 1,3 are same for canonical and plurality markings, whereas the canonical and plurality markings of node 2 are − and + respectively. Thus the choice of the marking function produce different corona product graphs.
\begin{figure}[hbtp]
    \centering
    \includegraphics[scale=0.35]{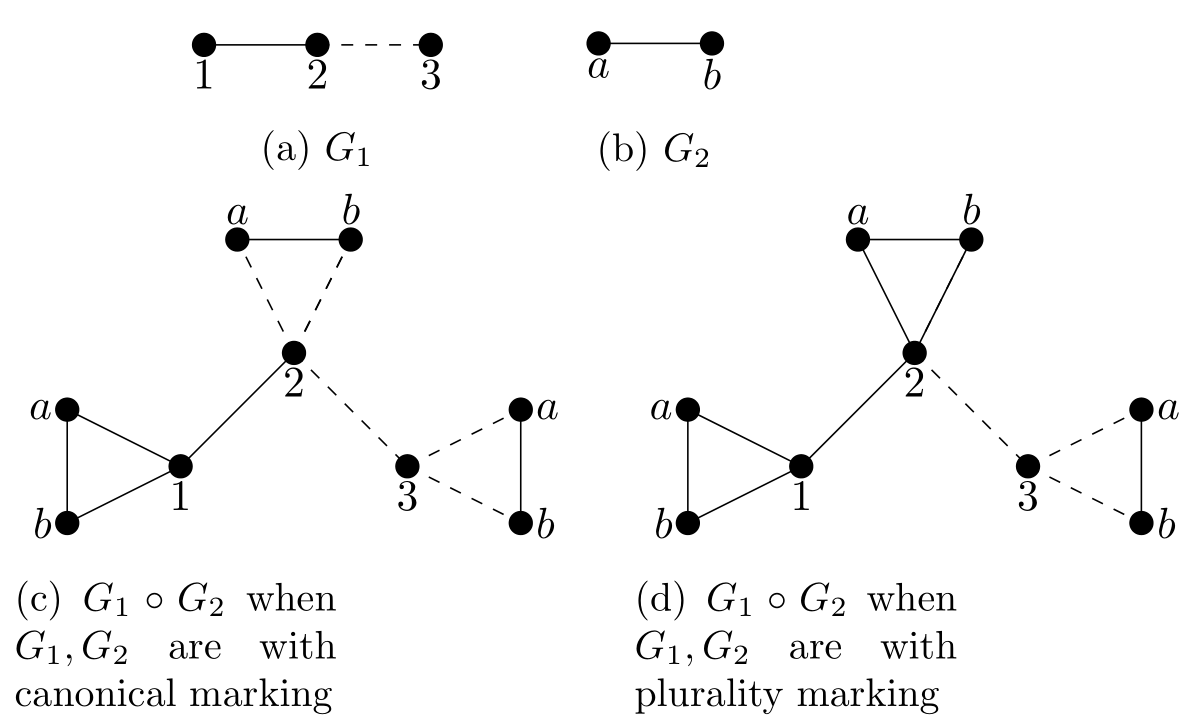}
    \caption{The corona product of $G_1 \circ G_2$ is shown in (c), (d) with canonical and plurality marking functions on $G_i$ , $i = 1,2$.}
    \label{fig:corona1}
\end{figure}


 Let $G = G^{ ( 0 )}$ be a simple connected graph \cite{ref:9}.  
 Then the corona graphs $G^{( m )}$ corresponding to the seed graph $G$ are defined by
 \begin{equation}
     G^{(m)}=G^{(m-1)}\circ G
 \end{equation}
 where $m (\geq 1 )$ is a natural number.
For example, the corona graphs $G^{ ( 1 )}$ and $G^{ ( 2 )}$ corresponding to the seed graph K$_3$ are shown in figure \ref{fig:corona2}.

The following are  
some observations associated with corona graphs.
\begin{itemize}
    \item  The number of nodes in $G^{ ( m )}$ is
    \begin{equation}
        |V(G^{(m)})|=n(n+1)^m
    \end{equation}
    \item  If $k$ is the number of edges in the seed graph $G ^{( 0 )}$ then the number of edges in $G^{ ( m )}$ is
    \begin{equation}
        | E ( G^{ ( m )} )| = k + ( k + n )(( n + 1 )^{m-1}-1) 
    \end{equation}
     \item The number of nodes added in $i$th $( i \leq m )$ step during the formation of $G^{ ( m )}$ is $n^2 ( n + 1 )^{i-1}$. 
\end{itemize}
\begin{figure}[hbtp]
    \centering
    \includegraphics[scale=0.4]{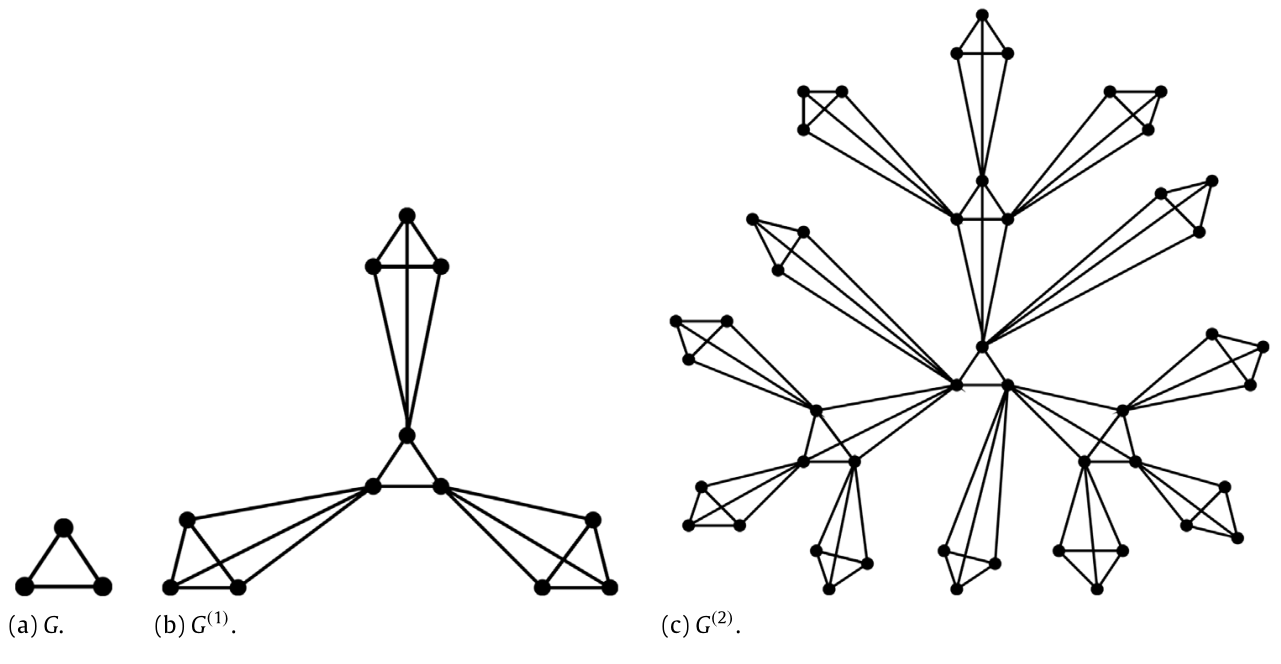}
    \caption{Examples of the Corona graphs (a) Seed graph K$_3$ (b) Corona graph for $G^{ ( 1 )}$ , and (c) Corona graph for $G^{ ( 2 )}$}
    \label{fig:corona2}
\end{figure}
We can do the similar construction of the signed graphs using a signed and marked seed. Let us define the adjacency and the Laplacian matrix of the signed graphs Corona product as well. We focus on the spectral properties of $G_1 \circ G_2$ . Let $G_1 = (V_1 ,E_1 ,\sigma_1 ,\mu_1 )$ and $G_2 = (V_2 ,E_2 ,\sigma_2 ,\mu_2 )$ be two signed graphs with $n$ and $k$ number of nodes, respectively. Suppose $V_1 = {u_1 ,...,u_n }$ and $V_2 = {v_1 ,...,v_k }$. Let us denote the marking vectors corresponding to vertices in $G_1$ and $G_2$ as
\begin{equation}
    \mu [V_1]=[\mu_1(u_1) \quad \mu_1(u_2) \quad ... \quad \mu_1(u_n)] \quad \text{and}\quad  \mu [V_2]=[\mu_2(v_1) \quad \mu_2(v_2) \quad ... \quad \mu_2(v_k)]
\end{equation}
where $µ_j (u) = 1$ if marking of $u = +$, otherwise $µ_j (u) = −1$, $j = 1,2$. Defining a matrix
\begin{equation}
    \text{diag}(\mu[V_1])=\text{diag}(\mu_1(u_1),\mu_1(u_2),...,\mu_1(u_n))
\end{equation}
defined using the marking $\mu[V_1]$, and similarly for $V_2$.

Then with a suitable labeling of the nodes the adjacency matrix of $G_1 \circ G_2$ is given by
\begin{equation}
    A(G_1\circ G_2)= \begin{bmatrix}
A(G_1) & \mu[V_2]\otimes \diag (\mu[V_1]) \\
\mu[V_2]^T\otimes \diag (\mu[V_1]) & A(G_2)\otimes I_n
\end{bmatrix}
\end{equation}
where $A(G_i )$ denotes the adjacency matrix associated with $G_i$ , $i = 1,2$, $\otimes$ denotes the Kronecker product of matrices, $I_n$ is the identity matrix of order $n$.

Similarly we have the definition for the Laplacian of the Corona product of two graphs $G_1$ and $G_2$ as
\begin{equation}
    L(G_1\circ G_2)= \begin{bmatrix}
L(G_1)+kI_n & -\mu[V_2]\otimes \diag (\mu[V_1]) \\
-\mu[V_2]^T\otimes \diag (\mu[V_1]) & (L(G_2)+I_k)\otimes I_n
\end{bmatrix}
\end{equation}
with rest similar definitions as for the adjacency matrix of the product. We can use these definitions to recursively construct $G^{(m)}$ using $G_1=G^{(m-1)}$ and $G_2=G$ and so on.

\subsubsection{Theorems on construction of eigenvalues and eigenvalues for product of corona graphs}
Constructing the higher order adjacency and Laplacian matrices using the previous section definitions is easy when done recursively. However, we need the eigenvalues and eigenvectors for these matrices which can be very tedious for even $m\geq4$ with $n=4$. We thus require some algorithm to construct the eigenvalues and eigenvectors recursively. This in general is not known, but for certain special graphs, this is possible if these special constrains on the graphs being multiplied are satisfied. We state two such extremely useful theorems (Theorem 2.3 and Theorem 2.6) extracted from \cite{ref:6} (without stating the proof here).

\begin{theorem}
\label{th:cor1}
Let $G_1$ be any signed graph on $n$ nodes and $G_2$ be a net-regular signed graph on $k$ nodes having net-regularity $d$. Let $(\lambda_i ,X_i )$ be an adjacency eigenpair of $G_1$, and $(\eta_j ,Y_j )$ be an eigenpair of $G_2$, $i = 1,...,n$, and $j = 1,...,k$. Let $\eta_k = d$. Then an adjacency eigenpair of $G_1 \circ G_2$ is given by $(\lambda^{(i)}_\pm,Z^{(i)}_\pm)$, $i = 1,...,n$ where
\begin{equation}
    \lambda^{(i)}_\pm=\frac{d+\lambda_i\sqrt{(d-\lambda_i)^2+4k}}{2}, \quad \text{ and } \quad Z^{(i)}_\pm=\begin{bmatrix}
X_i \\
\frac{\mu(v_1)}{\lambda^{(i)}_\pm -d}\diag (\mu[V_1])X_i\\
\frac{\mu(v_2)}{\lambda^{(i)}_\pm -d}\diag (\mu[V_1])X_i\\
\vdots \\
\frac{\mu(v_k)}{\lambda^{(i)}_\pm -d}\diag (\mu[V_1])X_i\\
\end{bmatrix}
\end{equation}
In addition, if all the nodes in $G_2$ are either positively or negatively marked, that is $\mu[V_2 ] = I_k$ or $−I k$ then
\begin{equation}
    \left( \eta_j,\begin{bmatrix}
    \boldsymbol{0}\\
    Y_j\otimes e_i
    \end{bmatrix}\right)
\end{equation}
is an eigenpair of $G_1 \circ G_2$ where $j = 1,...,k − 1$, and $\{e_i : i = 1,...,n\}$ the standard basis of $\mathbb{R}^n$ .
\end{theorem}

\begin{theorem}
\label{th:cor2}
Let $G_1= (V_1 ,E_1 ,\sigma_1 ,\mu_1 )$ be a signed graph on $n$ nodes and $G_2 = (V_2 ,E_2 ,\sigma_2 ,\mu_2 )$ be a signed graph on $k$ nodes. Let $V_2 = {v 1 ,...,v_k }$. Let $(\lambda_i ,X_i )$ be a signed Laplacian eigenpair of $G_1$, and $(\eta_j ,Y_j )$ are signed Laplacian eigenpairs of $G_2$, $i = 1,...,n, j = 1,...,k$. Let $d^− = d^−_j$
denote the negative degree of every node $v_j$ in $G_2$ . Then a signed Laplacian eigenpair of $G_1 \circ G_2$ is given by $(\lambda^ {(i)}_\pm ,Z^{(i)}_\pm)$ where
\begin{equation}
\begin{split}
        & \lambda^{(i)}_\pm=\frac{2d^-+1+\lambda_i+k+\pm\sqrt{[(2d^-+1)-(\lambda_i+k)]^2+4k}}{2}\\
        & Z^{(i)}_\pm=\begin{bmatrix}
X_i \\
-\frac{\mu(v_1)}{\lambda^{(i)}_\pm (2d^-+1)}\diag (\mu[V_1])X_i\\
-\frac{\mu(v_2)}{\lambda^{(i)}_\pm (2d^-+1)}\diag (\mu[V_1])X_i\\
\vdots \\
-\frac{\mu(v_k)}{\lambda^{(i)}_\pm (2d^-+1)}\diag (\mu[V_1])X_i\\
\end{bmatrix}
\end{split}
\end{equation}
where $i = 1,...,n$. Let $\eta_ k = 2d^−$ . In addition if all the nodes in $G_2$ are marked either positively or negatively marked then an eigenpair of $G_1 \circ G_2$ is
\begin{equation}
    \left( \eta_j+1,\begin{bmatrix}
    \boldsymbol{0}\\
    Y_j\otimes e_i
    \end{bmatrix}\right)
\end{equation}
is an eigenpair of $G_1 \circ G_2$ where $j = 1,...,k − 1$, and $\{e_i : i = 1,...,n\}$ the standard basis of $\mathbb{R}^n$ .
\end{theorem}

Both these theorems together give all the ordered eigenpairs for higher order Corona products recursively. They can be used when these special conditions mentioned in the theorems are satisfied and approaching the calculation for perfect state transfer is a lot easier.


%% file: Chapters/Corona.tex
This chapter corresponds to Part-I of the thesis project and was aimed at studying the perfect state transfer for signed graphs under the Corona product (as discussed in section \ref{sec:corona}). This consists of some important theorems and numerical results. Results from Part-II of the thesis are contained in chapters 4,5 and 6.

\section{No perfect state transfer in Corona product of graphs under Laplacian}
A general discussion for conditions for PST under Corona product of graphs is presented in \cite{ref:70}. A negative result indicating that no perfect state transfer is possible between any vertices of resulting Corona product of two given graphs, is obtained in \cite{ref:12}. We first present the the Theorem 4.1 in this paper.
\begin{theorem}
Let $G$ be a connected graph on $n \geq 2$ vertices and
$\Vec{H} = (H_1,\hdots,H_n)$ be an $n$-tuple of graphs on $m \geq 1$ vertices. Then there is no Laplacian perfect state transfer in $G \circ \Vec{H}$. 
\end{theorem}
The proof of this theorem can be found in \cite{ref:12}. Here, the second graph is a tuple of $n$ graphs for each vertex of $G$. This is more general scenario of Corona product where each vertex of $G$ is associated with a different graph $H_j$. Construction of this product is also given in section 3 of the same paper. This is a very strong result for Corona product. This result holds for the Laplacian evolution of the graph (that is, under Heisenberg coupling interaction of qubits) and is related to the impossibility of Laplacian perfect state transfer for trees \cite{ref:58}. Therefore, now adjacency state transfer remains to be explored for Corona product. More freedom can be explored in state transfer by using signed graphs which may support PST.

\section{Pretty good state transfer under Corona product of graphs}
The previous section indicates that there is no PST in Corona. However, there is a concept of pretty good state transfer where the fidelity is below 100\% but still close to it \cite{ref:66}\cite{ref:68}. Work done in \cite{ref:12} presents two theorems on pretty good state transfer in Coronas that we restate here without their proof. These theorems are very strong results for state transfer under Corona product. Proofs can be found in the original work. 

\begin{theorem}
Let $G$ be a graph on $n$ vertices and $\vec{H} = (H_1,\hdots ,H_n)$ be an $n$-tuple of graphs on $m \geq 1$ vertices. Suppose $G$ has perfect state transfer between vertices $u$ and $v$, and let $2^r$ be the greatest power of two dividing each element of the eigenvalue support of $u$. If $2^{r+1}$ divides $m+ 1$, then there is pretty good state transfer between vertices $(u, 0)$ and $(v, 0)$ in $G \circ \vec{H}$.
\end{theorem}
These are the sufficient conditions for pretty good state transfer under Corona product of graphs.

\begin{theorem}
Let $\vec{H} = (H_1,H_2)$ be a pair of graphs on $m \geq 1$ vertices. Then $K_2\circ \vec{H}$ has pretty good state transfer between the vertices of $K_2$.
\end{theorem}
This theorem establishes a particular class of Corona graphs where the pretty good transfer is possible. In the light of these three theorems and subsection \ref{subsec:balanced} what remains to explore are the signed balanced graphs under the Corona product with XY coupling. And also the unbalanced signed XY and Heisenberg coupling based graphs under Corona. We use Theorem \ref{th:cor1} and Theorem \ref{th:cor2} to construct some specific graphs and numerically study perfect state transfer and pretty good state transfer under Corona product (for the marking schemes mentioned in section \ref{sec:corona}) for both XY and Heisenberg interactions.

\section{Numerical study of some signed graphs under Corona product}
Some conclusive results based on the previous section study are presented here. All the results and conclusions for the current progress are based on numerical study and examples based on construction. We start with constructing the Hamiltonian (and thereby the adjacency and the Laplacian matrices) of a possible graph and then calculate the fidelity of a perfect state transfer from a given node to another given node of the graph. Perfect transfer is possible if a non-zero and finite time $t_0$ exists at which fidelity is unity. Most often the value of this $t_0$ is not very important for our work as is the check for a perfect transfer. If perfect transfer is possible for some time $t_0$, it is enough evident to classify the possible graphs which support perfect quantum state transfer. We plot the fidelity as a variation of time as a parameter in the problem. For mirror-symmetric graphs between a chosen pair of nodes, the fidelity is periodic as a function of time for the given pair of nodes. If periodicity does not hold, then fidelity follows quite complicated variation w.r.t. time evolution in large graphs. Numerics has been performed in \href{https://www.wolfram.com/mathematica/}{Wolfram Mathematics 11.3}.
\subsubsection{Example: 1}
This is the case of a 2-clique signed graph as shown in figure \ref{fig:eg1}. This is the simplest example of net regular balanced signed graph that satisfies Theorem \ref{th:cor1} and Theorem \ref{th:cor2}. Perfect transfer is possible from node 1 to node 3 and vice versa and also node 2 to node 4 and vice versa (shown in green). Perfect transfer is forbidden for the rest combinations, which are all adjacent (shown in red). All vertices are negatively marked according to canonical marking scheme. Graph is mirror symmetric between 1 \& 3 and 2 \& 4 and hence the transfer is periodic with a period $2t_0=\pi$, with $t_0$ as the time for perfect transfer from one node to another. All these properties are summarised by finding that the fidelity is
\begin{equation}
    F(t)=\sin^2(t)
\end{equation}
for both the cases.
For most cases, the fidelity is not at all simple to calculate in compact analytical form as above and instead only numerical computation is possible.
\begin{figure}[hbtp]
    \centering
    \includegraphics[scale=0.8]{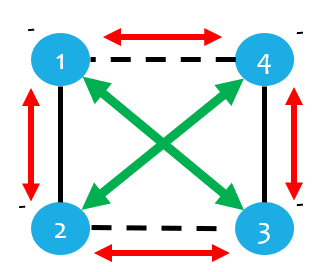}
    \caption{2-clique signed graph. Simplest case for a signed graph that is net-regular and balanced with markings as shown. All edge weights are $\pm1$, with solid lines as +1 and dotted lines as -1.}
    \label{fig:eg1}
\end{figure}
    
The first self corona product is shown in figure \ref{fig:eg1cor}

\begin{figure}[hbtp]
    \centering
    \includegraphics[scale=0.7]{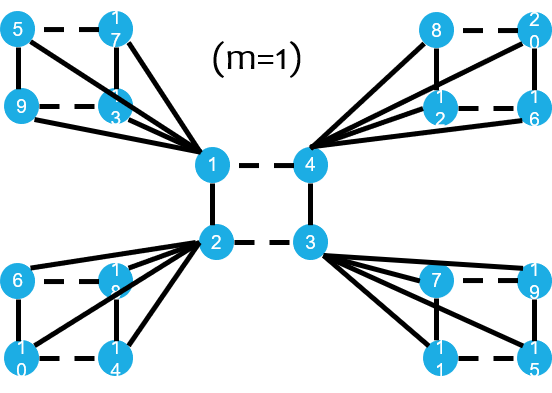}
    \caption{First corona product ($m=1$) of signed 2-clique.}
    \label{fig:eg1cor}
\end{figure}

More self-corona products with the seed can be found out in similar fashion. The fidelity between 1 \& 3 and 2 \& 4 is observed to \textit{decrease} as $m$ increases. This is plotted in figure \ref{fig:eg1graph}. This implies that Corona product for this graph does not support perfect state transfer and it satisfies the conclusions drawn in the previous section. These corona products are balanced and any signed balanced graph has the same transfer dynamics as its unsigned version where (in this case), corona product perfect state transfer is forbidden for Laplacian and adjacency constructions.

\begin{figure}[hbtp]
    \centering
    \includegraphics[scale=0.6]{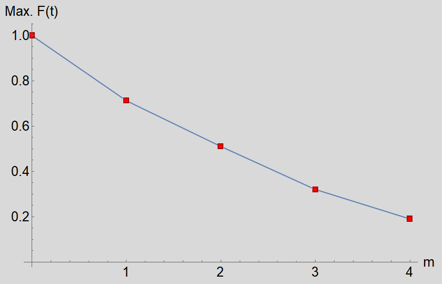}
    \caption{Plot of maximum fidelity $\max [F(t)]$ vs $m$ between nodes 1 \& 3 and 2 \& 4.}
    \label{fig:eg1graph}
\end{figure}

Inferences from this example are as follows:
\begin{itemize}
    \item Fidelity between any two nodes in any graph decreases as m increases (more graph bulk increases) for most (not all) graphs.
    \item Symmetry is important factor. In many cases, symmetric ends support perfect transfer.
\end{itemize}

\subsubsection{Example: 2}
There is not much similarity between graph dynamics when we change a graph slightly. This can be seen in this example. It does not support perfect state transfer between any given pair of nodes over all vertices. Fidelity below 90\% is practically very bad in terms of experiments and hence this example is a bad architecture. The maximum fidelity allowed by this seed graph is only 0.68. However, the first corona product of this graph allows the same fidelity between 1 \& 3 which does not decay (another important observation).
\begin{figure}[hbtp]
    \centering
    \includegraphics[scale=0.8]{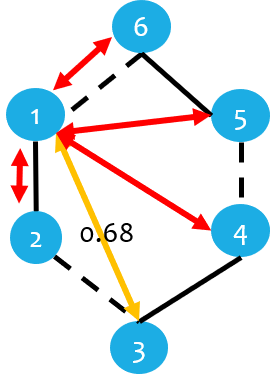}
    \caption{Another non-trivial example forming a signed hexagon $K_6$.}
    \label{fig:eg2}
\end{figure}

\subsubsection{Example: 3}
In this example of signed $K_8$, the near perfect quantum state transfer is possible between diametrically opposite ends. Also notice that the effective distance between these ends is more than the length of a $n=3$ chain yet a near perfect transfer is possible due to the signed graph nature and also the cyclicity imposed.

The corona products of the graph offers decaying fidelity similar to the 2-clique graph and much less fidelity on the rest of the pair of nodes.
\begin{figure}[hbtp]
    \centering
    \includegraphics[scale=0.8]{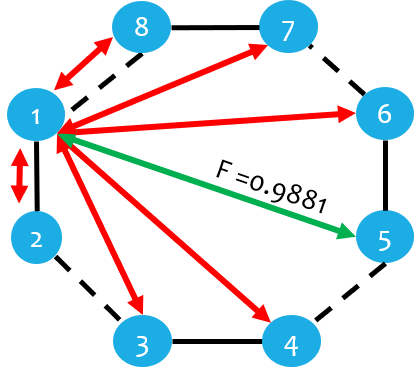}
    \caption{Another non-trivial example forming a signed octagon $K_8$. Near perfect state transfer is possible.}
    \label{fig:eg3}
\end{figure}

\subsubsection{Example: 4}
This is the fist non-trivial and simplest example of a signed unbalanced graph that satisfies the net-regularity and net-negative degree to make use of the concerned theorems for the adjacency and Laplacian eigenvalues and eigenvectors. It allows prfect state transfer for two symmetric adjacent nodes as shown in figure \ref{fig:eg4}. The importance of this graph is that it is not proved to show forbidden perfect state transfer for corona products. However, the dynamics of the corona product follows the similar results as for many other examples that the perfect state transfer is not supported after the corona products and fidelity between any pair of given nodes decays as the order increases. But it leaves the question of exploring other variations of similar graphs which are signed and unbalanced.

\begin{figure}[hbtp]
    \centering
    \includegraphics[scale=1]{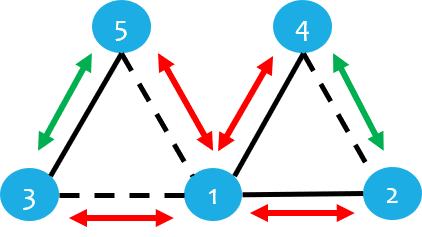}
    \caption{First non-tivial example that satisfies the theorems and is unbalanced.}
    \label{fig:eg4}
\end{figure}

\subsubsection{Example: 5}
This is another signed variation of the $K_6$ with more connectivity, however, the graph is still net-regular to take the advantages of the theorems. The graph is not symmetric. In contrast to symmetry argument before, this still allows perfect transfer between 1 \& 5 and near perfect transfer in other three pair of nodes marked with yellow in \ref{fig:eg5}. The graph is not complete as every node is not connected to every other node as can be seen trivially. This seed as a graph is very well constructed example for perfect state transfer. And this is a very peculiar example with the property that the fidelity for perfect transfer between 3 \& 5 is preserved even with corona products of higher order and it is also conserved for the other three near-perfect transfer ends. This graph for these pair of nodes behaves as being invariant for fidelity for corona products. This is generally true for other nodes as well that the fidelity does not decay with the product order $m$. However, after the corona product, new nodes are observed to have less fidelity than any pair on the seed itself.
\begin{figure}[hbtp]
    \centering
    \includegraphics[scale=0.9]{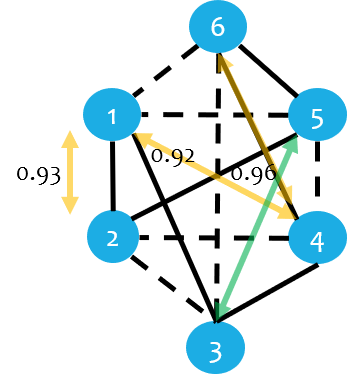}
    \caption{Another variation of signed net-regular $K_6$ with more connectivity.}
    \label{fig:eg5}
\end{figure}

\subsubsection{Example: 6}
These are two signed and conjugate pair variations of $K_6$ as shown in figure \ref{fig:eg6}. These two examples indicate that signs and connectivity of graphs change the perfect transfer with strong dependence. Perfect transfer is allowed between 3 \& 6 ends in both graphs. This also proves that these conjugate signing schemes are equivalent in perspective of perfect state transfer problem, however they are two very different graphs. Symmetry argument holds for this graph and perfect state transfer is periodic. For the corona product, fidelity between any two nodes chosen initially is conserved and also perfect transfer is possible between the same nodes even with the higher order corona products much similar to example 5.

\begin{figure}[hbtp]
    \centering
    \includegraphics[scale=0.8]{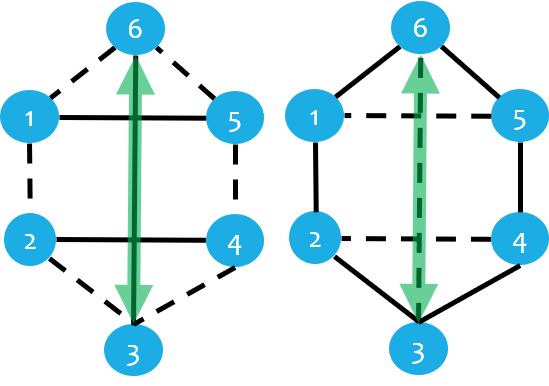}
    \caption{Two conjugate signed versions of $K_6$.}
    \label{fig:eg6}
\end{figure}

\subsubsection{Example: 7}
This example is a fully connected version of $K_8$. One negative edge allows perfect transfer between the same edge, indicating the importance of negative edges. There is very rich dynamics of this graph to see the different combinations of signed edges in different signed versions. This is symmetric between nodes 7 \& 8.
\begin{figure}[hbtp]
    \centering
    \includegraphics[scale=0.8]{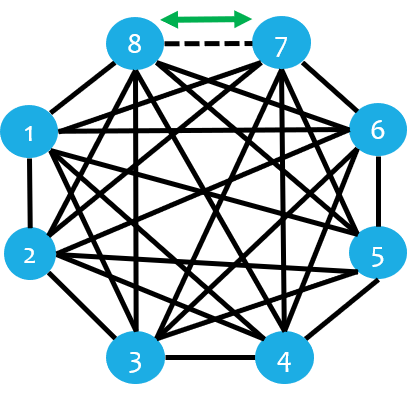}
    \caption{Another variation of signed net-regular $K_6$ with more connectivity.}
    \label{fig:eg7}
\end{figure}

\subsubsection{Example: 8}
Another signed and complete version of the 2-clique graph. Allows perfect state transfer between the same pair of nodes as the initial graph. has exactly the same dynamics. The fidelity decreases with $m$. No perfect state transfer possible in the higher corona products implying it is also a bad seed graph to start with. The graph should be chosen such that the fidelity is non-decreasing between any pair of nodes with increasing bulk on the graph.
\begin{figure}[hbtp]
    \centering
    \includegraphics[scale=0.8]{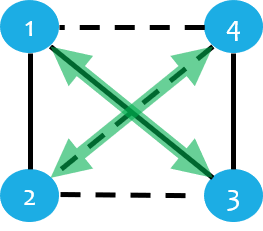}
    \caption{Another variation of signed 2-clique K$_4$ with more connectivity.}
    \label{fig:eg8}
\end{figure}

\subsubsection{Example: 9}
One signed unbalanced graph and another its completed connected version. Both these graphs allow only below 0.5 fidelity and the corresponding corona products are even far below this value. Also, for these non-symmetric graphs, the flow of fidelity with time is very irregular and chaotic as shown in figure \ref{fig:eg9graph}.
\begin{figure}[hbtp]
    \centering
    \includegraphics[scale=0.9]{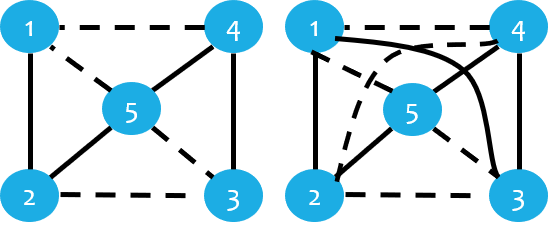}
    \caption{Another variation of signed net-regular $K_6$ with more connectivity.}
    \label{fig:eg9}
\end{figure}

\begin{figure}[hbtp]
    \centering
    \includegraphics[scale=0.6]{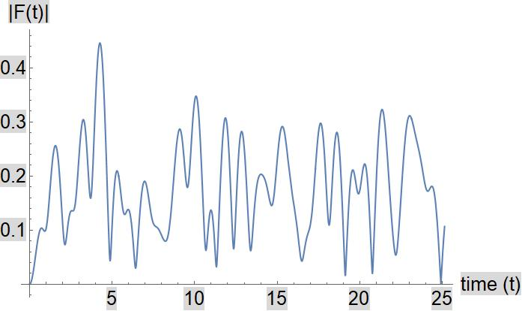}
    \caption{Typical flow of fidelity $F(t)$ w.r.t. time  $t$ for a pair of nodes of example 9.}
    \label{fig:eg9graph}
\end{figure}

\subsubsection{Example: 10}
This is not a net-regular graph and hence we cannot take the advantages of the theorems for the eigenpair construction. However, the manual computations reveal that no perfect state transfer s possible between any nodes in seed as well as the corona products. It gives very low fidelity for any pair of nodes on the seed as well as the corona products. It can be looked in contrast to the linear chain of two-links where perfect transfer is possible. We conclude that the perfect transfer is lost very quickly upon adding on extra link in between for signed as well as unsigned version.
\begin{figure}[hbtp]
    \centering
    \includegraphics[scale=0.8]{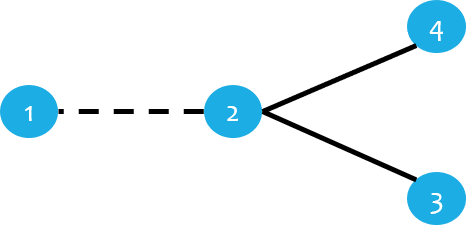}
    \caption{Non-cyclic signed graph.}
    \label{fig:eg10}
\end{figure}

\section{Interpretation from the numerical study}
The task was to check by the construction of examples that signed corona graphs can support perfect state transfer which can be seen from these special examples. Around 30 examples were constructed while these 10 were important. In summary:
\begin{itemize}
    \item We see that perfect state transfer is possible for certain signed graphs which preserve unity fidelity under Corona product for certain pair of nodes. Thereby showing that under signed graphs we can recover the PST in Coronas.
    \item This indicates a class of possible graphs which sustain perfect transfer in contrast to Theorem 4.1 in \cite{ref:12} which forbids perfect transfer in unsigned graphs
    \item Fidelity for some node pairs increase while for others it decreases (this needs to be classified analytically)
    \item Symmetries in the graph allow to choose the pair of nodes for perfect transfer
\end{itemize}

%% file: Chapters/TwoHop.tex
	\label{chap:twohop}
	\section{Introduction and motivation}

    
    In the light of section \ref{sec:routingimpossible} and section \ref{sec:chainlimitation}, there are limitations to routing and transfer distance. Transfer distance was tackled in \cite{ref:4} as presented in section \ref{sec:longchains} as Cartesian product resulting for PST for the pair of antipodal points on the hypercube of any order. This is limited to only antipodal points and cost of the constructing such large number of edges is high just to enable PST over two given vertices on the graph. In this chapter we aim to propose a solution to both of these problems. In our state transfer scheme,
    \begin{itemize}
        \item Arbitrary number of vertices $n$ (intrinsically qubits) is allowed for the qubit network
        \item Perfect state transfer is enabled from all-to-all vertices on the graph in at most time $2t_0$ with the same fidelity of unity
        \item Enables a growing network architecture for scalability of quantum network while preserving both the above properties
    \end{itemize}

	We assume that a quantum communication network is a connected graph, allowing perfect quantum state transfer between any two vertices. The above three features can be enabled when we have the freedom of edge switching, that is, we are allowed to switch off and on the couplings (edges) in the graph. This corresponds to switching off and on the interaction between the qubits. We justify this requirement both theoretically and experimentally (in chapter \ref{chap:phyreal}). A connected graph is described in a graph theoretic fashion, which has a path that is a sequence of vertices and edges between any two vertices. Therefore, the graph offers a classical platform for a quantum mechanical operation. There is no graph other than $K_2$ allowing perfect state transfer between any two vertices. In general, the perfect state transfer is possible between a few specific vertices, in a larger graph. Increasing the number of attempts for state transfer makes is limited between two specific vertices only. It leads us to the conclusion that only quantum mechanical process is not sufficient to fulfill our requirements, that is the perfect state transfer between any two vertices of a graph. Therefore, we propose a hybrid of combinatorial and quantum information theoretic method, such that, a perfect quantum state transfer is possible between any two vertices of the graph. Our results in this work hold both for XY as well as the Laplacian coupling Hamiltonian.

	\section{Graph labeling and associated Hilbert spaces}
	
		Let $G = (V(G), E(G))$ be a graph with $V$ vertices. We label the vertices by the integers $0, 1, 2, \dots (|V| - 1)$.
		Any integer $v \in \{0, 1, \dots (|V| - 1)\}$ has a $(k + 1)$ term binary representation $\bin(i)$, where $2^{k} < |V| \leq 2^{(k + 1)}$, for $k = 1, 2, \dots$. Now, $\ket{\bin(v)}$ represents a quantum state vector in $\mathbb{C}^{2^{(k + 1)}} = \mathbb{C}^2 \otimes \mathbb{C}^2 \otimes \dots \otimes  \mathbb{C}^2((k + 1)$-times). For example if $v = 2$, then $\bin(v) = 10$ and $\ket{\bin(v)} = \ket{10} = \ket{1} \otimes \ket{0}$, where $\ket{0} = \begin{bmatrix}1 \\ 0 \end{bmatrix}$ and $\ket{1} = \begin{bmatrix} 0 \\ 1 \end{bmatrix}$ are the standard basis vectors. This coincides with the first excitation subspace of the $XY$ and Laplacian Hamiltonian.
		
		Corresponding to the vertex $v$ we also associate a state vector $\ket{v} \in \mathbb{C}^{|V|}$. If we denote a vector in $\mathbb{C}^{|V|}$ as $\ket{m} = (m_0, m_1, \dots m_{|V| - 1})^T$ then the vector $\ket{v}$ is given by $m_u = 0$ for $u \neq v$ and $m_v = 1$. We define a linear transformation $R : \mathbb{C}^{2^k} \rightarrow \mathbb{C}^{2^{(k + 1)}}$ by $R \ket{m} = (m_0, m_1, \dots m_{|V| - 1}, 0, 0, \dots 0 ((2^{(k + 1)} - |V|)$ -times$))^T$ which will help us to extend over to a larger Hilbert space by appending extra fixed labels for a state. Therefore, now $R\ket{m}$ belongs to $\mathbb{C}^{2^{k+1}}$. When we have $n$ number of vertices, where $2^k\leq n \leq 2^{k+1}$, then we adopt the labeling for $2^{k+1}$ vertex graph.

	\section{Hypercubes and their properties}
	
		Cartesian product was presented in section \ref{sec:cartesian}. Here, for hypercubes, we are concerned with the Cartesian product of a graph with itself. The Cartesian product of $G$ with itself is denoted by $G^{\square 2} = G \square G$. Similarly, for any natural number $k$ we denote the $k$-th Cartesian product as $G^{\square k}$. The Cartesian product is associative. In general, it is commutative when the graphs are not labelled. Also, the graphs $G \square H$ and $H \square G$ are naturally isomorphic.
		  
		A hypercube $Q_k$ of dimension $k$ is a graph with $2^k$ vertices for $k = 0, 1, 2, \dots$. For $k = 0$ the graph $Q_0$ consists of a single vertex. For $k = 1$ we have two vertices and an edge in the hypercube $Q_1$, which can also be described as the complete graph $K_2$ with two vertices. When $k \geq 3$ we can justify $Q_k = (K_2)^{\square n}$. Hence, the Cartesian product of two hypercubes is another hypercube, that is $Q_i \square Q_j = Q_{i + j}$ \cite{harary1988survey}.
		
		Let the vertices of $K_2$ are given by $0$ and $1$. Then the vertices of $Q_k$ are represented by the elements in the set $\{0, 1\}^{\times k} = \{0, 1\} \times \{0, 1\} \times \dots \times \{0, 1\} (k$-times). Note that the elements of $\{0, 1\}^{\times k}$ are the $k$-term binary representations of the natural numbers $0, 1, \dots 2^k$. Hence $\bin(v)$ denotes the label of the vertex $v$ in the hypercube. An important property for the hypercubes is that any two vertices $u$ and $v$ in $Q_k$ are adjacent when the Hamming distance between $\bin(u)$ and $\bin(v)$ is $1$, for $k > 2$.
		
		\begin{definition}
			\textbf{Antipodal points}: Two vertices $u$ and $v$ which are labeled by the binary sequences $\bin(u) = (u_j)_{j = 0}^{(k - 1)}$ and $\bin(v) = (v_j)_{j = 0}^{(k - 1)}$ in the hypercube $Q_k$ are called the antipodal points if $u_j \neq v_j$, for all $j$.
		\end{definition}
		For example, the antipodal points of $Q_2$ are $0$ and $1$. The antipodal points of $00$ in $Q_3$ is $11$ and $01$ for $10$. In case of $Q_3$, we can write the antipodal points as pairs $(000, 111), (001, 110), (101, 010)$ and $(011, 100)$.
		
		Note that, in case of the hypergraphs $Q_k$ we have $\mathbb{C}^{|V(Q_k)|} = {C}^{2^{(k + 1)}}$. Therefore, the linear operator $T$ is the identity function for this case.
		
		The hypercube of dimension $0$ is a single vertex labeled by $0$ only. The hypercube of dimension $1$ is denoted by $Q_1$ which is depicted as follows
		\begin{center}
			\begin{tikzpicture}
				\draw [fill] (0, 0) circle [radius = 0.1];
				\node [left] at (0, 0) {$0$};
				\draw [fill] (1, 0) circle [radius = 0.1];
				\node [right] at (1, 0) {$1$};
				\draw (0, 0) -- (1, 0);
			\end{tikzpicture}
		\end{center}
		Note that $Q_1$ is the complete graph with two vertices $K_2$. The hypercube $Q_2 = K_2 \square K_2$ has four vertices which is represented by
		\begin{center} 
			\begin{tikzpicture}
			\draw [fill] (0, 0) circle [radius=0.1];
			\node [below left] at (0, 0) {$00$};
			\draw [fill] (0, 1) circle [radius=0.1];
			\node [above left] at (0, 1) {$01$};
			\draw [fill] (1, 0) circle [radius=0.1];
			\node [below right] at (1, 0) {$10$};
			\draw [fill] (1, 1) circle [radius=0.1];
			\node [above right] at (1, 1) {$11$};
			\draw (0,0) -- (0, 1) -- (1, 1) -- (1, 0) -- (0, 0);
			\end{tikzpicture}
		\end{center}
		Also, the hypercube $Q_3 = Q_2 \square K_2$ has $8$ vertices which is given by
		\begin{center}
			\begin{tikzpicture}[scale = .75]
			\draw [fill] (0, 0) circle [radius=0.1];
			\node [below left] at (0, 0) {$000$}; 
			\draw [fill] (0, 3) circle [radius=0.1];
			\node [above left] at (0, 3) {$001$};
			\draw [fill] (3, 0) circle [radius=0.1];
			\node [right] at (3, 0) {$010$};
			\draw [fill] (3, 3) circle [radius=0.1];
			\node [above right] at (3, 3) {$011$};
			\draw (0,0) -- (0, 3) -- (3, 3) -- (3, 0) -- (0, 0);
			\draw [fill] (1.5, 1.5) circle [radius=0.1];
			\node [below left] at (1.5, 1.5) {$101$};
			\draw [fill] (4.5, 1.5) circle [radius=0.1];
			\node [above right] at (4.5, 1.5) {$111$};
			\draw [fill] (4.5, -1.5) circle [radius=0.1];
			\node [below right] at (4.5, -1.5) {$110$};
			\draw [fill] (1.5, -1.5) circle [radius=0.1];
			\node [below left] at (1.5, -1.5) {$100$};
			\draw (1.5, 1.5) -- (4.5, 1.5) -- (4.5, - 1.5) -- (1.5, -1.5) -- (1.5, 1.5);
			\draw (0, 3) -- (1.5, 1.5);
			\draw (3, 3) -- (4.5, 1.5);
			\draw (0, 0) -- (1.5, -1.5);
			\draw (3, 0) -- (4.5, -1.5);
			\end{tikzpicture}
		\end{center}
		
		In general, the hypercube $Q_{k + 1} = Q_k \square K_2$ for $k \geq 2$.  The vertex labels of $Q_k$ are the distinct binary sequences of length $k$. The Cartesian product between $Q_k$ and $K_2$ makes the number of vertices doubled as well as add an additional index to the vertex labeling. 
		
		A hypercube of dimension $k$ consists of smaller hypercubes of dimension $i$ for $i = 0, 1, \dots (k - 1)$. All the hypercubes of dimension $i$ are unique upto isomorphism. But, all the hypercubes $Q_i$ embedded in $Q_k$ have different vertex labelling. The number of distinct hypercubes $Q_i$ embedded in $Q_k$ is given by $\binom{k}{i} 2^{(k - i)}$ \cite{klavvzar2006counting}. The next lemma suggests how to distinguish a particular subhypercube which is embedded in a larger hypercube.
		
		\begin{lemma}\label{projection_lemma} 
			Let the $2^k$ vertices of the hypercube $Q_k$ be labeled by the binary sequences $\bin(v) = (v_j)_{j = 0}^{(k - 1)}, v_j \in \{0, 1\}$. For some $i$ with $1 \leq i < k$, consider $(k - i)$ integers $\{m_t: t = 1, 2, \dots (k - i)\}$, such that $0 \leq m_1 < m_2 < \dots < m_{(k - i)} \leq (k - 1)$, and a binary sequence $M = (M_t)_{t = 1}^{(k - i)}$. Corresponding to the set of indices $\{m_t\}$ and the binary sequence $M$ construct a set of vertices $V_i = \{v: \bin(v) = (v_j)_{j = 0}^{(k - 1)}, v_{m_t} = M_t ~\text{for}~ t = 1, 2, \dots (k - i) \} \subset V(Q_k)$. Then the induced subgraph of $Q_k$ generated by $V_i$ is isomorphic to the hypercube $Q_i$  
		\end{lemma}
	
		\begin{proof}
			To label a vertex in $Q_i$ we need a binary sequence of length $k$. For constructing the vertex set $V_i$ we keep $(k - i)$ terms in the sequence constant, which are equal to the elements of $M$. Therefore, number of elements in $V_i$ is $2^i$, which is the number of vertices in $Q_i$.
			
			Let $H = (V(V_i), E(H))$ be the induced subgraph of $Q_k$ generated by $V_i$. Clearly, $V(Q_i) = V_i$. We write $\bin(v) = (v_j)_{j = 0}^{(k - 1)} = (v_0, v_1, \dots v_{(k - 1)})$. Given the set of indices $\{m_t\}$ define $\bin(v) \ominus \{m_t\} = (v_0, v_1, \dots v_{m_1 - 1}, v_{m_1 + 1}, \dots v_{m_2 - 1}, v_{m_2 + 1}, \dots v_k)$, that is we remove the terms of $\bin(v)$ corresponding to the indices in $\{m_t\}$. Clearly, after removing $(k - i)$ terms from $\bin(v)$ we find a new binary sequences of length $i$. \footnote{Sir, please find if the operation is standard in the literature of coding theory or Boolean functions.}
			
			Now, define a function $f: V_i \rightarrow V(Q_i)$, such that, $f(v) = \bin(v) \ominus \{m_t\}$ and prove that it is bijective function. For any two different $u$ and $v$ in $V_i$, we have $u_{m_t} = v_{m_t}$. When these equal entries are removed we get two different binary sequences. Hence, the function $f$ is injective. Consider any element $w \in V(Q_i)$. Note that, $\bin(w) = (w_j)_{j = 0}^{(i - 1)}, w_i \in \{0, 1\}$. This binary sequence of $i$ terms can be extended to a binary sequence of $k$ terms by including the elements $M_t$ of $M$ at the $m_t$-th index. It concludes that $f$ is surjective. Therefore $f$ is a bijective mapping. This function does the reverse of the function $R|.\rangle$ defined previously.
			
			Consider two adjacent vertices $u$ and $v$ in $H(V_i)$. As $H(V_i)$ is a subgraph of $Q_k$ the binary sequences $\bin(u)$ and $\bin(v)$ has Hamming distance $1$. The construct of $V_i$ suggests that $m_t$-th entries of $\bin(u)$ and $\bin(v)$ are equal, which are removed by the function $f$. Therefore, the sequences $f(u)$ and $f(v)$ has Hamming distance $1$. As $f(u)$ and $f(v)$ represents two vertices in $Q_i$, they are adjacent. 
			
			Alternatively consider two adjacent vertices in $Q_i$ which are labeled by binary sequence of length $i$. These sequences have Hamming distance $1$. We add equal entries at equal indexed position to get their inverse in $V_i \subset V(Q_k)$. The inverses also have Hamming distance $1$. Hence, they are adjacent in $Q_k$. As $H(V_i)$ is an induced subgraph of $Q_k$, they are also adjacent. Therefore, $f$ is a graph isomorphism. 
		\end{proof}
	
		\begin{corollary}\label{induced_hypercube}
			Consider two hypercubes $Q_p$ and $Q_q$ with $p > q$. Then there is an induced subgraph of $Q_p$ with $2^q$ vertices, which is isomorphic to $Q_q$. 
		\end{corollary}
	
		\begin{proof}
			The vertices of $Q_p$ can be labeled by the sequence of binary digits $\bin(v) = (v_j)_{j = 0}^{(p - 1)} = (v_0, v_1, \dots v_{(p - 1)})$ of length $p$. Similarly, the vertices of $Q_q$ can be given by the sequences $\bin(u) = (u_j)_{j = 0}^{(q - 1)} = (u_0, u_1, \dots u_{(q - 1)})$. Now we construct a set of vertices $V_q = \{v: v_i = 0 ~\text{for}~ i = 0, 1, 2, \dots (p - i - 1)\}$. Clearly, $V_q$ has $2^q$ vertices, which is the number of the vertices in $Q_q$. Consider the induced subgraph of $Q_p$ generated by $V_q$ which is isomorphic to $Q_q$. It can be easily shown by considering that the adjacent vertices have hamming distance one. 
		\end{proof}
	
		Recall that, a Hamiltonian path is a path in a graph that visits each vertex only once. A Hamiltonian cycle is a Hamiltonian path which is a cycle. Every hypercube $Q_n$ with $n > 1$ has a Hamiltonian cycle.
		
		\begin{corollary}
			Let the vertices $v$ of a hypergraph $Q_i$ are labeled by the binary sequences $\bin(v) = (v_j)_{j = 0}^{(k - 1)}$ where $k > i$, such that, the Hamming distance between the labels of any two adjacent vertices is one. Then, there are a sequence of non-negative integers $\{m_t\}$ of length $(k - i)$, such that, $0 \leq m_1 < m_2 < \dots < m_{(k - i)} \leq (k - 1)$, and a binary sequence $M = (M_t)_{t = 1}^{(k - i)}$ which determine a set of vertices $V_i = \{v: \bin(v) = (v_j)_{j = 0}^{(k - 1)}, v_{m_t} = M_t$ for $t = 1, 2, \dots (k - i) \} \subset V(Q_k)$. Then the induced subgraph of $Q_k$ generated by $V_i$ has the same vertex labeling as of the hypercube $Q_i$  
		\end{corollary}
	
		\begin{proof}
			Assume that $v = v_1, v_2, \dots v_{2^i}, v_{2^i +1} = v$ is a Hamiltonian cycle starting and ending at a vertex $v \in V(Q_i)$. Let the vertex labels $v_p$ are given by $\bin(v_p) = (v_{pj})_{j = 0}^{(k - 1)}$ for $p = 1, 2, \dots 2^i$. The Hamming distance between $\bin(v_1)$ and $\bin(v_2)$ is $1$. Therefore, there is an index $q_1$ such that $v_{1j} = v_{2j}$ when $j \neq q_1$. Similarly there are indices $q_2, q_3, \dots q_{2^i}$, such that $v_{2j} = v_{3j}$ for $j \neq q_2$; $v_{3j} = v_{4j}$ for $j \neq q_3$, and so on. Note that, the binary sequence $\bin(v)$ has length $k$. Therefore, $q_1, q_2, \dots q_{2^i}$ may not be all distinct. In the sequence $\bin(v)$ the element $v_j$ represents $0$ or $1$. Therefore, to represent $q_1, q_2, \dots q_{2^n}$ we need only $i$ positions in the sequence $\bin(v)$, which are given by $q_1', q_2', \dots q_i'$. Now define the entries of $\{m_t\}$, such that, $0 \leq m_1 < m_2 < \dots < m_{(k - i)} \leq (k - 1)$ and $m_t \notin \{q_1', q_2', \dots q_i'\}$ for any $t$. Construct the sequence $M = \{M_t\}_{t = 1}^{(k - i)}$, such that $M_t = v_{1m_t}$. Note that, for any $t$ we have $M_t = v_{1m_t} = v_{2m_t} = \dots = v_{2^im_t}$, otherwise the condition of unite Hamming distance between the vertex labeling of adjacent vertices will be violated. 
			
			Now, in the hypergraph $Q_k$ we construct the set of vertices $V_i = \{v: v_{m_t} = M_t ~\text{for}~ t = 1, 2, \dots (k - i) \}$ with respect to the sequences $\{m_t\}_{t = 1}^{(k - i)}$ and $M$. Clearly, the induced subgraph $G(V_i)$ of $Q_k$ generated by $V_i$ is a hypercube $Q_i$. The vertex labeling of $Q_i$ considered in the statement and vertex labellings of $G(V_i)$ are equal because of the particular choice of $\{m_t\}$ and $M$.
		\end{proof}

	\section{Quantum walk and perfect state transfer}
		
	To define quantum walk and state transfer on the graphs we associate a basis vector of $\mathbb{C}^{|V|}$ to the individual vertex $v \in V(G)$.  A continuous time quantum walk on a graph $G$ is defined using the Schr\"{o}dinger equation with the $A(G)$ as the Hamiltonian \cite{kendon2011perfect}. If $\ket{\zeta(t)} \in \mathbb{C}^{|V|}$ is a time-dependent quantum state, then the evolution of the quantum walk is given by
\begin{equation}\label{adj:dyn}
			\ket{\zeta(t)} = \exp(-itA(G))\ket{\zeta(0)},
\end{equation} 
where $\ket{\zeta(0)}$ is the initial state vector. The probability for getting the quantum state localised at the vertex $v$ at time $t$ is given by $|\braket{v|\zeta(t)}|^2$. We say $G$ has a perfect state transfer from vertex $u$ to vertex $v$ at time $t_0$ if
		\begin{equation}
			|\braket{v| \exp(-it_0A(G)) | u}| = 1.
		\end{equation}
		This is the same condition for perfect state transfer expressed in graph theoretic fashion \cite{ref:7} and implies equation \ref{eqn:pstgeneral}. When $\{\ket{v}: v \in V(G)\}$ represents the computational basis of $\mathbb{C}^{|V|}$, we say that the graph $G$ allows a perfect state transfer from the vertex $u$ to $v$ if the $(u,v)$-th term of $\exp(-itA(G))$ has magnitude $1$. Besides in \cite{ref:2}, a necessary and sufficient condition is proved for PST. The well-known examples of graphs allowing perfect state transfer over long distances are described below \cite{ref:4}\cite{ref:5}
		\begin{enumerate}
		\label{knownPST}
			\item 
				The complete graph $K_2$ with two vertices allow perfect state transfer between its vertices in time $t_0=\pi/2$ (in the units of energy inverse).
			\item
				The path graph $P_3$ has perfect state transfer between its end vertices  in time $t_0=\pi/\sqrt{2}$ (in the units of energy inverse).
			\item
				The hypercube of any order has perfect state transfer between its antipodal points in the same time $\pi/2$. And any order of Cartesian product of $P_3$ has PST between its antipodal vertices in the same time $\pi/\sqrt{2}$.
		\end{enumerate}	
Above three results hold both for the XY-coupling as we as the Heisenberg interaction. We make use of these results in order to establish a scalable and routing enabled quantum many-body network.

	\section{Perfect State Transfer from all-to-all nodes in two hoppings}
	
		Recall the known result, which will be applicable for proving the next lemma and its corollary. 
		\begin{lemma}
			Let $A = \diag\{B_1, B_2, \dots B_k\}$ be a block diagonal matrix, where $B_i$ are square matrices of arbitrary order for $i = 1, 2, \dots k$, then $\exp(A) =  \diag\{\exp(B_1), \exp(B_2), \dots \exp(B_k)\}$.
		\end{lemma}
		
		Let $G = (V(G), E(G))$ and $H = (V(H), E(H))$ be two graphs. The union of $G$ and $H$ is denoted by $G \cup H = (V(G \cup H), E(G \cup H))$ where $V(G \cup H) = V(G) \cup V(H)$ and $E(G \cup H) = E(G) \cup E(H)$ \cite{west2001introduction}.
		
		\begin{lemma}
			Let $G$ be a connected graph with perfect state transfer between two vertices $u$ and $v$ at time $\tau$. Also, let $H$ be a connected graph with perfect state transfer between two vertices $p$ and $q$, at time $\tau$. Then the graph $G \cup H$ has state transfer between $u$ and $v$ as well as $p$ and $q$, at time $\tau$.
		\end{lemma} 
	
		\begin{proof}
			Let the graph $G$ has $|V_1|$ vertices. The graph $G$ has perfect state transfer between the vertices $u$ and $v$ at time $\tau$. It indicates $|\braket{u|\exp(-i \tau A(G)) | v}| = 1$, where $\ket{u}$ and $\ket{v}$ are the state state vectors in $\mathbb{C}^{|V_1|}$ corresponding to the vertices $u$ and $v$, respectively. Similarly, if $H$ contains $|V_2|$ vertices, we have $|\braket{p|\exp(-i \tau A(G)) | q}| = 1$, where $\ket{p}$ and $\ket{q}$ are the state vectors in $\mathbb{C}^{|V_2|}$ corresponding to $p$ and $q$ respectively. We know that the graph $G \cup H$ has $|V_1| + |V_2|$ vertices. Corresponding to the vertices $u, v, p$, and $q$ define state vectors in $\mathbb{C}^{|V_1| + |V_2|}$ as $\ket{u'} = \begin{bmatrix} \ket{u} \\ (0)_{|V_2| \times 1}\end{bmatrix}, \ket{v'} = \begin{bmatrix} \ket{v} \\ (0)_{|V_2| \times 1}\end{bmatrix}, \ket{p'} = \begin{bmatrix} (0)_{|V_1| \times 1} \\ \ket{p} \end{bmatrix}$ and $\ket{q'} = \begin{bmatrix} (0)_{|V_1| \times 1} \\ \ket{q} \end{bmatrix}$, respectively. Note that
			\begin{equation}
				\begin{split}
					& A(G \cup H) = \begin{bmatrix} A(G) & (0)_{|V_1| \times |V_2|} \\ (0)_{|V_2| \times |V_1|} & A(H) \end{bmatrix}\\
					\text{or}~ & \exp(-i \tau A(G \cup H)) = \begin{bmatrix} \exp(-i \tau A(G)) & (0)_{|V_1| \times |V_2|} \\ (0)_{|V_2| \times |V_1|} & \exp(-i \tau A(H)) \end{bmatrix}. \\
				\end{split}
			\end{equation} 
			Now, 
			\begin{equation}
				\begin{split}
				|\braket{u' | \exp(-i \tau A(G \cup H)) | v'}| & = \Bigg|\bra{u'} \begin{bmatrix} \exp(-i \tau A(G)) & (0)_{|V_1| \times |V_2|} \\ (0)_{|V_2| \times |V_1|} & \exp(-i \tau A(H)) \end{bmatrix} \begin{bmatrix} \ket{v} \\ (0)_{|V_2| \times 1}\end{bmatrix} \Bigg|\\
				& =\Bigg| \begin{bmatrix} \bra{u} &  (0)_{1 \times |V_2|}\end{bmatrix} \begin{bmatrix} \exp(-i \tau A(G)) \ket{v} \\ (0)_{|V_2| \times 1} \end{bmatrix} \Bigg| \\
				& =| \braket{u | \exp(-i \tau A(G)) | v}| = 1.
				\end{split}
			\end{equation} 
			Similarly, $|\braket{p' | \exp(-i \tau A(G \cup H)) | q'}| = 1$. Therefore the graph $G \cup H$ has state transfer between $u$ and $v$ as well as $p$ and $q$, at time $\tau$.
		\end{proof}
		
		\begin{corollary}\label{state_transfer_with_isolated_vertices} 
			Let $G$ be a graph with state transfer between the vertices $u$ and $v$ at time $\tau$. Then the graph $G \cup \{v_1\}\cup \{v_2\} \dots \cup \{v_k\}$ has perfect state transfer between $u$ and $v$, where $v_1, v_2, \dots v_k$ are isolated vertices at time $\tau$.
		\end{corollary}				
	
		\begin{proof}
			Proof follows trivially.
		\end{proof}
	
		\begin{definition}
			Hopping: In our formalism, the hopping on a graph $G = (V(G), E(G)$ is a combination of a Classical (C), a true Quantum hopping (Q), and again a Classical process (C) which consists of the following steps:
			\begin{enumerate}
				\item 
					\textbf{Switch off the selected edges (Classical):} Construct a subgraph $H = (V(H), E(H)$ of $G$ such that $V(H) = V(G)$ and $E(H) \subset E(G)$. Let $S = E(G) - E(H)$.
				\item 
					\textbf{Perfect state transfer (Quantum):} Perform quantum operations on $H$, such that, the quantum state can be transferd from vertex $u$ to vertex $v$ in the graph $H$.
				\item 
					\textbf{Switch on the edges (Classical):} Construct the graph $G$ form the graph $H$ by adding the edges from $S$ in the graph $G$. 
			\end{enumerate}
		\end{definition}
		We call these processes together as CQC-hopping. So, CQC-hopping is a process of a classical switching, followed by true quantum evolution (hopping), followed by another switching of edges. Now we have the following important result.
    \subsection{Perfect State Transfer in hypercubes from all-to-all vertices in single CQC-hopping}
    In this section, we prove that, given any hypercube of any dimension, we enable the perfect state transfer from all-to-all vertices of the graph. This is in contrast to the main result of \cite{ref:15} presented in literature section \ref{sec:hypercubetransfer} where only pair of antipodal vertices are PST enabled. This is accomplished with the edge switching.
		\begin{theorem}\label{perfect_state_transfer_in_hypercube}
			The perfect state transfer is possible between any two vertices of a hypercube $Q_k$ by a single hopping CQC process. 
		\end{theorem}
		
		\begin{proof}
			Consider any two vertices $u$ and $v$ in the hypercube $Q_k$ which are labeled by the binary sequences $\bin(u) = \{u_j\}_{j = 0}^{(k - 1)}$ and $\bin(v)  = \{v_j\}_{j = 0}^{(k - 1)}$, respectively. If $u$ and $v$ are the antipodal points of the hypercube then there is a state transfer between $u$ and $v$. In this case, we need no edge to switch on or switch off. 
			
			Let $u$ and $v$ are not the antipodal points of the hypercube $Q_k$. Then there is a sequence of indices $\{m_t\}$, such that, $0 \leq m_1 < m_2 < \dots < m_{(k - i)} \leq (k - 1)$ and $u_{m_t} = v_{m_t} = M_t$ holds. Corresponding to the set of indices $\{m_t\}$ and the binary sequence $M$ construct a set of vertices $V_i = \{v: \bin(v) = (v_j)_{j = 0}^{(k - 1)}, v_{m_t} = M_t ~\text{for}~ t = 1, 2, \dots (k - i) \} \subset V(Q_k)$. Denote the induced subgraph of $Q_k$ generated by $V_i$ as $G(V_i)$. Using lemma \ref{projection_lemma}, we find that $G(V_i)$ is isomorphic to the hypercube $Q_i$.
			
			Now we perform a CQC hopping on $Q_k$. First, we switched off all the edges which are not included in the induced subgraph $G(V_i)$ that is in $E(Q_k) - E(G(V_i))$. The new graph can be expressed as $H = Q_i \cup (V(Q_k) - V_i)$, where $(V(Q_k) - V_i)$ denotes the set of isolated vertices. The considered vertices $u$ and $v$ are the antipodal points of $G(V_i)$ or $Q_i$. Hence, there is a perfect state transfer between $u$ and $v$. By corollary \ref{state_transfer_with_isolated_vertices} we find that there is state transfer between the vertices $u$ and $v$ in the graph $H$. After state transfer we switched on the edges in $E(Q_k) - E(G(V_i))$. In this way we can transfer the sate between any two vertices of a hypercube $Q_k$.
		\end{proof}
		
		This theorem suggests the number of edges to switch off and switch on in CQC process. When $u$ and $v$ are two vertices belonging to the hypercube $Q_k$, following this theorem we construct a sub-hypercube $Q_i$ which is essential for the state transfer. Now $Q_k$ has $k 2^{k - 1}$ edges and $Q_i$ has $i 2^{i - 1}$ edges. Therefore, we need to switched off $k 2^{k - 1} - i 2^{i - 1} = 2^{i - 1}(k2^{k - i} - i)$ edges, in the hypercube $Q_k$. We can also present this result as algorithm \ref{algo_state_transfer_in_hypercube}.
		
		\begin{algorithm}
			\caption{Find the sub-hypercube for perfect state transfer between two arbitrarily chosen vertices.}\label{algo_state_transfer_in_hypercube}
			\begin{algorithmic}
				\REQUIRE Vertices $u$ and $v$ of $Q_k$ for state transfer.
				\ENSURE $\bin(u) = \{u_j\}_{j = 0}^{(k - 1)}$ where $u_j \in \{0, 1\}$ for all $u \in V(Q_k)$.
				\IF{$u_j \neq v_j$ for all $j = 0, 1, \dots (k - 1)$}
					\STATE Perfect state transfer between $u$ and $v$.
					\RETURN $Q_k$
				\ELSE
					\STATE Count = 0 
					\STATE m = [] \COMMENT{List of indices $m_t$ such that $u_{m_t} = v_{m_t}$.}
					\STATE M = [] \COMMENT{List of common indices in $\bin(u)$ and $\bin(v)$.}
					\FOR{$j \gets 0$ to $(k - 1)$} 
						\IF{$u_j = v_j$}
							\STATE Count = Count + 1
							\STATE m.insert(j)
							\STATE M.insert($u_j$)
						\ENDIF 
					\ENDFOR 
					\STATE Construct $V(Q_i) = \{v: \bin(v) = (v_j)_{j = 0}^{(Count - 1)}\}$ \COMMENT{Here the vertices are labelled by the binary sequences of length $i$.}
					\STATE $G = (V(G), E(G))$.
					\FOR{$v \in V(Q_i)$}
						\STATE $v: \bin(v) = (v'_j)_{j = 0}^{(k - 1)}$ \COMMENT{We want to label the vertices by the binary sequences of length $k$.}
						\FOR{$j \gets 0$ to $(k - 1)$}
							\STATE $C = 0$
							\IF{$j = m_t$}
								\STATE $v'_{m_t} = M_t$
								\STATE $C = C + 1$
							\ELSE 
								\STATE $v'_j = v_{j + C}$
							\ENDIF
						\ENDFOR  
					\ENDFOR 
					\STATE $V(G) = V(Q_i)$.
					\STATE $E(G)$ = Set of edges in the induced subgraph of $Q_k$ generated by $V(Q_i)$ 
					\STATE Perfect state transfer between $u$ and $v$.
					\RETURN $G$
				\ENDIF
			\end{algorithmic}
		\end{algorithm}
	
		\begin{example}
			Consider two arbitrary verticees $u$ and $v$ in $Q_8$ where $\bin(u) = 00101101$ and $\bin(v) = 10011000$. Note that, $u_j = v_j$ for $j = 1, 4$ and $6$ that is $m_1 = 1, m_2 = 4$, and $m_3 = 6$, as well as $M_1 = 0, M_2 = 1$ and $M_3 = 0$. Following lemma \ref{projection_lemma} we construct a set of vertices $V_5$ and an induced subgraph $H(V_5)$ of $Q_8$ which is isomorphic to $Q_5$. The nodes $u$ and $v$ correspond two antipodal points in $H(V_5)$. To make state transfer between $u$ and $v$ we switched off and switched on all the edges in $E(Q_8) - E(H(V_5))$.
		\end{example} 


	\subsection{Constructing large graphs supporting state transfer between arbitrary nodes in two CQC-hoppings}
		
		Given any natural number $n$ there is a natural number $k$, such that, $2^k \leq n < 2^{(k + 1)}$ and $n$ has a $(k + 1)$-term binary representation. For any natural number $n$ there is a graph $G$ which allows perfect state transfer between any two nodes in two CQC-hoppings. The graph $G$ can be constructed as follows.
		
		\begin{procedure}\label{construction_of_two_CQC_graph_for_any_given_node}
			\textbf{Constructing graphs allowing perfect state transfer:} Let the natural number $n$ can be written as $n = a_{p_0}2^{(k - p_0)} + a_{p_1} 2^{k - p_1} + a_{p_2} 2^{k - p_2} + \dots$, where $p_0 = 0$ and $a_{p_1} = a_{p_2} = \dots =1$.
			\begin{enumerate}
				\item 
					Use $2^k$ vertices to construct $Q_k$ and label its nodes $v$ with a $(k + 1)$-term sequence $v = (v_j)_{j = 0}^k$ where $v_j \in \{0, 1\}$ and $v_0 = 0$.
				\item 
					Now fix $v_0 = 1$ for all remaining constructions. 
				\item 
					For every $p_i$ with $0 < p_1 < p_2 < \dots $ construct a hypercube $Q_{k - p_i}$ with $2^{k - p_i}$ nodes. First label the vertices with a $(k - p_i)$-term sequence $u = (u_j)_{j = 0}^{k - p_i - 1}$ where $u_j \in \{0, 1\}$. A vertex can be included in at most one hypercube.
				\item
					Re-label the vertices of $Q_{k - p_i}$ with a $(k + 1)$-term sequence $v = (v_j)_{j = 0}^k$ where $v_j \in \{0, 1\}$ with $v_0 = 1, v_j = 0$ for $j = 1, 2, \dots (p_i + 1)$ and $v_{p_i + 1 + j} = u_j$ for $j = 0, 1, \dots (k - p_i - 1)$. 
				\item 
					Join pairs of nodes $u$ and $v$ with an edge where $u \in V(Q_{k - p_i})$ and $v \in V(Q_k)$ as well as hamming distance between their labeling is $1$.
				\item 
					Fix $v_{p_i + 1} = 1$ for all remaining constructions.
				\item 
					Similarly, relabel all other hypercubes and create edges between them.
			\end{enumerate}
			Note that, a single node is a hypercube of dimension $0$.
		\end{procedure}		
	
		\begin{lemma}
			Let $n = 2^k + a_1 2^{k - 1} + a_2 2^{k - 2} + \dots a_k 2^0$ where $a_1, a_2, \dots a_k \in \{0, 1\}$ and there are integers $p_1 < p_2 < \dots $. such that $a_{p_1} = a_{p_2} = \dots = 1$. Then, the graph $G$ with $n$ vertices which is constructed by following the procedure \ref{construction_of_two_CQC_graph_for_any_given_node} has $T(n)$ edges, where
			$T(n) = \sum_{p_i} \left[ (k - p_i)2^{(k - p_i - 1)} + i \times 2^{(k - p_i)} \right].$
		\end{lemma}
		
		\begin{proof}
			Recall that any hypercube $Q_k$ has $k 2^{k - 1}$ edges. For every $p_i$, with $a_{p_0} = a_{p_1} = a_{p_2} = \dots = 1$ we construct a hypercube $Q_{k - p_i}$. Therefore we add $\sum_{a_{p_i} \neq 0} (k - p_i)2^{(k - p_i - 1)}$ edges in the graph $G$.
			
			Consider a vertex $u \in Q_{(k - p_1)}$. As a vertex of $G$, the labeling of $u$ is given by the sequence $\bin(u) = (u_j)_{j = 0}^k$. Note that $u_0 = 1$. Therefore, there is exactly one vertex $v = (v_j)_{j = 0}^k$ in $Q_{2^k}$ such that, $v_0 = 0$ and $v_j = u_j$ for $j = 1, 2, \dots k$. Clearly the Hamming distance between $\bin(u)$ and $\bin(v)$ is $1$. Hence, $u$ and $v$ are adjacent. Therefore every vertex of $Q_{(k - p_1)}$ is adjacent to only one vertex of $Q_k$. Therefore, there are $2^{(k - p_1)}$ edges $(u, v) \in E(G)$ such that $u \in V(Q_k)$, and $v \in V(Q_{(k - p_1)})$.
			
			In a similar fashion, we can justify that there are $2^{k - p_2}$ edges $(u, v) \in E(G)$ such that $u \in V(Q_k)$, and $v \in V(Q_{(k - p_2)})$. In addition, there are another $2^{k - p_2}$ edges $(u, v) \in E(G)$ such that $u \in V(Q_{(k - p_1)})$, and $v \in V(Q_{(k - p_2)})$. Adding we get there are $2 \times 2^{k - p_2}$ edges whose one end vertex in $V(Q_{(k - p_2)})$ and another end vertex is either in $V(Q_k)$ or in $V(Q_{(k - p_1)})$.
			
			Extending we get, there are $3 \times 2^{k - p_3}$ edges whose one end vertex in $V(Q_{(k - p3)})$ and another end vertex is either in $V(Q_k)$ or in $V(Q_{(k - p_1)})$ or in $V(Q_{(k - p_2)})$. Collecting all the edges we get the total number of edges, which is mentioned in the statement.
		\end{proof}
	
		\begin{theorem}
			Let $u$ and $v$ be two vertices in the graph $G$ of order $n$ which is constructed following the procedure \ref{construction_of_two_CQC_graph_for_any_given_node}. Then there is a perfect state transfer between $u$ and $v$ by at most two CQC-hoppings. 
		\end{theorem}
	
		\begin{proof}
			According to the construction procedure there are hypercubes containing the vertices $u$ and $v$. If $u$ and $v$ belongs to same hypercube then by theorem \ref{perfect_state_transfer_in_hypercube} state transfer from $u$ to $v$ is possible by one CQC-hopping only.
			
			Let $u$ and $v$ belong two two different hypercubes. The procedure \ref{construction_of_two_CQC_graph_for_any_given_node} indicates that no two hypercubes can be of equal size. For simplicity let $u \in Q_p$ and $v \in Q_q$ where $p > q$. Now corollary \ref{induced_hypercube} suggests that $Q_p$ has an induced subgraph isomorphic to $Q_q$. Therefore, there is a vertex $w$ in $Q_p$ which is equivalent to $v$ in $Q_q$.
			
			The vertices $u$ and $v$ are represented by binary sequences of length $(k + 1)$. Note that, the binary sequences representing $v$ and $w$ have Hamming distance $1$. Therefore, there is an edge $(v, w)$ in the graph $G$.
			
			Now the first CQC-hop allows state transfer from $v$ to $w$. The second CQC-hop allows state transfer from $w$ to $u$. Hence, the proof.  
		\end{proof}
		
		This theorem suggests the number of edges to switch off and switch on in CQC-hopping process. Let $u$ and $v$ belong to two different hypercubes. For first CQC we need only one edge. Therefore, we need to switched off $T(n) - 1$ edges. For the second state transfer we need the hypercube $Q_q$. Recall that number of edges in $Q_q$ is $q 2^{q - 1}$. We need to switched on only these edges. When the state transfer is done (state has been recovered at the desired vertex), we shall switched on all the edges in $G$ as desired.
			
		Construction of this graph can be seen as a growing network. This tool allows us for constructing a graph $G_1$ with an additional node from a given graph $G$ allowing state transfer in two-CQCs. The motivation for such an argument is due to the fact that experimentally only a few number of qubits are added with every technological improvement in quantum technologies, therefore, suggesting the need of an architecture which allows the growth as addition of one qubit each time.

    \subsection{Perfect State Transfer in growing network supporting all-to-all transfer in two CQC-hoppings}
    In the following procedure, we propose the growing network architecture where each new qubit can be added to the existing network. It suffices to start with a $2^k$-vertex hypercube and grow it to the next hypercube of $2^{k+1}$ vertices. 
	\begin{procedure}
			\label{procedure2}
			\textbf{Growing network method}
			Let $Q_k$ be a hypercube with $2^k$ vertices which are labeled by all possible binary sequences of length $k$. Follow the steps below to construct a graph $G$ with number of nodes $< 2^{k + 1}$ allowing state transfer with two CQC-hoppings. Note that, we can add at most $(2^k - 1)$ new vertices with $Q_k$ to construct the graph $G$. 
			\begin{enumerate}
				\item 
					Relabel all the vertices of $Q_k$ with binary sequences of length $(k + 1)$ such that the left-most element of the sequence is $0$ and the others are equal to the old labeling of $Q_k$.
				\item 
					The vertex labeling of the $l$-th new vertex will be given by a $(k+1)$ term binary sequence initiated by $1$. The remaining $k$-terms are the $k$-term binary representation of $l$. 
				\item 
					The new vertex will be adjacent to all other vertices with Hamming distance $1$. 
			\end{enumerate}   
		\end{procedure}
		Each application of this procedure will increase the number of vertices in the initial graph $Q_k$ for any integer $k > 0$. In each step the new graph $G$ will allow perfect state transfer in two CQCs.

		\textbf{Equivalent definition:} The above procedure can be see in terms of binary addition (denote it as $\oplus$) where the one bit is carried to the left. For any $n$ in the range given above, the first node is $|10...00\rangle$, which is $(2^k+1)$th node. For all subsequent nodes just add a 1 via binary addition to the existing labeling of the last added node with any possible carry to the left. Hence, the next node will be $|10...00\oplus 1\rangle=|10...01\rangle$. The next node will be $|10...01\oplus1\rangle=|10...10\rangle$ where one carry is taken over to the second index from the right and so on until the desired vertex. The second last node will be $|11...110\rangle$ and adding 1 further to the first index will give give the last node $|111...11\rangle$ which completes the $2^{k+1}$-vertex hypercube. For further addition of nodes simply append one more index and start the same set of steps. Therefore, it can easily be seen that any $(2^k+1+m)$th node is simply $|10...00\bigoplus^{m-1}_{j=1}1\rangle$, where $2^k+m=n$, and this determines the full connectivity of any new node with the existing graph, which defines a simple graph upon addition of every new node.

\section{Edge minimization vs. hopping minimization}
In our CQC-hopping scheme we have given the optimal number of hoppings required to transfer a state from any vertex to any other vertex in the graph for arbitrary number of total vertices. At least one of the two CQC-hoppings for a dual hop is a hypercube jump, where the involved hypercube maybe of a large dimension. This state transfer will be ensured by ensuring the correct set of edges where the total number of edges are indeed large. If the cost of creating a large number of edges is high, experimentally, then it may pose a new problem to the architecture. It can be asked whether there exists another architecture which minimizes the number of edges instead of CQC-hoppings. Ideally, the least number of edges are there for the case of a path graph P$_{\text{n}}$ for $n$ vertices. But then the question of finding a network is trivial. Therefore, we need a trade-off between both of these extremes because creating and manipulating a large number of edges will also leave the system prone to more error in state transfer resulting in low fidelity.

A trade-off between both of these extremes can be can be discussed for a neighborhood zone for two given vertices. The possible moves allowed under two CQC-hopping are: one-link jump $(t_0=\pi/2)$ and the two-link jump $(t_0=\pi/\sqrt{2})$. So, the maximum distance that can be reached using these possible moves for two CQC-hoppings is the four corresponding to all the fourth nearest vertices of any chosen vertex. The farthest fourth nearest adjacent vertices can be reached in time $\sqrt{2}\pi$ using a total of four edges only. In contrast, had these two vertices separated by a distance of 4 were the antipodal vertices, we would have to create $4\times 2^3=32$ edges in total. However, this would perform the task in single CQC-hopping in time $\pi/2$. Depending upon experiment to experiment, various constraints maybe posed on the cost of edges and the cost of CQC-hoppings and state transfer time. For our physical implementation proposed in chapter \ref{chap:phyreal}, the edges can be easily constructed and destructed and we follow our scheme for optimal number of hoppings and transfer time.

We can classify the set of vertices that can be reached from a vertex $|x\rangle$ by partition into the following sets. These sets also do not contain any vertices that are missing due to an incomplete hypercube, we assume all these vertices to not be in these sets by construction. For any arbitrary $2^k\leq n <2^{k+1}$, we have	
	\begin{itemize}
	    \item \textit{Set of all adjacent vertices:} Any vertex has some $i$th adjacent vertex as the $i$th flipped index. So, the set of all these adjacent nodes is $\{|\alpha_i^x\rangle \equiv |x_0,x_1,...,x_i\oplus 2,...,x_n\rangle\}$ $\forall i\in \{0,1,...n\}$. Total number of such distinct nodes at most is $k+1$. Each of these nodes can be reached by one unique path.
		\item \textit{Set of all next-adjacent vertices:} Set of next-adjacent nodes can be reached in one CQC-hopping by the two-link jump in time $\pi/\sqrt{2}$. All available vertices accessible via $|\overline{\alpha^x_i}\rangle$ except trivially $|x\rangle$ itself. Hence this is the set $\{ |\beta^x_{ij}\rangle \equiv |x_0,x_1,...,x_j\oplus 2,...,x_i\oplus 2,...,x_k\rangle \}$ $\forall i,j\in \{0,1,...,k\}$ $\&$ $i\neq j$. Total number of such distinct nodes is $(k+1)k/2$ at most. Each of these vertices can be reached by two unique paths if no hypercube jump is involved and all adjacent vertices are available relative to a complete $2^{k+1}-$vertex hypercube.
		\item \textit{Set of all next-to-next adjacent vertices:} These can be reached by one-link jump followed by a two-link jump or vice versa with time $\pi/2+\pi/\sqrt{2}$. Using the previous argument it is simply the set defined as $\{ |\gamma^x_{ijp}\rangle \equiv |x_0,x_1,...,x_p\oplus 2,...,x_j\oplus 2,...,x_i\oplus 2,...,x_k\rangle \}$ $\forall i,j,p\in \{0,1,...,k\}$ $\&$ $i\neq j\neq p$. Total number of such distinct nodes is at most $(k+1)k(k-1)/6$. Each of these nodes can be reached by three unique paths if no hyperjump is involved and all the adjacent vertices are available relative to the complete hypercube.
		\item \textit{Set of all next-to-next-to-next adjacent vertices:} These can be reached by two-link jump followed by another two-link jump with time $\sqrt{2}\pi$. It is the set defined as $\{ |\delta^x_{ijpq}\rangle \equiv |x_0,x_1,...,x_q \oplus 2,...,x_p\oplus 2,...,x_j\oplus 2,...,x_i\oplus 2,...,x_k\rangle \}$ $\forall i,j,p,q\in \{0,1,...,k\}$ $\&$ $i\neq j\neq p \neq q$.  Total number of such distinct nodes is at most $(k+1)k(k-1)(k-2)/24$. Each of these nodes can be reached by four unique paths if no hyperjump is involved if all vertices are available.
		\end{itemize}
This lets us minimize the number of edges for small distance of 4 between any two given vertices.
		
Another argument from the physical point of view is for the decoherence error. During the quantum evolution, the system is most prone to error. In contrast, the edge deletion and creation are purely classical phenomenon and do not introduce any direct error quantum mechanically into the system. Minimising this evolution time by minimising the number of CQC-hoppings will be more robust approach for the implementation for state transfer.

\section{Switching and global evolution of the graph}
\label{sec:globalevo}
For any pair of two chosen nodes for arbitrary number of nodes $n\neq 2^k$ for the graph, we have two subgraphs corresponding to the action-space of two adjacency matrices $A_1$ and $A_2$ that describe the first and second hop respectively over the entire graph as described before. We perform the first CQC-hop from some vertex $|u\rangle$ to a vertex $|v\rangle$ keeping all edges switched off that belong to the second hop, followed by another CQC-hop from $|v\rangle$ to $|w\rangle$ while all edges corresponding to first hop are switched off. So, both the adjacency matrices have one common element of action which is the intermediate vertex $|v\rangle$ such that $A_1|v\rangle\neq 0\quad \& \quad A_2|v\rangle\neq0$. This means we act with $\exp(-iA_1t_0)$ followed by $\exp(-iA_2t_0)$ over the entire graph as continuous quantum walk. $A_1$ and $A_2$ corresponds to the first and second CQC-hop adjacency matrices respectively, resulting from the switching. Is it possible to capture these two unitary evolution as one single unitary evolution? In the light of Baker–Campbell–Hausdorff (BCH) formula, let us look at the commutator $[A_1,A_2]$ for all possible vertices on this graph. We may classify all nodes in three sets: All nodes $|u'\rangle$ which can be reached by the repeated action of $A_1$ (except $|v\rangle$), the intermediate node $|v\rangle$ itself and all the nodes $|w'\rangle$ that can be reached by the repeated action of $A_2$ (except $|v\rangle$). Any arbitrary chosen initial graph state linear combination of all of these elements (which act as a basis for the action space of the two adjacency matrices), at most. Let us evaluate the action of the commutator on a general state (un-normalised)
\begin{equation}
    [A_1,A_2](|u\rangle +|v\rangle +|w\rangle)=(A_1A_2-A_2A_1)(|u\rangle +|v\rangle +|w\rangle
\end{equation}
\begin{equation}
    =\sum_{\{u'\},\{w'\}}\left( \alpha_{u'}|u'\rangle +\alpha_{w'}|w'\rangle \right) +\alpha_v|v\rangle
\end{equation}
which is not necessarily zero, except for a special initial state. Here, $\alpha_\lambda \in \mathbb{Z}$ are the coefficients. Therefore, $A_1$ and $A_2$ do not commute in general and hence we cannot write down a general evolution $\exp (-iA_1t)\cdot\exp( -iA_2t)=\exp \left[ -i(A_1+A_2)t\right]$ and it implies they have to be treated strictly as non-simultaneous unitary evolution in time which is also expected otherwise as it is a strictly two step task, which does not allow the leakage of state out of the action space of one adjacency matrix.

Generally, for the part of our protocol which requires a dual CQC-hop when we are away from a perfect hypercube, let us define the overall adjacency matrix as a step-function in time as follows:
\begin{equation}
\label{eqn:adjtwohop}
    A(G(t))=
    \begin{cases}
    A_1, \quad \text{for } 0< t\leq t_0,\\
    A_2, \quad \text{for } t_0< t\leq 2t_0
    \end{cases}
\end{equation}
Hence, for any given pair of nodes we quickly identify $A_1$ and $A_2$ and perform the PST from all-to-all nodes in at most time $2t_0$. At time $t_0$ we perform the switching and the adjacency matrix of the graph is switched from $A_1$ to $A_2$. In the ideal scenario, it is assumed to take no time which is realistically not the case. For those cases and to quantify the error introduced in the evolution due to switching by approximating the step function by some function close to it and see how much error it introduces into the fidelity. Switching will introduce error in the quantum system as it will be probed by a classical system to switch the couplings which is close to a classical task.

\section{Quantum computation assisted with perfect state transfer}
The original idea \cite{ref:1} of perfect state transfer was to transport a state from one location to another using a channel like mechanism which is non-computational for quantum computing usage. For example, a spin lattice channel fabricated between two quantum processors can be used to transport arbitrary state from one processor to another for large scale quantum architecture. So, perfect state transfer has been seen as a means of communication between two or more quantum computational units. This idea has been experimentally demonstrated in \cite{ref:47}. The implication of such as idea involves transferring state from one specific qubit from the first processor to a specific qubit in the second processor. Within each processing unit, conventional quantum SWAP gates have to be used transport the state from one qubit to another (given that we do not care about the state of the other qubit). In fact, a large amount of cascaded SWAP operations have to be used in order to accomplish this task within each processor, if the qubits are distant, assuming one SWAP gate is one operation quantum mechanically. 

The possibility of application of two-qubit gates between for given qubits depends upon whether these two qubits are allowed to interact, that is, that are coupled by some coupling Hamiltonian which defines an exchange interaction. Two qubits which are not directly coupled to one another cannot have two-qubit interactions directly. This means we cannot have two qubit gates, such as SWAP or CNOT, acting between these qubits. This connectivity on the network of qubits is what we impose by defining a graph scheme. Hence, in a quantum circuit, we cannot have two-qubit quantum gates between the qubits which are physically distant on the real processor. In theoretical quantum circuits, we assume that two qubit interactions are possible for every pair of qubits, but the disconnectedness imposes even more number of gates to be used for the same task. When two qubits are distant and cannot interact directly, we may have cascaded two qubit gates. For example, on a linear chain we are allowed to perform SWAP operation for every $(i,i+1)$ qubit pair. To transfer a state $|1\rangle$ from one end to another we have to apply $(l-1)$ SWAP gates, where $l$ is the length of the chain, and all other qubits are intialized to $|0\rangle$. Similarly, we have to use a large number of cascaded SWAP gates for to perform state transfer to distant qubits. We deduce that the minimum number of SWAP gates required for the state transfer is equal to the minimum edge distance between two given vertices.

In this section we aim to present perfect state transfer as a quantum operation in the conventional quantum circuit model. In chapter \ref{chap:phyreal}, we present a physical model which serves as high fidelity quantum processor along with the ability of perfect state transfer all over the computational qubits. This model serves as the physical implementation of our scalable architecture proposed in this section. Having realised in chapter \ref{chap:phyreal} that same quantum computing hardware can be used for perfect state transfer, we look at perfect state transfer from a another perspective of a quantum gate. We propose the idea of quantum computing assisted with perfect state transfer on graphs.

We can take the advantage of perfect state transfer as a SWAP gate operation. Consider a quantum circuit with qubit $u$ in some state $|\psi\rangle$ while all other qubits are initialized to zero. Notice the action of a CQC-hop state transfer from a given qubit at site $u$
\begin{equation}
    e^{-iA_1t_0}|u\rangle_\psi=e^{i\phi}|v\rangle \xrightarrow{\text{Phase correction}(\varphi_1)}|v\rangle_\psi
\end{equation}
for some arbitrary phase $\phi$, which can be corrected post the state transfer. Furthermore, if we are not at a perfect hypercube we require another state transfer, due to equation (\ref{eqn:adjtwohop}), we have
\begin{equation}
    e^{-iA_2t_0}|v\rangle_\psi=e^{i\phi}|w\rangle\xrightarrow{\text{Phase correction}(\varphi_2)}|w\rangle_\psi
\end{equation}
which takes us to a final vertex $w$. Each of these CQC-hoppings' action is equivalent to the SWAP gate between the qubits at those vertices.

The reported time for a iSWAP gate in \cite{ref:14} is reported to be around $\tau=35$ ns (differs for different nature of decoherence errors incorporated in the model. However, to show the advantage of CQC-hop state transfer scheme, only the order is important). For the same values of expermental data in \cite{ref:14}, we show that the total time $2t_0$ for our scheme is of the order of 1.2 ns only. Moreover, SWAP and CNOT gates on a physical architecture are not a single step operations, but usually involves a series of steps with time controlled quantum evolution with various parameters of control (such as the famous Cirac-Zoller CNOT gate in trapped ion \cite{ref:35} and similar implementation in C-QED Jaynes-Cummings interaction \cite{ref:23}). However, for the purpose of quantum circuits, all two-qubit gates are seen as single operation for the circuit. In contrast to a SWAP gate, a perfect state transfer requires less time if parameters are chosen correctly, for our architecture (see chapter \ref{chap:phyreal}). In contrast, perfect state transfer is simply the time controlled free quantum evolution of the whole network and is more robust to errors and decoherence because there is no quantum manipulation into the network during the process.
\begin{figure}[hbtp]
    \centering
    \includegraphics[scale=0.5]{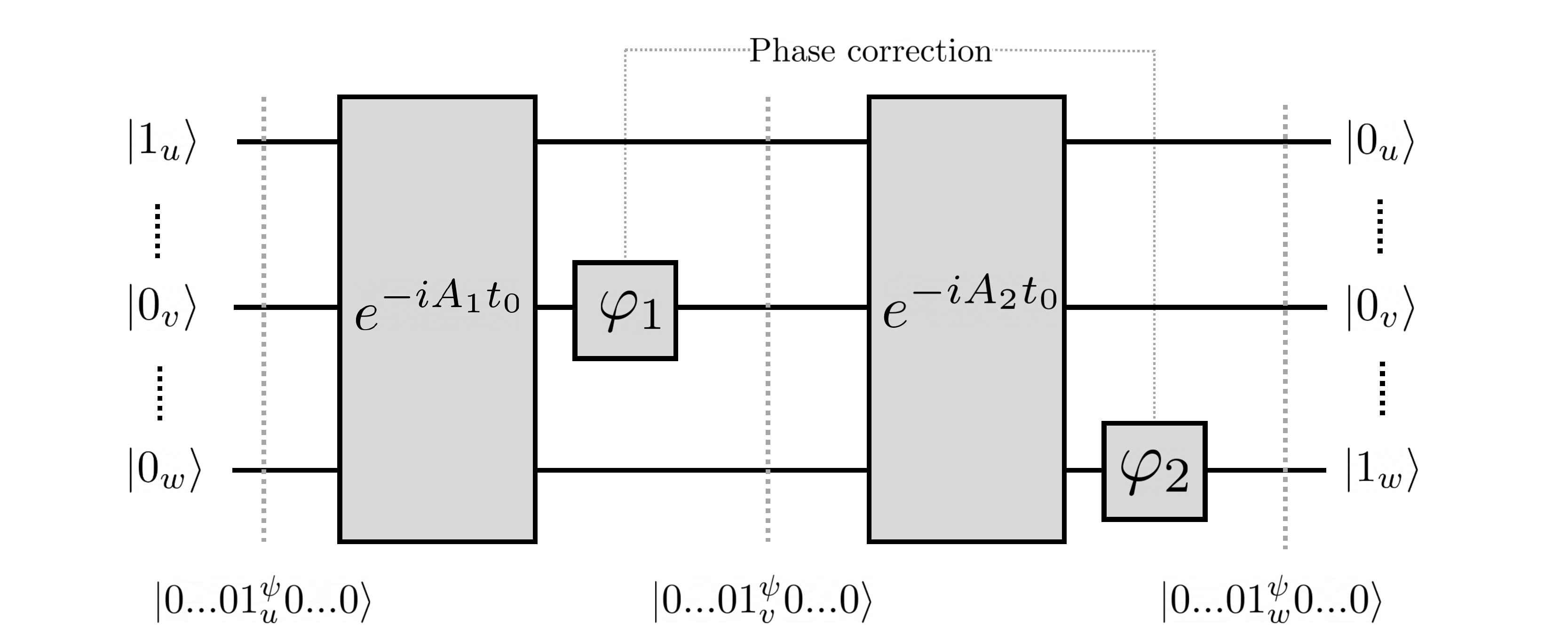}
    \caption{Quantum circuit for our CQC-hopping scheme}
    \label{fig:qcircuit}
\end{figure}

For a $d$-dimensional hypercube, the shortest path length between two given vertices is at most $d$. Hence, at most $d$ SWAP gates will be required to transfer a state between these two vertices. Whereas CQC-hopping can make this task possible in at most 2 operations for arbitrarily chosen qubits in the network. For example, consider the case of $\square^5 K_2-\{|11111\rangle\}$ (which is not a perfect hypercube with 31 qubits). The vertices can be labeled as $|x_0x_1x_2x_3x_4\rangle$. Let us say we want to the state transfer from $|
10100\rangle$ ($u$) to $|01011\rangle$ ($w$). Notice that they are not antipodal because the vertex $|11111\rangle$ is missing. There are multiple shortest paths for this given pair. One of these is
\begin{equation}
    |10100\rangle_u \rightarrow |00100\rangle \rightarrow |01100\rangle \rightarrow |01110\rangle \rightarrow |01010\rangle \rightarrow |01011\rangle_w.
\end{equation}
This requires us to use five SWAP gates in this desired sequence. Whereas for the CQC-hopping scheme we have
\begin{equation}
    |10100\rangle_u\xrightarrow{e^{-iA_1t_0}}|00100\rangle_v\xrightarrow{e^{-iA_2t_0}} |01011\rangle_w  
\end{equation}
which is just a two step task.  For any order of hypercube and any large arbitrary large number of qubits $n$, this holds. In this example, the first CQC-hopping is one-link jump whereas the second CQC-hopping is a hypercube (of dimension 4) jump. Hence, the computational advantage. Refer to figure \ref{fig:qcircuit} for the quantum circuit for CQC-hopping scheme for general two hoppings. If it so happens that $u$ and $w$ are antipodal vertices of a hypercube then only one hopping is required. Phase correction gates  $\varphi_1$ and $\varphi_2$ are applied to recover the original state $|\psi\rangle$ at the desired qubits. This phase correction can be easily performed at any given qubit when our state transfer qubits are themselves the computational qubits (see chapter \ref{chap:phyreal}). Note that the ket states in \ref{fig:qcircuit} denote the excitation states at different qubits and this is different from the binary ket representation of the same. Therefore, whenever a quantum circuit has evaluated a segment of circuit and is brought to halt with a result stored in one of qubit's state, it can be easily transferred to any other qubit purely by our state transfer scheme in contrast to multiple SWAP gates. The only condition that needs to be ensured is to initialize all other qubits to ground state except the one we want to transfer.


%% file: Chapters/QuditsPST.tex
Instead of a qubit (a two-level quantum system), we may require to transfer the state of a $d-$level system, called a qudit. If a qudit is thought of as stacked qubits with equally spaced levels (like finite harmonic oscillator), then there is a finite possibility that the $d-$level excitation can \textit{split} and distribute as a sum of lower order excitations (such as different $\Gamma^k$ in section \ref{sec:firstexcitation}) which splits this qudit state while preserving the total excitation. However, we want to perfectly transfer the full qudit state from one site to another in the network (graph) of qudits. This requires special conditions. We want to transfer any general qudit state from one vertex to another. We denote a general qudit state $|\psi_d \rangle$ as
\begin{equation}
    |\psi_d \rangle=\sum_{j=0}^d\alpha_j|j\rangle
\end{equation}
with $\alpha_i\in \mathbb{C}$ and $\sum_{j=0}^d|\alpha_j|^2=1$, as the normalization condition.

\section{Perfect State Transfer in weighted qudit chains of arbitrary length}
Similar to equation (\ref{eqn:XYmodel}), the transfer Hamiltonian for the $XY$ model can be formulated using Lie algebra. This was proposed in \cite{ref:20} for chains. Analogous to the Pauli operators for $SU(2)$ group for qubit, Lie algebra generators can be defined for $SU(d)$ group for the $d-$level system transfer dynamics. First step is to define the corresponding projectors
\begin{equation}
    (P^{k,j})_{\mu,\nu}=|k\rangle \langle j|=\delta_{\mu,j}\delta_{\nu,k}, \quad 1\leq \mu, \quad  \nu\leq d
\end{equation}
We want to obtain all the $d^2-1$ operators for $SU(d)$ group. The first set of these are the projectors
\begin{equation}
    \Theta^{k,j}=P^{k,j}+P^{j,k}, \quad \beta^{k,j}=-i(P^{k,j}-P^{j,k}),
\end{equation}
with $1\leq k \leq j \leq d$. And the remaining $d-1$ operators are defined to be the following,
\begin{equation}
    \eta^{r,r}=\sqrt{\frac{2}{r(r+1)}}\left[ \left( \sum_{j=1}^r P^{j,j}\right)-rP^{r+1,r+1} \right]
\end{equation}
with $1\leq r\leq (d-1)$. These two sets give all the $d^2-1$ operators for the $SU(d)$ group that are necessary to define the dynamics for qudit chains. For $d=3$, these sets give the Gell-Mann matrices, and so on. For defining the dynamics for the $XY$ model, only the off-diagonal operators of the first set are enough.

For defining the qudit network we assume that all the qudits are identical. Then we define the $XY$ Hamiltonian as
\begin{equation}
    H^d_{XY}=\sum_{(i,i+1)\in \mathcal{L}(G)}\frac{J_i}{2}\left( \Theta^{k,j}_{(i)}\Theta^{k,j}_{(i+1)}+\beta^{k,j}_{(i)}\beta^{k,j}_{(i+1)} \right)
\end{equation}
where $\mathcal{L}(G)$ denotes a line graph of $G$. And the couplings are weighted couplings with $J_i=\sqrt{i(n-i)}/2$ for $1\leq k \leq j \leq d$. This is analogous to equation (\ref{eqn:XYweighted}) deduced earlier. This Hamiltonian preserves a quantity similar to the total $z-$spin in the $SU(2)$ model. The conservation relation is
\begin{equation}
    \left[ H_{XY}^d,\sum_{j=1}^n \eta^{r,r}_{(j)}\right]=0, \quad \text{for } 1\leq r \leq (d-1).
\end{equation}
The dynamics of Laplacian can be similarly defined. The transfer dynamics of this Hamiltonian for entanglement transfer with chains has been shown in \cite{ref:20}.

\section{PST in arbitrary qudit graphs (Bosonic model)}
\label{sec:hubbardbose}
Spin-1/2 (fermionic) systems can be seen as qubits ($d=2$). Similarly, arbitrary $d$-level system can be realised by an appropriate spin. Here we investigate perfect state transfer of the bosonic spin particles on bosonic network lattice - the Bose-Hubbard model. This work was proposed in \cite{ref:21}. First, we present the incorrect calculation reported in \cite{ref:21}. Equation $\sum_lU_{l1}\Tilde{b}_l^\dagger=\sum_l\Tilde{b}_l^\dagger$ following after equations 2-8 on page 174 in \cite{ref:21} is incorrect. This equation has been used to derive equation 2-9 and finally the condition 2-12 for perfect state transfer which is incorrect. This relation implies that every entry in the first column of $U$ is unity. This is because $U$ is a unitary matrix after all where a complete column cannot have each entry as unity. This would imply a non-unitary transformation over each adjacency matrix $A_k$.

We reformulate this method to derive the right condition for the perfect state transfer of qudits. The Hamiltonian governing the Bose-Hubbard model is
\begin{equation}
    H_d=\sum_{k=1}^{n-1}J_k\left( b_k^\dagger b_k+b_k b_k^\dagger \right)+\sum_{k=1}^n\epsilon_kb_k^\dagger b_k
\end{equation}
where the first term governs the hopping from one site to its adjacent site and second term is the local energy of that site. $b_k^\dagger$ and $b_k$ are the bosonic creation and annihilation operators for the site $k$ whose action on the basis kets $|0\rangle,|1\rangle$ (for each site) is well known. We are mostly concerned with the first coupling term. This model is applicable on general arbitrary finite graphs also. This can be generalised to an arbitrary graph in a weighted scheme. Suppose $\Omega$ is a connected graph. For each coupling strength $J_k$, $k = 0, 1,..., d$, we can form a graph $\Omega_k$ in which vertices are adjacent if their coupling in $\Omega$ equals $k$. Let $A_k$ be the adjacency matrix of $\Omega_k$. For instance, $A_1$ is the adjacency matrix $A$ of $\Omega$. Also, let $A_0 = I$, the identity matrix. This gives us $d + 1$ matrices $A_0, A_1,..., A_d$, called the adjacency matrices of $\Omega$. Their sum is the matrix $J$ in which every entry is 1. In the other words, we assume that the dynamics of bosons, in a system with $n$ sites (associated with the nodes of a finite group), is governed by
the following off-diagonal Bose-Hubbard Hamiltonian
\begin{equation}
    H_d=\sum_{i,j=1}^n\sum_{k=0}^dJ_k(A_k)_{ij}b_i^\dagger b_j
\end{equation}
The ground assumption is that the matrices $A_k$ for $k=0,1,...,d$ commute with each other such that all of them are simultaneously by the matrix $U$. The same matrix $U$ pointed above for a calculation mistake in the original work. The diagonalization follows as $UA_kU^\dagger=D_k$, where $D_k=\diag (\lambda_1^{(k)},\lambda_2^{(k)},...,\lambda_n^{(k)})$ is a diagonal matrix with eigenvalues of $A_k$ on its diagonal. Using this relation in the Hamiltonian expression we get
\begin{equation}
    H_d=\sum_{i,j=1}^n\sum_{k=0}^dJ_k(U^\dagger D_kU)_{ij}b_i^\dagger b_j=\sum_{k=0}^d\sum_{l=1}^nJ_k\lambda^{(k)}_l\left( \sum_iU^\dagger_{il}b_i^\dagger \right)\left( \sum_jU_{lj}b_j \right)
\end{equation}
with $l=1,2,...,n$. The term in the brackets can be treated as Bogoliubov transformation which is a linear transformation on creation/annihilation operators. Define this change of basis as
\begin{equation}
    \Tilde{b}_l=\sum_jU_{lj}b_j, \quad \quad b_l=\sum_jU^*_{jl}\Tilde{b}_j
\end{equation}
Define the number operator in the new tilde basis as $\Tilde{n}_l=\Tilde{b}^\dagger_l\Tilde{b}_l$ and the effective coupling as $\Tilde{J}_l=\sum_{k=0}^dJ_k\lambda^{(k)}_l$. Then Hamiltonian simply takes the diagonal form
\begin{equation}
    H_d=\sum_{l=1}^n\Tilde{J}_l\Tilde{n}_l.
\end{equation}
Without loss of generality, at $t=0$, we start with an initial state localised at the first qudit
\begin{equation}
|\psi_d(0)\rangle=\sum_{j=0}^d\alpha_j(b_1^\dagger)^j|\mathbf{0}\rangle
\end{equation}
where $|\mathbf{0}\rangle:=|00...0\rangle$ and $(b_k^\dagger)^j|\\mathbf{0}\rangle=|0_10_2...j_k...0_n\rangle=|j\rangle$ (which is understood to be a action on site $k$ by context). This is a harmonic oscillator (energy levels equally spaced) like $d$-level state. All the qudits are identical in the network. For $d=2$, analogous to \ref{eqn:bose1evo}, we will have the evolution of the initial state as
\begin{equation}
\label{eqn:quditevo}
    |\psi_d(t)\rangle=e^{-iH_dt}|\psi_d(0)\rangle=\alpha_0|\mathbf{0}\rangle+\alpha_1\sum_{j=1}^nf_{j1}(t_0)b^\dagger_j|\mathbf{0}\rangle
\end{equation}
where 
\begin{equation}
    f_{j1}(t)=\langle \\mathbf{0}|e^{-iH_dt\sum_{l=1}^n\Tilde{J}_l\Tilde{n}_l}b_1^\dagger|\mathbf{0}\rangle
\end{equation}
and for the perfect state transfer from forst qubit to some $m$th qubit, we impose $|f_{m1}(t_0)|=1$ for some finite $t=t_0$. We derive a condition for qubit using this Hamiltonian and generalise it over $d$-levels. Change of basis in first part of \ref{eqn:quditevo} yields
\begin{equation}
     |\psi_d(t_0)\rangle=\alpha_0|\mathbf{0}\rangle+\alpha_1e^{-it_0\left(\sum_{l=1}^n\Tilde{J}_l\Tilde{b}^\dagger_l\Tilde{b}_l\right)}\sum_{m=1}^nU_{m1}\Tilde{b}^\dagger_m|\mathbf{0}\rangle
\end{equation}
\begin{equation}
\begin{split}
    &=\alpha_0|\mathbf{0}\rangle+\alpha_1\sum_{\substack{l=1\\m=1}}^ne^{-it_0\Tilde{J}_l}|\Tilde{l}\rangle \langle\Tilde{l}| U_{m1}|\Tilde{m}\rangle \quad \text{(Spectral decomposition)}\\
    & =\alpha_0|\mathbf{0}\rangle+\alpha_1\sum_{\substack{l=1\\m=1}}^ne^{-it_0\Tilde{J}_l}U_{m1}|\Tilde{l}\rangle \delta_{\Tilde{l}\Tilde{m}} \\
    &=\alpha_0|\mathbf{0}\rangle+\alpha_1\sum_{l=1}^ne^{-it_0\Tilde{J}_l}U_{l1}\Tilde{b}^\dagger_l|\mathbf{0}\rangle.
\end{split}
\end{equation}
After reverting back to the initial basis we obtain
\begin{equation}
\label{eqn:correctedqudit}
      |\psi_d(t_0)\rangle=\alpha_0|\mathbf{0}\rangle+\alpha_1\sum_{l,j=1}^ne^{-it_0\Tilde{J}_l}U_{l1}U^*_{lj}b^\dagger_j|\mathbf{0}\rangle
\end{equation}
in contrast to equation 2-9 in \cite{ref:21} which we report as incorrect. To extract a more compact form define the column vector
\begin{equation}
    \Tilde{\mathbf{J}}=\left( U_{11}e^{-it_0\Tilde{J}_1}\quad U_{21}e^{-it_0\Tilde{J}_2} \quad ... \quad U_{n1}e^{-it_0\Tilde{J}_n} \right)^T
\end{equation}
then equation (\ref{eqn:correctedqudit}) can be re-expressed as
\begin{equation}
\begin{split}
    |\psi_d(t_0)\rangle&= \alpha_0|\mathbf{0}\rangle+\alpha_1\sum_{j=1}^n \left( \sum_{l=1}^nU_{jl}^\dagger \Tilde{\mathbf{J}}_l \right)b^\dagger_j|\mathbf{0}\rangle\\
    &= \alpha_0|\mathbf{0}\rangle+\alpha_1\sum_{j=1}^n \left(U^\dagger \Tilde{\mathbf{J}} \right)_jb^\dagger_j|\mathbf{0}\rangle.
\end{split}
\end{equation}
Comparing with equation (\ref{eqn:quditevo}) we have that
\begin{equation}
\label{eqn:fidelutyqudit}
    f_{j1}(t_0)=\left(U^\dagger \Tilde{\mathbf{J}} \right)_j=\sum_{l=1}^n\left(U^\dagger\right)_{jl}U_{l1} e^{-it_0\Tilde{J}_l}
\end{equation}
and the condition for perfect state transfer to the $m$th site imposes
\begin{equation}
\label{eqn:fidelityconditionqudit}
    \left(U^\dagger \Tilde{\mathbf{J}} \right)_j=e^{i\phi}\delta_{jm}
\end{equation}
for some arbitrary phase $\phi$ which can always be corrected, post the transfer.

To have estimation about the entries of the matrix $U$, define the product
\begin{equation}
    b_i^\dagger b_j:=E_{ij}.
\end{equation}
Here the action of $b_i^\dagger b_j$ is $b_i^\dagger b_j |k\rangle=\delta_{jk}|i\rangle$ as understood by the operator action on the basis. We can deduce that $E_{ij}$ is a matrix with all entries zero except the $(i,j)$ entry ($\langle i|E_{ij}|j\rangle=1$), that is, $(E_{ij})_{kl}=\delta_{ik}\delta_{jl}$. Then for the indices which are adjacent in graph $\Omega_k$, denoted as $i\sim _kj$
\begin{equation}
    \sum_{i\sim _kj}=\sum_{i\sim _kj}E_{ij}=A_k.
\end{equation}
Then the Hamiltonian can be expressed in terms of adjacency matrices $A_k$ as
\begin{equation}
    H_d=\sum_{k=0}^dJ_k\sum_{i\sim _kj}E_{ij}=\sum_{k=0}^dJ_k A_k.
\end{equation}
The matrices $A_k$ can always be diagonalized as $A_k=\sum_l\lambda_l^{(k)}|l\rangle \langle l|=\sum_l\lambda_l^{(k)}E_l$, where $E_l$ are the corresponding projectors of the $l$th subspace spanned by the eigenvectors corresponding to the eigenvectors of $\lambda_l^{(k)}$. Using this observation for spectral decomposition we have
\begin{equation}
    f_{j1}(t_0)=\langle \mathbf{0}|b_je^{-iHt_0}b_1^\dagger |\mathbf{0}\rangle=\sum_le^{-it_0\Tilde{J}_l}\langle \mathbf{0}|b_jE_lb_1^\dagger |\mathbf{0}\rangle
\end{equation}
\begin{equation}
    =\sum_le^{-it_0\Tilde{J}_l}\langle j|E_l|1\rangle=\sum_le^{-it_0\Tilde{J}_l}U_{l1}U^\dagger _{jl}.
\end{equation}
Last step is done using equation (\ref{eqn:fidelutyqudit}).
This gives us the information about the entries of $U$ as follows
\begin{equation}
    U_{l1}U^\dagger_{jl}=\langle j|E_l|1\rangle=(E_l)_{j1}.
\end{equation}

This procedure for two levels can be applied to arbitrary $d$ levels treating each level as some higher level excitation of the ground level. Label 1 can be replaced by any arbitrary label $x$ which is initial site to begin the transfer from. For simplicity, let us stick with the initial site as the first site. The free evolution of a qudit will be as follows for each arbitrary level $i$, and all levels evolve independently as they are the basis kets,
\begin{equation}
    e^{-iH_dt}(b^\dagger _1)^i|\mathbf{0}\rangle=e^{-iH_dt}\sum_{l_1,...,l_i}\Tilde{b}^\dagger _{l_1}\Tilde{b}^\dagger _{l_2}...\Tilde{b}^\dagger _{l_i}|\mathbf{0}\rangle
\end{equation}
\begin{equation}
   =\sum_{l_1,...,l_i}e^{-it\sum_{k=1}^i\Tilde{J}_{l_k}} \Tilde{b}^\dagger _{l_1}\Tilde{b}^\dagger _{l_2}...\Tilde{b}^\dagger _{l_i}|\mathbf{0}\rangle
\end{equation}
\begin{equation}
\begin{split}
    =&\sum_{k_1,...,k_i}\left( \sum_{l_1}e^{-it\Tilde{J}_{l_1}}U_{l_11}U^*_{l_1k_1} \right)\left( \sum_{l_2}e^{-it\Tilde{J}_{l_2}}U_{l_21}U^*_{l_2k_2} \right)...\left( \sum_{l_i}e^{-it\Tilde{J}_{l_i}}U_{l_i1}U^*_{l_ik_i} \right)\\
    & \times b^\dagger _{k_1}b^\dagger _{k_2}...b^\dagger _{k_i}|\mathbf{0}\rangle
\end{split}
\end{equation}
which can again be written in the compact form similar to the previous analysis as
\begin{equation}
    e^{-iH_dt}(b^\dagger _1)^i|\mathbf{0}\rangle=\sum_{k_1,...,k_i}\left( U^\dagger \Tilde{\mathbf{J}} \right)_{k_1}\left( U^\dagger \Tilde{\mathbf{J}} \right)_{k_2}...\left( U^\dagger \Tilde{\mathbf{J}} \right)_{k_i} \times b^\dagger _{k_1}b^\dagger _{k_2}...b^\dagger _{k_i}|\mathbf{0}\rangle.
\end{equation}
This transfers each excitation level of the qudit independently from one site to another. The final state of the system will be
\begin{equation}
\begin{split}
    |\psi_d(t)\rangle=e^{-iH_dt}|\psi_d(0)\rangle &=\alpha_0|\mathbf{0}\rangle+\alpha_1\sum_{k_1}\left( U^\dagger \Tilde{\mathbf{J}} \right)_{k_1}b^\dagger _{k_1}|\mathbf{0}\rangle+\alpha_2\sum_{k_1,k_2}\left( U^\dagger \Tilde{\mathbf{J}} \right)_{k_1}\left( U^\dagger \Tilde{\mathbf{J}} \right)_{k_2}b^\dagger _{k_1}b^\dagger _{k_2}|\mathbf{0}\rangle\\
    & +...+\alpha_d \sum_{k_1,...,k_d}\left( U^\dagger \Tilde{\mathbf{J}} \right)_{k_1}\left( U^\dagger \Tilde{\mathbf{J}} \right)_{k_2}...\left( U^\dagger \Tilde{\mathbf{J}} \right)_{k_d} \times b^\dagger _{k_1}b^\dagger _{k_2}...b^\dagger _{k_d}|\mathbf{0}\rangle.
\end{split}
\end{equation}
The condition for perfect state transfer applies to each excitation level which exactly remains the same condition as equation (\ref{eqn:fidelityconditionqudit}) for $t=t_0$.
The final state when transfer is accomplished to the $m$th site is simply
\begin{equation}
     |\psi_d(t_0)\rangle=\sum_{j=0}^d\alpha_j\left( b^\dagger_m \right)^j|\mathbf{0}\rangle.
\end{equation}
Special graphs can be studied under this model of perfect state transfer for qudits. PST for qudits has been demonstrated experimentally on superconducting transmon qudits in \cite{ref:51}. Qudit PST for pseudo-regular networks is explored in \cite{ref:52}. An architecture for arbitrary long distances qudit perfect state transfer using multiple hoppings is presented in \cite{ref:53}.

%% file: Chapters/PhysicalImplementation.tex
\label{chap:phyreal}

Our proposed architecture in chapter \ref{chap:twohop} addresses the problem of scalability of quantum processors while preserving the perfect fidelity for state transfer with full control on routing of initial states. Scalability of quantum processors is a deep concern in the development of quantum computing hardware \cite{ref:26}\cite{ref:27}\cite{ref:28}. In the NISQ (Noisy Intermediate-Scale Quantum) \cite{ref:29} era quantum processors, high-fidelity quantum operations and attempt of scalable quantum networks are key features. These features are essential for determining how good is an architecture for quantum information processing \cite{ref:30}. One of the main challenges for scalable quantum network is the imperfect two-qubit interaction. For perfect state transfer with maximum fidelity, the pairwise interaction should be improved for large-scale quantum processors \cite{ref:31}. Different physical systems for quantum computation have different advantages, such as, high-fidelity and control in ion-traps \cite{ref:35}\cite{ref:34} versus the scalability of superconducting circuits \cite{ref:33}\cite{ref:23}\cite{ref:32}. Our focus in this work will be with the high-fidelity, state control and scalability of the superconducting quantum computing architecture and proving that our two-hopping all-to-all state transfer scheme can be perfectly realised with this approach.

For the usual perfect state transfer scheme, the underlying assumption is that all the edges in the graph network are precisely engineered at the desired coupling strength. If this is violated, we will have to compromise a lot with the fidelity of the state transfer. PST Hamiltonians like $XY$ and Heisenberg model are nothing but the sum of the pairwise interactions between the connected qubits that are allowed to interact in these two defined schemes. The coupling Hamiltonian for an isolated pair in PST scheme is simply the gate Hamiltonian for the quantum computation. This implies that PST is nothing but the complete free quantum evolution (for a desired time interval which is the control parameter) of the entire network in contrast to the controlled quantum evolution of an isolated pair of qubits (which is a quantum gate operation). Improving two qubit gate fidelity will improve the perfect state transfer fidelity. Connectivity of the qubits (as a graph structure) in the quantum circuit model will determine which pair of qubits can have two-qubit gate operations between them mutually. Disconnected qubits on the physical architecture will prohibit two-qubit interactions between these qubits. This is how the graph determines the architecture. More connectivity will enable more two-qubit interactions. However, there is also a physical limitation for the maximum nearest neighbour interactions a qubit can sustain \cite{ref:69}.

There are two sources of two-qubit interaction errors: decoherence (stochastic) and nonideal interactions (deterministic). The latter includes parasitic coupling, leakage to non-computational states, and control crosstalk. As one example of parasitic coupling, the next-nearest-neighbor (NNN)
coupling is a phenomenon commonly seen in many systems, including Rydberg atoms \cite{ref:38}, trapped ions \cite{ref:39}\cite{ref:34}, semiconductor spin qubits \cite{ref:40}, and superconducting qubits \cite{ref:41}\cite{ref:42}. Unwanted interactions (such as next-nearest neighbour) between qubits are meant to be unconnected.

\section{Requirements for our CQC-hopping scheme}
To implement our network on a real architecture we follow the physical architecture given in \cite{ref:14} for two qubits and generalize it over arbitrary number of qubits $n$. Physical key requirements of our architecture are the following:
\begin{itemize}
    \item Multiple nearest-neighbour (NN) interactions
    \item Since distant qubits are connected, implementation is not possible in planar integration. Three-dimensional (3D) integration is needed \cite{ref:36}\cite{ref:69}
    \item Tunable (eventually switchable) edges as couplings for each pair of nodes (qubits)
    \item Multiple qubit coupling controlled independently (each qubit independently be brought in dispersive regime) \cite{ref:31}
    \item Addition of a new node to the existing network and so on (addressing scalability)
    \item High fidelity control over the processor \cite{ref:37}
\end{itemize}
Most of these requirements cannot be fulfilled by the usual spin model where switching and tuning of edge coupling becomes nearly impossible as the experimentalist control is negligible for spins on a fixed lattice. Moreover, such coupling is a function of the distance between two nodes which is not changeable in practice. Multiple edges are very hard to be tuned to the same coupling strength and the nodes that are physically distant are impossible to couple. We show that all these implementation problems can be addressed with the architecture based on tunable-coupling superconducting circuits. A tunable coupler can also help mitigate the problem of frequency crowding that exacerbates the effect from nonideal interactions. However, these additional elements often add architectural complexity, as well as open a new channel for decoherence and crosstalk. Many prototypes of
a tunable coupler have been demonstrated in superconducting quantum circuits, such as the the gmon design (\cite{ref:43}, two-qubit gate fidelity limited by decoherence) and xmon design (\cite{ref:44}, qubit's coherence time is decreased by the tunable coupler), are two examples in the literature. Addition and deletion of a new node (qubit) to the graph (network) has been treated as the inclusion and exclusion of that qubit in the computational network. That is, this is achieved by turning of every interaction of that node with the rest of the network and switching it on when this node is to be added. Building and fabricating a new qubit into the existing processor is a matter of the available resources and depends upon lab to lab. We do not address that matter. For physical realization, the addition of nodes can be regarded as turning on interaction with more qubits that already exist by construction in the quantum processor.

\section{Tunable-coupling effective Hamiltonian for CQC-hopping architecture}

In the superconducting transmon charge qubit architecture \cite{ref:25} all these limitations can be tackled and all the above requirements can be fulfilled as we demonstrate. Moreover, this architecture maybe used for the conventional quantum computation model in addition to quantum state transfer application due to the high-fidelity gate operations that are possible \cite{ref:14}. This architecture allows high control over the coupling macroscopically and decoherence times are quite longer. A gate (two qubit interaction) fidelity of 99.999\% is reported in \cite{ref:14} in the absence of decoherence.

We consider a generic system consisting of an arbitrary network of qubits  with exchange coupling between nearest qubits (which have an edge between them) and a coupler between these that couples to both these qubits. The approach is based on a generic three-body system with exchange-type interaction. A central component, the coupler, frequency tunes the virtual exchange interaction between two qubits and features a critical bias point, at which the exchange interaction offsets the direct qubit-qubit coupling, effectively turning off the net coupling. Two-qubit interactions are executed for each pair of qubits by operating the respective couplers in the dispersive regime, strongly suppressing leakage to the coupler’s excited states. The two qubits (with Zeeman splittings $\omega_i$ and $\omega_j)$ each couple to a center tunable coupler $(\omega^c_{ij})$ with a coupling strength $g_i$ $(i = 1, 2,...,n)$, as well as to each other with a coupling strength $g_{ij}$. The ancillary coupling is stronger than the direct coupling , $g_i,g_j > g_{ij} > 0$. Ancillary coupler does not count as the part of our network, it is only needed as the part of the implementation of architecture. We have the total number of ancillary couplers equal to the number of edges in the network. Without loss of generality, we begin our analysis with a two-level Hamiltonian,
\begin{equation}
\begin{split}
    H= \frac{1}{2}\sum_{i=1}^n\omega_i\sigma^z_i+\frac{1}{2}\sum_{\langle i,j \rangle}\omega^c_{ij}\sigma^z_{C_{ij}}+\sum_{\langle i,j \rangle}g_i\left( \sigma^+_i\sigma^-_{C_{ij}}+\sigma^-_i\sigma^+_{C_{ij}} \right) +\sum_{\langle i,j \rangle}g_{ij}\left( \sigma^+_i\sigma^-_j+\sigma^-_i\sigma^+_j \right)
\end{split}
\end{equation}
where all the operators are defined for the respective modes and $\langle i,j \rangle$ means there is an edge between qubits $i$ and $j$. First two terms are just the Zeeman splittings of the network qubits and the ancillary couplers for all edges. Third term (call it $V$, treated as a perturbation) describes the coupling of each qubit to coupler and the last term is the direct coupling between the network qubits which have an edge between them. $H=H_0+V$ can be written where we treat the coupling to the ancilla as external coupling one wants to get rid of. In this setting, all qubits are negatively detuned from the resonance with the ancillary coupler with $\Delta_j=\omega_j-\omega^c_{ij}<0$ and we operate in the dispersive regime for all qubits with $g_j\ll |\Delta_j|\quad \forall j$. Any two connected qubits interact through two channels, the direct nearest coupling and the indirect coupling via the ancilla (which can be regarded as a virtual exchange interaction). The idea is to make these two couplings compete against each other and tune the desired strength of coupling for each pair of qubit. We desire to tune all couplings at equal magnitude to make an uniformly coupled network as proposed.

\begin{figure}[hbtp]
    \centering
    \includegraphics[scale=0.5]{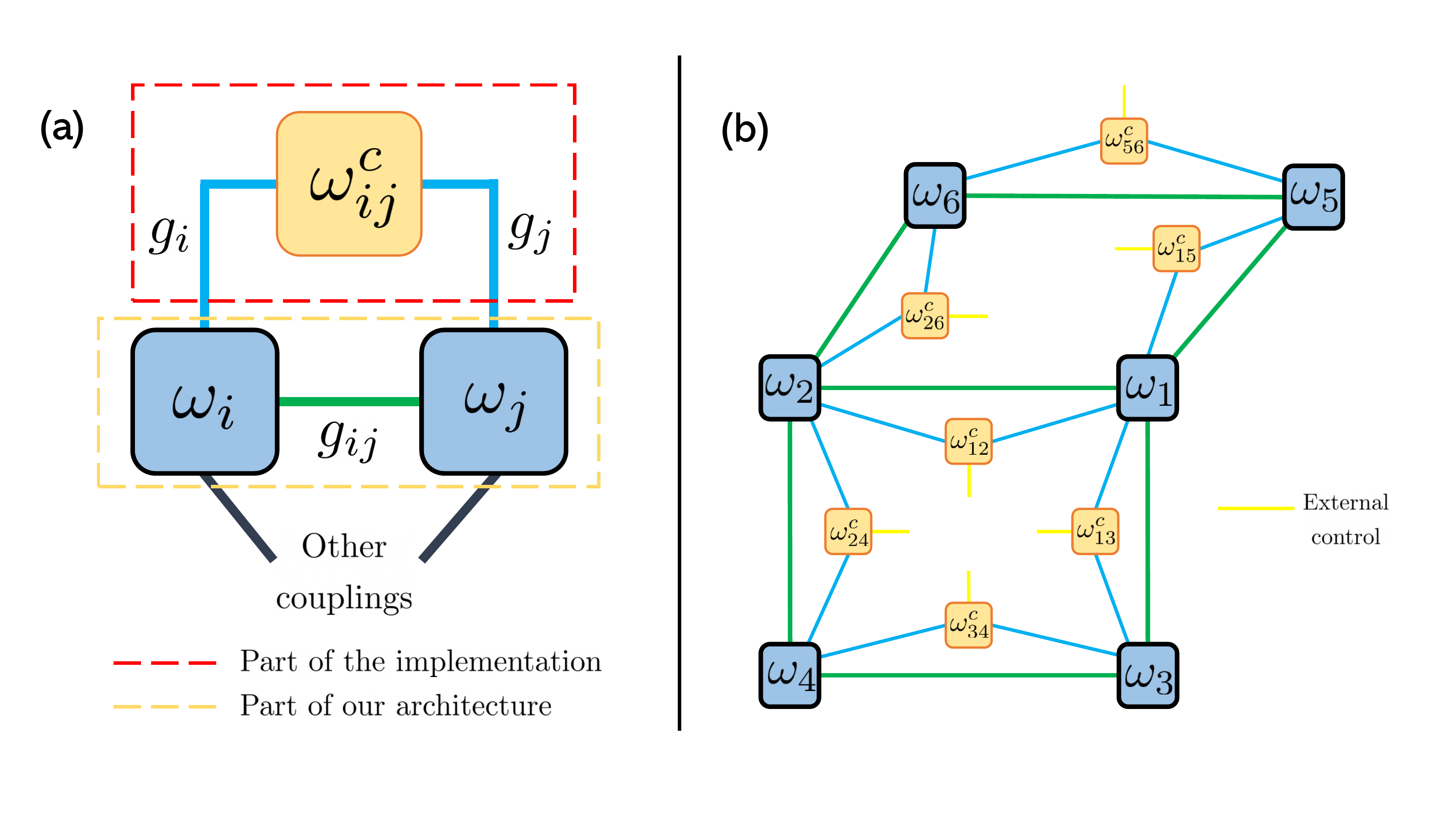}
    \caption{The part (a) shows the isolated pair of coupled qubits with the ancillary coupler acting as the inter-mediator. The situation for every pair that has an edge between them is the same. Part (b) shows our qubit architecture for $n=6$ qubits with the ancillary couplers involved. The control parameter for edge switching here is the external control on the capacitance for each coupler, denoted as yellow.}
    \label{fig:coupling}
\end{figure}

To find the effective qubit-qubit coupling and eliminate the qubit-ancilla coupling, we take advantage of the Schrieffer-Wolff unitary transformation $U_{SW}=e^\eta$.
The transformation if formally represented as 
\begin{equation}
    \Tilde{H}=U_{SW}HU_{SW}^\dagger=H+[\eta,H]+\frac{1}{2!}[\eta,[\eta,H]]+...
\end{equation}
\begin{equation}
    = H_0+V+[\eta,H_0]+[\eta,V]+\frac{1}{2!}[\eta,[\eta,H_0]]+\frac{1}{2!}[\eta,[\eta,V]]+...
\end{equation}
If one can find a transformation $\eta$ such that $V + [\eta, H_0] = 0$, the transformed Hamiltonian becomes:
\begin{equation}
    \Tilde{H}=H_0+\frac{1}{2}[\eta,V]+\mathcal{O}(V^3)
\end{equation}
because $\eta \propto \mathcal{O}(V )$. The standard way to find this transformation requires to calculate the commutator
\begin{equation}
    \left[H_0,V\right]=\sum_{\langle i,j\rangle}g_i\left( g_{ij}\sigma_i^z -\Delta_i \right)\left( \sigma^+_i\sigma^-_{C_{ij}}-\sigma^-_i\sigma^+_{C_{ij}} \right)
\end{equation}
and impose this operator ansatz form with free parameters $\alpha_i^{\pm} as$
\begin{equation}
    \eta=\alpha_i^+\sigma^+_i\sigma^-_{C_{ij}}+\alpha^-_i\sigma^-_i\sigma^+_{C_{ij}}.
\end{equation}
Then we evaluate the commutator $[H_0,\eta]$ and equate to $V$ to find the free parameters $\alpha^\pm_i=\pm g_i/\Delta_i$. This gives the transformation $U_{SW}$ for our case as
\begin{equation}
    U_{SW}=e^\eta =\exp\left( \sum_{\langle i,j\rangle} \frac{g_i}{\Delta_i}\left( \sigma^+_i\sigma^-_{C_{ij}}-\sigma^-_i\sigma^+_{C_{ij}} \right) \right).
\end{equation}
For general SW-transformation within and beyond rotating wave approximation (RWA), see \cite{ref:24} and Appendix: B of \cite{ref:14}. Performing the SW transformation up to second order we get the Hamiltonian
\begin{equation}
\begin{split}
    &\Tilde{H}=\frac{1}{2}\sum_{j=1}^n\omega_j\sigma^z_j+\frac{1}{2}\sum_{\langle i,j\rangle}^n\omega^c_{ij}\sigma^z_{C_{ij}}+\sum_{\langle i,j\rangle}\left(\frac{g_ig_j}{\Delta_{ij}}+g_{ij}\right)\left( \sigma^+_i\sigma^-_j+\sigma^-_i\sigma^+_j \right)\sigma^z_{C_{ij}}\\
    &+ \sum_{\langle i,j \rangle}\frac{g_i^2}{\Delta_i}\left( \sigma^z_i\sigma^-_{C_{ij}}\sigma^+_{C_{ij}}+\sigma^-_i\sigma^+_i\sigma^z_{C_{ij}}\right)+\sum_{\substack{\langle i,j \rangle \\ \langle i,k \rangle} }\frac{g_i^2}{\Delta_i}\left(\sigma^+_i\sigma^-_i\sigma^-_{C_{ij}}\sigma^+_{C_{ik}}-\sigma^-_i\sigma^+_i\sigma^+_{C_{ij}}\sigma^-_{C_{ik}}\right).
\end{split}
\end{equation}
Now we make an important assumption that the ancillary couplers always remain in their ground state for all the edges and drop out the constant energy terms. This combines first and fourth term in the Hamiltonian and removes the last term which is exchange interaction between different ancillary couplers of the same qubit. This transformation finally decouples the ancillary coupler from the qubits up to second order in $g_i/\Delta_i$ resulting in
\begin{equation}
\label{eqn:SCtuned}
    \Tilde{H}=\frac{1}{2}\sum_{j=1}^n\Tilde{\omega}_j\sigma^z_j+\sum_{\langle i,j\rangle}\left(\frac{g_ig_j}{\Delta_{ij}}+g_{ij}\right)\left( \sigma^+_i\sigma^-_j+\sigma^-_i\sigma^+_j \right)
\end{equation}
where $\Tilde{\omega}_j=\omega_j+g_j^2/\Delta_j$ is the Lamb-shifted frequency revealed by the SW transformation and $\Delta_{ij}=2\Delta_i\Delta_j/\left( \Delta_i+\Delta_j \right)<0$ is the effective detuning. the coefficient of the second term, denote it as $\Tilde{g}_{ij}$, is the effective \textit{tunable} coupling between any two qubits that are coupled in the network wherever there is an edge. Identifying $\tilde{g}_{ij}$ as $2 J_{ij}$ gives the identical coupling Hamiltonian in equation (\ref{eqn:XYmodel}). However, now our $J_{ij}$ is tunable to a range of values we desire by setting the desired couplings and detunings. 
The effective coupling $\tilde{g}_{ij}$ in equation (\ref{eqn:SCtuned}) can be adjusted by the ancilla coupler frequency through $\Delta_{ij}$, as well as $g_i$ and $g_j$, both of which may be implicitly dependent on $\omega^c_{ij}$. Therefore, $\tilde{g}_{ij}$ is a function of $\omega^c_{ij}$ in general. The first term (indirect coupling) in the expression of $\tilde{g}_{ij}$ is negative while the second (direct coupling) is positive and this enables a competition between the two where $\omega^c_{ij}$ can be taken to act as the tunable parameter since it can be externally controlled in the experiment. $\tilde{g}_{ij}(\omega^c_{ij})$ can be tuned negative when ancilla coupler frequency is decreased or positive when this frequency is increased. And this is a continuous parameter, therefore we have some $\omega^c_{ij_{\text{off}}}$ such that $\tilde{g}_{ij}(\omega^c_{ij_{\text{off}}})=0$ which should be permitted by the bandwidth of the ancilla coupler. It is shown \cite{ref:14} that this cut-off frequency can be found even in weak dispersive regime with $g_j<|\Delta_j|$. Thus, in principle, we obtain the switchable edges with $\omega^c_{ij}$ as the parameter. We can simply tune each frequency $\omega^c_{ij}$ for each edge $E(i,j)$ to switch it on or off when our protocol requires and this is essentially a classical operation in experiment. Therefore, some of the edges maybe switched off by selecting special cut-off values. The couplers remain in their ground state throughout the quantum evolution as the effective interaction is only for one quantum exchange between the two qubits which are part of the network. Similar effective coupling Hamiltonians based on Cavity and Circuit-QED have been proposed in \cite{ref:22} (scalability has been addressed with experimental concerns using molecular architecture for qubits in superconducting resonators) and \cite{ref:23} (foundational reference for superconducting electrical circuits).

\section{Circuit Hamiltonian quantization}
Now, the above general formalism can be applied to specific system Hamiltonian. For our case we use the transmon qubits \cite{ref:45}. Josephson energy Hamiltonian with tunable energy is
\begin{equation}
    E_{J_\lambda}=\left( E_{J_\lambda,L}+E_{J_\lambda,R} \right)\sqrt{\cos^2\left( \frac{\pi\Phi_{e,\lambda}}{\Phi_0} \right)+\left( \frac{E_{J_\lambda,L}-E_{J_\lambda,R}}{E_{J_\lambda,L}+E_{J_\lambda,R}} \right)^2\sin^2\left( \frac{\pi\Phi_{e,\lambda}}{\Phi_0} \right)}
\end{equation}
where $\lambda=i,j,c_{ij}\quad \forall i,j$ ($c_{ij}$ is the labeling for the coupler connecting $i$ and $j$ qubits) and $\Phi_0=h/2e$ is the superconducting flux quantum. The coefficient of the sine term quantifies the asymmetry of the junction \cite{ref:46}. Refer to figure \ref{fig:SCqubits} for notation. $E_{J_{\lambda,L(R)}}$ is the Josephson energy of the left(right) junction in mode $\lambda$. $C_\lambda$ is the dominant capacitance for that mode. $C_{jc_{ij}}$ is the coupling capacitance between the qubit $j$ and the coupler $\langle ij\rangle$. $C_{ij}$ is the directing coupling capacitance between the two qubits $i$ and $j$, and $\phi_\lambda$ is the reduced total flux for that node.

\begin{figure}[hbtp]
    \centering
    \includegraphics[scale=0.65]{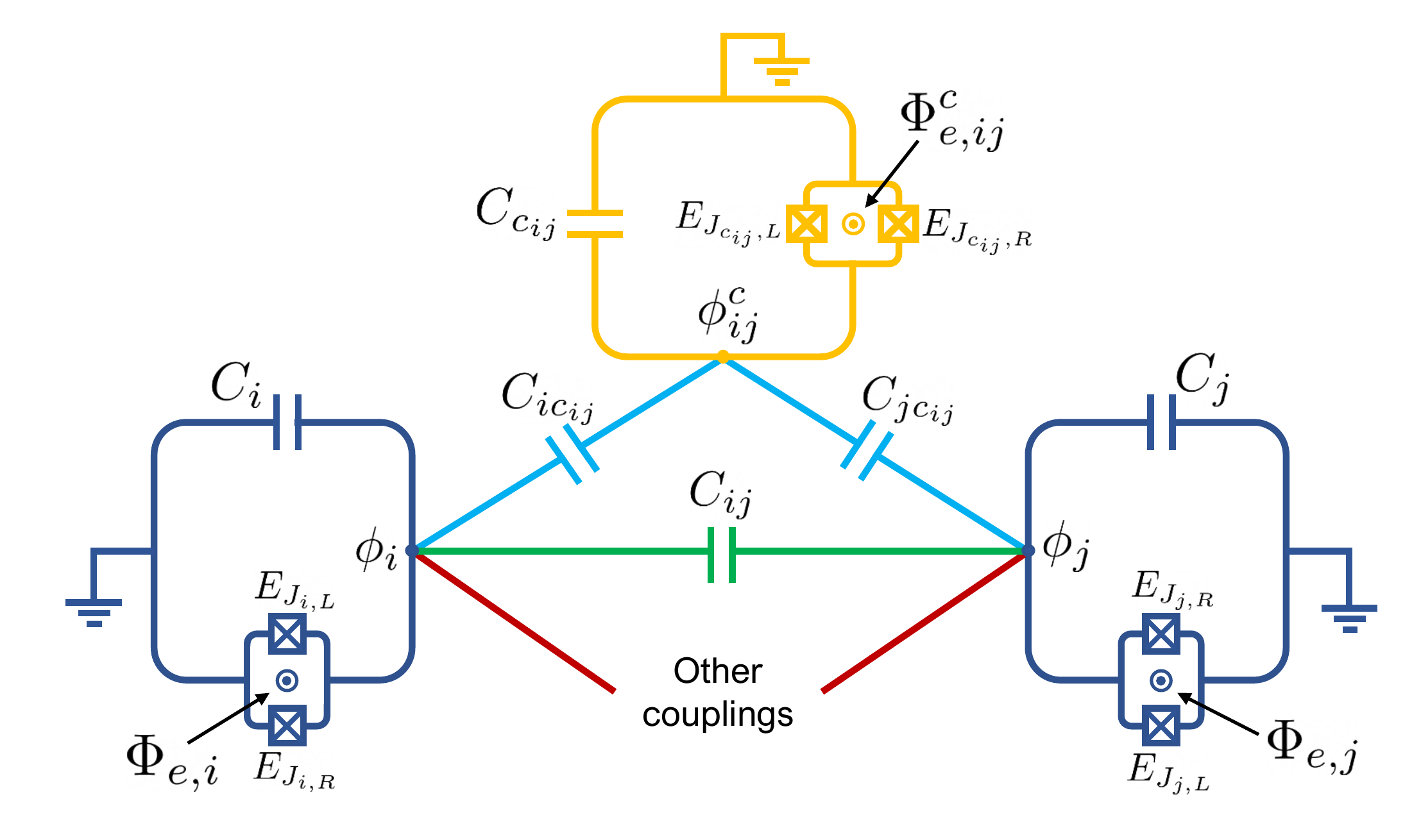}
    \caption{Schematic circuit diagram for a pair of connected tunable transmon qubits. Each connected pair which forms an edge on the graph has this structure.}
    \label{fig:SCqubits}
\end{figure}

This Hamiltonian can be canonically quantized in second quantization in the transmon regime with $E_{J_\lambda}/E_{C_\lambda}\gg 1$, where $E_{C_\lambda}=e^2/2C_\lambda$ for the corresponding mode \cite{ref:14}. The system is then described by the coupled oscillators ($\hbar=1$):
\begin{equation}
    H=\sum_{\langle i,j\rangle}\left( H_i+H_j+H_{c_{ij}}+H_{ic_{ij}}+H_{jc_{ij}}+H_{ij}  \right)
\end{equation}
Let the corresponding creation and annihilation operators for the respective mode be $b_\lambda^\dagger$ and $b_\lambda$ respectively, with similar action as in \ref{sec:hubbardbose}.
The terms of the Hamiltonian are
\begin{equation}
    H_\lambda=\omega_\lambda b_\lambda^\dagger b_\lambda+\frac{\alpha_\lambda}{2}b_\lambda^\dagger b_\lambda^\dagger b_\lambda b_\lambda,
\end{equation}
where $\alpha_\lambda$ is the anharmonicity (energy difference between the first excitation energy and further second excitation energy) of the oscillator, 
\begin{equation}
    H_{jc_{ij}}=g_j\left( b_j^\dagger b_{c_{ij}}+b_j b_{c_{ij}}^\dagger -b_j^\dagger b_{c_{ij}}^\dagger -b_j b_{c_{ij}} \right),
\end{equation}
\begin{equation}
    H_{ij}=g_{ij}\left( b_i^\dagger b_j+b_i b_j^\dagger-b_i^\dagger b_j^\dagger-b_i b_j \right),
\end{equation}
where the second order effect has also been taken into consideration. The energies are given as
\begin{equation}
    \omega_\lambda=\sqrt{8E_{J_\lambda}E_{C_\lambda}}-E_{C_\lambda},
\end{equation}
\begin{equation}
    g_j=\frac{1}{2}\frac{C_{jc_{ij}}}{\sqrt{C_jC_{c_{ij}}}}\sqrt{\omega_j\omega^c_{ij}}
\end{equation}
and
\begin{equation}
    g_{ij}=\frac{1}{2}(1+\eta_{ij})\frac{C_{ij}}{\sqrt{C_iC_j}}\sqrt{\omega_i\omega_j}
\end{equation}
where $\eta_{ij}=C_{ic_{ij}}C_{jc_{ij}}/C_{ij}C_{c_{ij}}$. Single quantum exchange is due to the Jaynes-Cumming type interaction while the double excitation and de-excitation effect arises due to counter rotating terms, without the rotating wave approximation (RWA), which is important when the coupler frequency is higher than of the related qubit frequencies.

\section{Circuit Hamiltonian dynamics beyond Rotating Wave Approximation (RWA)}
Following the general SWT formulated in \cite{ref:24}, the dynamics beyond the RWA for the above Hamiltonian is straightforward. The right SWT is
\begin{equation}
        U_{SW}=\exp \left(  \sum_{\langle i,j\rangle} \left[ \frac{g_j}{\Delta_j}\left( b_j^\dagger b_{c_{ij}}-b_j b_{c_{ij}}^\dagger \right)-\frac{g_j}{\Sigma_j}\left( b_j^\dagger b_{c_{ij}}^\dagger-b_j b_{c_{ij}} \right) \right] \right)
\end{equation}
where the second term takes care for the counter-rotating terms and $\Sigma_j=\omega_j+\omega^c_{ij}$. In the weak anharmonic limit, $\alpha_\lambda \ll \Delta_j$ (detuning of nearly the same order for all qubits). Expansion in the second order in couplings $g_j$ we obtain the effective qubit-qubit Hamiltonian
\begin{equation}
    \tilde{H}=U_{SW}HU_{SW}^\dagger=\sum_j\left( \tilde{\omega}_jb^\dagger_jb_j+\frac{\tilde{\alpha}_j}{2}b^\dagger_jb^\dagger_jb_jb_j \right)+\sum_{\braket{i,j}} \tilde{g}_{ij}\left( b_i^\dagger b_j+b_ib_j^\dagger \right)
\end{equation}
which is identical to equation (\ref{eqn:SCtuned}), plus the anharmonic terms along with counter rotating contribution to the coupling coefficients. Here,
\begin{equation}
    \tilde{\omega}_j\approx \omega_j+g_j^2\left( \frac{1}{\Delta_j}+\frac{1}{\Sigma_j} \right),
\end{equation}
\begin{equation}
    \tilde{\alpha}_j\approx \alpha_j,
\end{equation}
and the effective coupling as
\begin{equation}
    \tilde{g}_{ij}\approx \frac{g_ig_j}{2}\left( \frac{1}{\Delta_i}+\frac{1}{\Delta_j}-\frac{1}{\Sigma_i}-\frac{1}{\Sigma_j} \right)+g_{ij}.
\end{equation}
Here, the same assumption of the coupler being strictly in its ground state has been taken into account. Note that this is exactly similar to making the assumption that the cavity resonator remains in the constant photon number in the conventional cavity-quantum-electrodynamics Hamiltonians, while performing their Schrieffer-Wolff transformations. 
For the dispersive regime we simply have $|\Delta_j|\approx |\Sigma_j|$, which is when counter rotating terms contribute significantly. The computational states are $|1_i0_{c_{ij}}0_j\rangle$ and $|0_i0_{c_{ij}}1_j\rangle$ and they exchange their energy virtually through the non-computational coupler excited state $|0_i1_{c_{ij}}0_j\rangle$ by the virtue of the Jaynes-Cummings interaction ($b_j^\dagger b_{c_{ij}}+b_j b_{c_{ij}}^\dagger$). The counter-rotating term ($b_j^\dagger b_{c_{ij}}^\dagger+b_j b_{c_{ij}}$) involves exchange via the higher non-computational state $|1_i1_{c_{ij}}1_j\rangle$. Substituting the values for the couplings we obtain
\begin{equation}
\label{eqn:tunablecoupfinal}
    \tilde{g}_{ij}\approx \frac{1}{2}\left[ \frac{\omega^c_{ij}}{4}\left( \frac{1}{\Delta_i}+\frac{1}{\Delta_j}-\frac{1}{\Sigma_i}-\frac{1}{\Sigma_j} \right)\eta_{ij}+\eta_{ij}+1 \right]\times \frac{C_{ij}}{\sqrt{C_iC_j}}\sqrt{\omega_i\omega_j}.
\end{equation}
For the case when all qubits are identically set in the dispersive regime with identical construction, we set $\omega_i=\omega_j=\omega$. This results in
\begin{equation}
\label{eqn:coupfinal}
    \tilde{g}_{ij}=\frac{1}{2}\left[ \frac{\omega^2}{ \Delta_i \Sigma_i}\eta_{ij}+1 \right]\frac{C_{ij}}{\sqrt{C_iC_j}}\omega
\end{equation}
The first term in the can be made arbitrarily small for very high values $\omega^c_{ij}$ and very large values of $\eta_{ij}$ which can give zero effective coupling (a switched-off edge). This way, by just adjusting the coupler dynamics (which is externally under full control) we can turn off selected edges (couplings) in the network.

\section{CQC switching with the tunable coupling Hamiltonian}
The tunable couplers are used to turn off the graph edge interactions
by biasing their frequency at $\omega^c_{ij_{\text{off}}}$ during the switching period. To activate the two-qubit interaction as the edge in the network graph, one tunes the couplers' frequency to a desired value $\omega^c_{ij_{\text{on}}}$, yielding a finite $\tilde{g}(\omega^c_{ij_{\text{on}}})$. All the couplers are set to the same strength of coupling. Then a PST can be performed by modulating only the coupler frequency to $\tilde{g}(\omega^c_{ij_{\text{on}}})$ for all the edges $E(i,j)$ which have to be switched on while leaving the other qubits unperturbed during the PST. The edges which are switched on and switched off are known from our formalism in Chap. \ref{chap:twohop}. By operating the couplers in the dispersive limit, parasitic effects from higher-order terms that are ignored after SWT are strongly suppressed, leading to higher two-qubit hopping fidelity. During this process, the control Hamiltonian $\sigma^z_{C_{ij}}$ commutes with the qubits’ degrees of freedom within the dispersive approximation, causing reduced leakage to the non-computational (coupler) state. The non-adiabatic effect in this case is suppressed by the relatively large qubit-coupler detuning $\Delta_j$, allowing a shorter PST time and therefore and reduced decoherence error.

\begin{figure}[hbtp]
    \centering
    \includegraphics[scale=0.4]{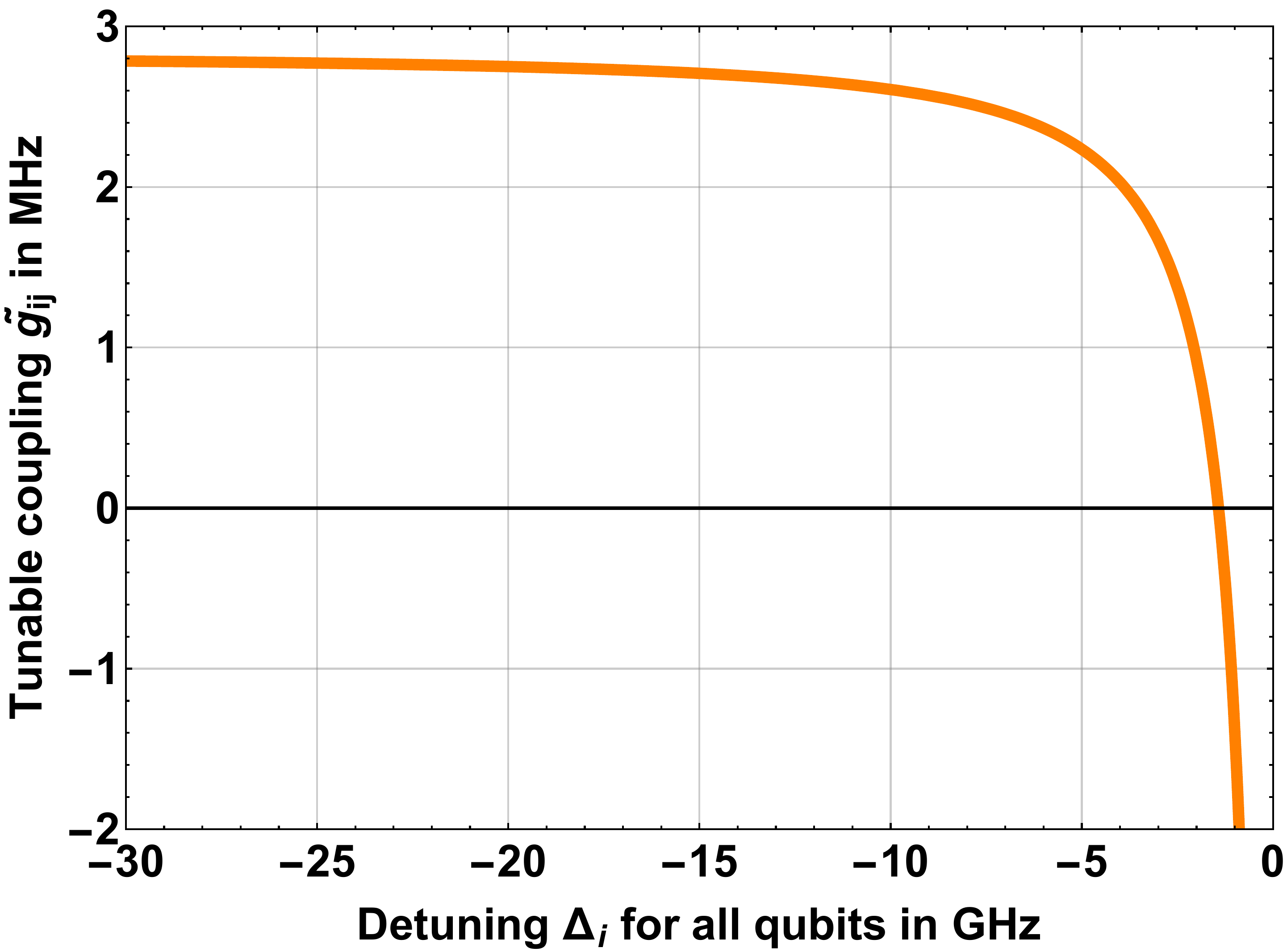}
    \caption{Variation of the dynamic tunable coupling $\tilde{g}_{ij}$ w.r.t. the detuning $\Delta_i$ for each qubit. There exists a cutoff value, in this case $\Delta_i=-1.426$ GHz, corresponding to $\omega^c_{ij_{\text{off}}}=5.426$ GHz. For all configurations, such a cut-off value can always be obtained.}
    \label{fig:coupling cutoff}
\end{figure}

Figure \ref{fig:coupling cutoff} shows the variation of the dynamic tunable coupling $\tilde{g}_{ij}$ with respect to the control parameter $\omega^c_{ij}$. The reasonable experimental values for the parameters used are \cite{ref:14}: $C_i=70$ fF, $C_j=72$ fF, $C_{c_{ij}}=200$ fF, $C_{ic_{ij}}=4$ fF, $C_{jc_{ij}}=4.2$ fF and $C_{ij}=0.1$ fF. All $\omega_j=\omega_j=\omega=4$ GHz (because all qubits are identical). The fabrication defects and imperfection is accounted in the different values for the capacitances. However, for a quite good variation amongst these values, we can still guarantee a cut-off value existence.

The perfect state transfer time for our CQC-hopping (2$t_0$) is plotted against the detuning of qubits. PST time is $t_0$ if we are at a perfect hypercube or $2t_0$ otherwise. For the same experimental values considered above, this is plotted in figure \ref{fig:PSTtime}

\begin{figure}[hbtp]
    \centering
    \includegraphics[scale=0.5]{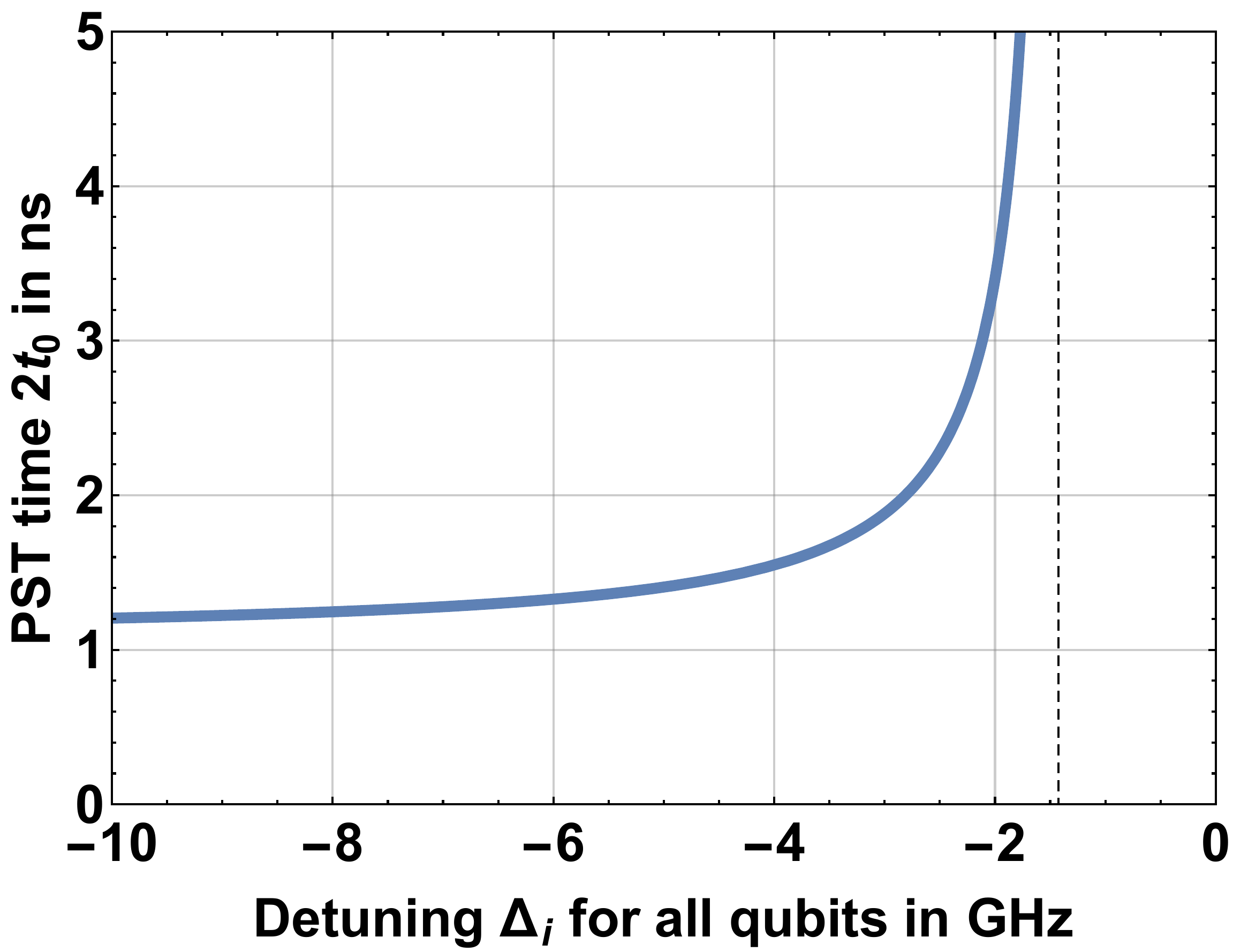}
    \caption{Variation of the Perfect State Transfer time $2t_0$ w.r.t. qubit detuning $\Delta_i$ for all qubits. The typical PST time is around 1.5 ns.}
    \label{fig:PSTtime}
\end{figure}

It can be deduced that the typical time scale for our CQC-hopping scheme is around 1.5 ns. This is less than the single iSWAP quantum gate time for the same experimental values which is reported to be around 35 ns. Even if we consider that switching takes few nanoseconds, still we perform PST in less time relative to the SWAP gate operation between two given qubits. We have to ensure in the experiment that all detunings are onset to the same value to realise a uniformly coupled qubit network. If all detunings are not equal this will actually realise a weighted coupled qubit network and introduce an error since our CQC-hopping protocol is not for weighted graph networks. This error can be estimated by calculation pairwise transfer fidelity via two different couplings. The typical coupling can be tuned around to 2 MHz (for example), to give PST time of 1.55 ns, by setting the detuning to $-4$ GHz, corresponding to $\omega^c_{ij}=8$ GHz for all couplers.

%% file: Chapters/Future_Work.tex
We proposed a switching procedure on a memory-enhanced hypercube such that an induced hypercube can be determined with a desired pair of antipodal vertices. A framework of superconducting qubits is defined for physical implementation of the switching procedure under the XY coupling. This same physical architecture also realises growing network scheme. It was shown that perfect state transfer between any pair of vertices in a hypercube or more generally in a network of arbitrary number of vertices is possible utilizing the proposed switching scheme. We have proved the computational advantage of using PST assisted quantum computing. We showed counter-examples numerically for PST under Corona product of graphs and corrected an error in a qudit PST paper which is a very rich resource for computational power in quantum computing.

There are certain results which can further be extended. Such as finding a physical system for scalable Laplacian perfect state transfer for switchable hypercubes similar to our scheme. One can also look to find more optimal graph operations, if they exist, for long distance PST to realize a growing network. Analytical results for signed Corona transfer may also be attempted to look for the class of graphs which support PST under Corona product. Qudit state transfer for large distances using least hoppings is still not reported in literature, to this date. It would be a big result to find such an optimal network for qudits under a suitable Hamiltonian so that large amount of quantum information can be transferred over long distances, similar to qubits.

%% file: Main.bbl
\begin{thebibliography}{10}

\bibitem{ref:48}
Charles~H. Bennett and David~P. DiVincenzo.
\newblock \href{https://doi.org/10.1038/35005001}{Quantum information and
  computation}.
\newblock {\em Nature}, 404(6775):247--255, Mar 2000.

\bibitem{ref:23}
Alexandre Blais, Ren-Shou Huang, Andreas Wallraff, S.~M. Girvin, and R.~J.
  Schoelkopf.
\newblock \href{https://link.aps.org/doi/10.1103/PhysRevA.69.062320}{Cavity
  quantum electrodynamics for superconducting electrical circuits: An
  architecture for quantum computation}.
\newblock {\em Phys. Rev. A}, 69:062320, Jun 2004.

\bibitem{ref:39}
Joseph~W. Britton, Brian~C. Sawyer, Adam~C. Keith, C.-C.~Joseph Wang, James~K.
  Freericks, Hermann Uys, Michael~J. Biercuk, and John~J. Bollinger.
\newblock \href{https://doi.org/10.1038/nature10981}{Engineered two-dimensional
  Ising interactions in a trapped-ion quantum simulator with hundreds of
  spins}.
\newblock {\em Nature}, 484(7395):489--492, Apr 2012.

\bibitem{ref:34}
Colin~D. Bruzewicz, John Chiaverini, Robert McConnell, and Jeremy~M. Sage.
\newblock \href{http://dx.doi.org/10.1063/1.5088164}{Trapped-ion quantum
  computing: Progress and challenges}.
\newblock {\em Applied Physics Reviews}, 6(2):021314, Jun 2019.

\bibitem{ref:41}
R.~Barends, J.~Kelly, A.~Megrant, A.~Veitia, D.~Sank, E.~Jeffrey, T.~C. White,
  J.~Mutus, A.~G. Fowler, B.~Campbell, Y.~Chen, Z.~Chen, B.~Chiaro,
  A.~Dunsworth, C.~Neill, P.~O'Malley, P.~Roushan, A.~Vainsencher, J.~Wenner,
  A.~N. Korotkov, A.~N. Cleland, and John~M. Martinis.
\newblock \href{https://doi.org/10.1038/nature13171}{Superconducting quantum
  circuits at the surface code threshold for fault tolerance}.
\newblock {\em Nature}, 508(7497):500--503, Apr 2014.

\bibitem{ref:42}
M.~W. Johnson, M.~H.~S. Amin, S.~Gildert, T.~Lanting, F.~Hamze, N.~Dickson,
  R.~Harris, A.~J. Berkley, J.~Johansson, P.~Bunyk, E.~M. Chapple, C.~Enderud,
  J.~P. Hilton, K.~Karimi, E.~Ladizinsky, N.~Ladizinsky, T.~Oh, I.~Perminov,
  C.~Rich, M.~C. Thom, E.~Tolkacheva, C.~J.~S. Truncik, S.~Uchaikin, J.~Wang,
  B.~Wilson, and G.~Rose.
\newblock \href{https://doi.org/10.1038/nature10012}{Quantum annealing with
  manufactured spins}.
\newblock {\em Nature}, 473(7346):194--198, May 2011.

\bibitem{ref:54}
Marcin Markiewicz and Marcin Wie\ifmmode~\acute{s}\else \'{s}\fi{}niak.
\newblock \href{https://link.aps.org/doi/10.1103/PhysRevA.79.054304}{Perfect
  state transfer without state initialization and remote collaboration}.
\newblock {\em Phys. Rev. A}, 79:054304, May 2009.

\bibitem{ref:1}
Sougato Bose.
\newblock
  \href{https://pdfs.semanticscholar.org/d178/cd5e21421e6228690f80d99db572494c2612.pdf}{Quantum
  Communication through an Unmodulated Spin Chain}.
\newblock {\em Phys. Rev. Lett.}, 91(207901), Nov 2003.

\bibitem{ref:15}
Tobias~J. Osborne.
\newblock \href{https://link.aps.org/doi/10.1103/PhysRevB.74.094411}{Statics
  and dynamics of quantum $XY$ and Heisenberg systems on graphs}.
\newblock {\em Phys. Rev. B}, 74:094411, Sep 2006.

\bibitem{ref:49}
Daniel Burgarth, Vittorio Giovannetti, and Sougato Bose.
\newblock \href{https://doi.org/10.10882F0305-44702F382F302F013}{Efficient and
  perfect state transfer in quantum chains}.
\newblock {\em Journal of Physics A: Mathematical and General},
  38(30):6793--6802, jul 2005.

\bibitem{ref:57}
Ashok Ajoy and Paola Cappellaro.
\newblock
  \href{https://link.aps.org/doi/10.1103/PhysRevA.85.042305}{Mixed-state
  quantum transport in correlated spin networks}.
\newblock {\em Phys. Rev. A}, 85:042305, Apr 2012.

\bibitem{ref:18}
Alastair Kay.
\newblock \href{http://dx.doi.org/10.1142/S0219749910006514}{Perfect,
  efficient, state transfer and its applications as a constructive tool}.
\newblock {\em International Journal of Quantum Information}, 08(04):641–676,
  Jun 2010.

\bibitem{ref:60}
V.~Subrahmanyam.
\newblock
  \href{https://link.aps.org/doi/10.1103/PhysRevA.69.034304}{Entanglement
  dynamics and quantum-state transport in spin chains}.
\newblock {\em Phys. Rev. A}, 69:034304, Mar 2004.

\bibitem{ref:61}
Antoni Wojcik, Tomasz Luczak, Pawel Kurzynski, Andrzej Grudka, Tomasz Gdala,
  and Malgorzata Bednarska.
\newblock \href{https://link.aps.org/doi/10.1103/PhysRevA.75.022330}{Multiuser
  quantum communication networks}.
\newblock {\em Phys. Rev. A}, 75:022330, Feb 2007.

\bibitem{ref:66}
R{\'{u}}ben Sousa and Yasser Omar.
\newblock \href{https://doi.org/10.10882F1367-26302F162F122F123003}{Pretty good
  state transfer of entangled states through quantum spin chains}.
\newblock {\em New Journal of Physics}, 16(12):123003, dec 2014.

\bibitem{ref:67}
Leonardo Banchi, Gabriel Coutinho, Chris Godsil, and Simone Severini.
\newblock \href{https://doi.org/10.1063/1.4978327}{Pretty good state transfer
  in qubit chains - The Heisenberg Hamiltonian}.
\newblock {\em Journal of Mathematical Physics}, 58(3):032202, 2017.

\bibitem{ref:68}
Alastair Kay.
\newblock \href{https://arxiv.org/abs/1906.06223}{Perfect and Pretty Good State
  Transfer for Field-Free Heisenberg Chains}, 2019.

\bibitem{ref:58}
Gabriel Coutinho and Henry Liu.
\newblock \href{http://dx.doi.org/10.1137/140989510}{No Laplacian Perfect State
  Transfer in Trees}.
\newblock {\em SIAM Journal on Discrete Mathematics}, 29(4):2179–2188, Jan
  2015.

\bibitem{ref:12}
Ethan Ackelsberg, Zachary Brehm, Ada Chan, Joshua Mundinger, and Christino
  Tamon.
\newblock
  \href{http://www.sciencedirect.com/science/article/pii/S002437951630180X}{Laplacian
  state transfer in coronas}.
\newblock {\em Linear Algebra and its Applications}, 506:154 -- 167, 2016.

\bibitem{ref:4}
Matthias Christandl, Nilanjana Datta, Artur Ekert, and Andrew~J. Landahl.
\newblock
  \href{https://journals.aps.org/prl/abstract/10.1103/PhysRevLett.92.187902}{Perfect
  State Transfer in Quantum Spin Networks}.
\newblock {\em Phys. Rev. Lett.}, 92:187902, May 2004.

\bibitem{ref:5}
Matthias Christandl, Nilanjana Datta, Tony~C. Dorlas, Artur Ekert, Alastair
  Kay, and Andrew~J. Landahl.
\newblock \href{https://link.aps.org/doi/10.1103/PhysRevA.71.032312}{Perfect
  transfer of arbitrary states in quantum spin networks}.
\newblock {\em Phys. Rev. A}, 71:032312, Mar 2005.

\bibitem{kendon2011perfect}
Vivien~M Kendon and Christino Tamon.
\newblock
  \href{https://www.ingentaconnect.com/content/asp/jctn/2011/00000008/00000003/art00015}{Perfect
  state transfer in quantum walks on graphs}.
\newblock {\em Journal of Computational and Theoretical Nanoscience},
  8(3):422--433, 2011.

\bibitem{ref:50}
John Brown, Chris Godsil, Devlin Mallory, Abigail Raz, and Christino Tamon.
\newblock \href{https://dl.acm.org/doi/10.5555/2481614.2481624}{Perfect State
  Transfer on Signed Graphs}.
\newblock {\em Quantum Info. Comput.}, 13(5–6):511–530, May 2013.

\bibitem{ref:19}
Chris Godsil and Sabrina Lato.
\newblock \href{https://arxiv.org/abs/2002.04666v1}{Perfect State Transfer on
  Oriented Graphs}, 2020.

\bibitem{ref:53}
Marzieh Asoudeh and Vahid Karimipour.
\newblock \href{https://doi.org/10.1007/s11128-013-0676-8}{Perfect state
  transfer on spin-1 chains}.
\newblock {\em Quantum Information Processing}, 13(3):601--614, Mar 2014.

\bibitem{ref:21}
Jafarizadeh M.A., Sufiani R., Azimi M., and et. al.
\newblock \href{https://doi.org/10.1007/s11128-011-0237-y}{Perfect state
  transfer over interacting boson networks associated with group schemes}.
\newblock {\em Quantum Inf Process}, 11:171–187, 2012.

\bibitem{ref:2}
Alastair Kay.
\newblock \href{https://link.aps.org/doi/10.1103/PhysRevA.84.022337}{Basics of
  perfect communication through quantum networks}.
\newblock {\em Phys. Rev. A}, 84:022337, Aug 2011.

\bibitem{ref:16}
V.~Kostak, G.~M. Nikolopoulos, and I.~Jex.
\newblock \href{https://link.aps.org/doi/10.1103/PhysRevA.75.042319}{Perfect
  state transfer in networks of arbitrary topology and coupling configuration}.
\newblock {\em Phys. Rev. A}, 75:042319, Apr 2007.

\bibitem{ref:56}
Sougato Bose, Andra Casaccino, Stefano Mancini, and Simone Severini.
\newblock \href{https://doi.org/10.1142/S0219749909005389}{Communication in XYZ
  All-to-All Quantum Networks with a Missing Link}.
\newblock {\em International Journal of Quantum Information}, 07(04):713--723,
  2009.

\bibitem{ref:59}
Christopher Facer, Jason Twamley, and James Cresser.
\newblock \href{http://dx.doi.org/10.1103/PhysRevA.77.012334}{Quantum Cayley
  networks of the hypercube}.
\newblock {\em Physical Review A}, 77(1), Jan 2008.

\bibitem{ref:62}
Peter~J. Pemberton-Ross and Alastair Kay.
\newblock
  \href{https://link.aps.org/doi/10.1103/PhysRevLett.106.020503}{Perfect
  Quantum Routing in Regular Spin Networks}.
\newblock {\em Phys. Rev. Lett.}, 106:020503, Jan 2011.

\bibitem{ref:64}
Xiang Zhan, Hao Qin, Zhi-hao Bian, Jian Li, and Peng Xue.
\newblock \href{http://dx.doi.org/10.1103/PhysRevA.90.012331}{Perfect state
  transfer and efficient quantum routing: A discrete-time quantum-walk
  approach}.
\newblock {\em Physical Review A}, 90(1), Jul 2014.

\bibitem{ref:29}
John Preskill.
\newblock \href{http://dx.doi.org/10.22331/q-2018-08-06-79}{Quantum Computing
  in the NISQ era and beyond}.
\newblock {\em Quantum}, 2:79, Aug 2018.

\bibitem{ref:55}
Michael~A. Nielsen and Isaac~L. Chuang.
\newblock {\em Quantum Computation and Quantum Information: 10th Anniversary
  Edition}.
\newblock Cambridge University Press, USA, 10th edition, 2011.

\bibitem{ref:17}
Norman Margolus and Lev~B. Levitin.
\newblock \href{https://doi.org/10.1016/S0167-2789(98)00054-2}{The maximum
  speed of dynamical evolution}.
\newblock {\em Physica D: Nonlinear Phenomena}, 120(1):188 -- 195, 1998.
\newblock Proceedings of the Fourth Workshop on Physics and Consumption.

\bibitem{west2001introduction}
Douglas~Brent West et~al.
\newblock {\em Introduction to graph theory}, volume~2.
\newblock Prentice hall Upper Saddle River, 2001.

\bibitem{ref:6}
Bibhas Adhikari, Amrik Singh, and Sandeep~Kumar Yadav.
\newblock \href{https://arxiv.org/abs/1908.10018}{Corona product of signed
  graphs and its application to signed network modelling}.
\newblock {\em arXiv}, 1908.10018, 2019.

\bibitem{ref:8}
Rohan Sharma and Bibhas Adhikari.
\newblock \href{https://arxiv.org/abs/1509.08773}{Self-Coordinated Corona
  Graphs: a model for complex networks}.
\newblock {\em arXiv}, 1509.08773, 2015.

\bibitem{ref:9}
Rohan Sharma, Bibhas Adhikari, and Abhishek Mishra.
\newblock
  \href{http://www.sciencedirect.com/science/article/pii/S0166218X17300264}{Structural
  and spectral properties of corona graphs}.
\newblock {\em Discrete Applied Mathematics}, 228:14 -- 31, 2017.

\bibitem{ref:70}
E.~Ackelsberg, Z.~Brehm, A.~Chan, J.~Mundinger, and C.~Tamon.
\newblock \href{https://doi.org/10.37236/6145}{Quantum State Transfer in
  Coronas}.
\newblock {\em The Electronic Journal of Combinatorics}, 24(02), 2017.

\bibitem{harary1988survey}
Frank Harary, John~P Hayes, and Horng-Jyh Wu.
\newblock
  \href{https://deepblue.lib.umich.edu/bitstream/handle/2027.42/27522/0000566.pdf?sequence=1}{A
  survey of the theory of hypercube graphs}.
\newblock {\em Computers \& Mathematics with Applications}, 15(4):277--289,
  1988.

\bibitem{klavvzar2006counting}
Sandi Klav{\v{z}}ar.
\newblock
  \href{https://www.sciencedirect.com/science/article/pii/S0012365X06005085}{Counting
  hypercubes in hypercubes}.
\newblock {\em Discrete mathematics}, 306(22):2964--2967, 2006.

\bibitem{ref:7}
John Brown, Chris Godsil, Devlin Mallory, Abigail Raz, and Christino Tamon.
\newblock \href{http://dl.acm.org/citation.cfm?id=2481614.2481624}{Perfect
  State Transfer on Signed Graphs}.
\newblock {\em Quantum Info. Comput.}, 13(5-6):511--530, May 2013.

\bibitem{ref:47}
Robert~J. Chapman, Matteo Santandrea, Zixin Huang, Giacomo Corrielli, Andrea
  Crespi, Man-Hong Yung, Roberto Osellame, and Alberto Peruzzo.
\newblock \href{http://dx.doi.org/10.1038/ncomms11339}{Experimental perfect
  state transfer of an entangled photonic qubit}.
\newblock {\em Nature Communications}, 7(1), Apr 2016.

\bibitem{ref:14}
Fei Yan, Philip Krantz, Youngkyu Sung, Morten Kjaergaard, Daniel~L. Campbell,
  Terry~P. Orlando, Simon Gustavsson, and William~D. Oliver.
\newblock
  \href{https://link.aps.org/doi/10.1103/PhysRevApplied.10.054062}{Tunable
  Coupling Scheme for Implementing High-Fidelity Two-Qubit Gates}.
\newblock {\em Phys. Rev. Applied}, 10:054062, Nov 2018.

\bibitem{ref:35}
J.~I. Cirac and P.~Zoller.
\newblock \href{https://link.aps.org/doi/10.1103/PhysRevLett.74.4091}{Quantum
  Computations with Cold Trapped Ions}.
\newblock {\em Phys. Rev. Lett.}, 74:4091--4094, May 1995.

\bibitem{ref:20}
Marek Sawerwain and Roman Gielerak.
\newblock Transfer of quantum continuous variable and qudit states in quantum
  networks.
\newblock In Andrzej Kwiecie{\'{n}}, Piotr Gaj, and Piotr Stera, editors, {\em
  Computer Networks}, pages 63--72, Berlin, Heidelberg, 2012. Springer Berlin
  Heidelberg.

\bibitem{ref:51}
Tong Liu, Qi-Ping Su, Jin-Hu Yang, Yu~Zhang, Shao-Jie Xiong, Jin-Ming Liu, and
  Chui-Ping Yang.
\newblock \href{https://europepmc.org/articles/PMC5539217}{Transferring
  arbitrary d-dimensional quantum states of a superconducting transmon qudit in
  circuit QED}.
\newblock {\em Scientific reports}, 7(1):7039, August 2017.

\bibitem{ref:52}
M~A Jafarizadeh, R~Sufiani, S~F Taghavi, and E~Barati.
\newblock \href{https://doi.org/10.10882F1751-81132F412F472F475302}{Optimal
  transfer of a d-level quantum state over pseudo-distance-regular networks}.
\newblock {\em Journal of Physics A: Mathematical and Theoretical},
  41(47):475302, oct 2008.

\bibitem{ref:26}
Hannes Bernien, Sylvain Schwartz, Alexander Keesling, Harry Levine, Ahmed
  Omran, Hannes Pichler, Soonwon Choi, Alexander~S. Zibrov, Manuel Endres,
  Markus Greiner, Vladan Vuleti{\"A}{\textdaggerdbl}, and Mikhail~D. Lukin.
\newblock \href{https://doi.org/10.1038/nature24622}{Probing many-body dynamics
  on a 51-atom quantum simulator}.
\newblock {\em Nature}, 551(7682):579--584, Nov 2017.

\bibitem{ref:27}
Chao Song, Kai Xu, Wuxin Liu, Chui-ping Yang, Shi-Biao Zheng, Hui Deng, Qiwei
  Xie, Keqiang Huang, Qiujiang Guo, Libo Zhang, Pengfei Zhang, Da~Xu, Dongning
  Zheng, Xiaobo Zhu, H.~Wang, Y.-A. Chen, C.-Y. Lu, Siyuan Han, and Jian-Wei
  Pan.
\newblock
  \href{https://link.aps.org/doi/10.1103/PhysRevLett.119.180511}{10-Qubit
  Entanglement and Parallel Logic Operations with a Superconducting Circuit}.
\newblock {\em Phys. Rev. Lett.}, 119:180511, Nov 2017.

\bibitem{ref:28}
J.~Kelly, R.~Barends, A.~G. Fowler, A.~Megrant, E.~Jeffrey, T.~C. White,
  D.~Sank, J.~Y. Mutus, B.~Campbell, Yu~Chen, Z.~Chen, B.~Chiaro, A.~Dunsworth,
  I.C. Hoi, C.~Neill, P.~J.~J. O'Malley, C.~Quintana, P.~Roushan,
  A.~Vainsencher, J.~Wenner, A.~N. Cleland, and John~M. Martinis.
\newblock \href{https://doi.org/10.1038/nature14270}{State preservation by
  repetitive error detection in a superconducting quantum circuit}.
\newblock {\em Nature}, 519(7541):66--69, Mar 2015.

\bibitem{ref:30}
David~P. DiVincenzo.
\newblock
  \href{http://dx.doi.org/10.1002/1521-3978(200009)48:9/11<771::AID-PROP771>3.0.CO;2-E}{The
  Physical Implementation of Quantum Computation}.
\newblock {\em Fortschritte der Physik}, 48(9-11):771–783, Sep 2000.

\bibitem{ref:31}
Frank Arute, Kunal Arya, Ryan Babbush, Dave Bacon, Joseph~C. Bardin, Rami
  Barends, Rupak Biswas, Sergio Boixo, Fernando G. S.~L. Brandao, David~A.
  Buell, Brian Burkett, Yu~Chen, Zijun Chen, Ben Chiaro, Roberto Collins,
  William Courtney, Andrew Dunsworth, Edward Farhi, Brooks Foxen, Austin
  Fowler, Craig Gidney, Marissa Giustina, Rob Graff, Keith Guerin, Steve
  Habegger, Matthew~P. Harrigan, Michael~J. Hartmann, Alan Ho, Markus Hoffmann,
  Trent Huang, Travis~S. Humble, Sergei~V. Isakov, Evan Jeffrey, Zhang Jiang,
  Dvir Kafri, Kostyantyn Kechedzhi, Julian Kelly, Paul~V. Klimov, Sergey Knysh,
  Alexander Korotkov, Fedor Kostritsa, David Landhuis, Mike Lindmark, Erik
  Lucero, Dmitry Lyakh, Salvatore Mandr{\~A} , Jarrod~R. McClean, Matthew
  McEwen, Anthony Megrant, Xiao Mi, Kristel Michielsen, Masoud Mohseni, Josh
  Mutus, Ofer Naaman, Matthew Neeley, Charles Neill, Murphy~Yuezhen Niu, Eric
  Ostby, Andre Petukhov, John~C. Platt, Chris Quintana, Eleanor~G. Rieffel,
  Pedram Roushan, Nicholas~C. Rubin, Daniel Sank, Kevin~J. Satzinger, Vadim
  Smelyanskiy, Kevin~J. Sung, Matthew~D. Trevithick, Amit Vainsencher, Benjamin
  Villalonga, Theodore White, Z.~Jamie Yao, Ping Yeh, Adam Zalcman, Hartmut
  Neven, and John~M. Martinis.
\newblock \href{https://doi.org/10.1038/s41586-019-1666-5}{Quantum supremacy
  using a programmable superconducting processor}.
\newblock {\em Nature}, 574(7779):505--510, Oct 2019.

\bibitem{ref:33}
Alexandre Blais, Jay Gambetta, A.~Wallraff, D.~I. Schuster, S.~M. Girvin, M.~H.
  Devoret, and R.~J. Schoelkopf.
\newblock
  \href{https://link.aps.org/doi/10.1103/PhysRevA.75.032329}{Quantum-information
  processing with circuit quantum electrodynamics}.
\newblock {\em Phys. Rev. A}, 75:032329, Mar 2007.

\bibitem{ref:32}
Yu-xi Liu, C.~P. Sun, and Franco Nori.
\newblock \href{http://dx.doi.org/10.1103/PhysRevA.74.052321}{Scalable
  superconducting qubit circuits using dressed states}.
\newblock {\em Physical Review A}, 74(5), Nov 2006.

\bibitem{ref:69}
Leniency Marbaniang and Kamalika Datta.
\newblock \href{https://doi.org/10.1142/S0218126619500841}{Efficient Design of
  Quantum Circuits Using Nearest Neighbor Constraint in 3D Architecture}.
\newblock {\em Journal of Circuits, Systems and Computers}, 28(05):1950084,
  2019.

\bibitem{ref:38}
M.~Viteau, M.~G. Bason, J.~Radogostowicz, N.~Malossi, D.~Ciampini, O.~Morsch,
  and E.~Arimondo.
\newblock
  \href{https://link.aps.org/doi/10.1103/PhysRevLett.107.060402}{Rydberg
  Excitations in Bose-Einstein Condensates in Quasi-One-Dimensional Potentials
  and Optical Lattices}.
\newblock {\em Phys. Rev. Lett.}, 107:060402, Aug 2011.

\bibitem{ref:40}
D.~M. Zajac, T.~M. Hazard, X.~Mi, E.~Nielsen, and J.~R. Petta.
\newblock
  \href{https://link.aps.org/doi/10.1103/PhysRevApplied.6.054013}{Scalable Gate
  Architecture for a One-Dimensional Array of Semiconductor Spin Qubits}.
\newblock {\em Phys. Rev. Applied}, 6:054013, Nov 2016.

\bibitem{ref:36}
D.~Rosenberg, D.~Kim, R.~Das, D.~Yost, S.~Gustavsson, D.~Hover, P.~Krantz,
  A.~Melville, L.~Racz, G.~O. Samach, S.~J. Weber, F.~Yan, J.~L. Yoder, A.~J.
  Kerman, and W.~D. Oliver.
\newblock \href{https://doi.org/10.1038/s41534-017-0044-0}{3D integrated
  superconducting qubits}.
\newblock {\em npj Quantum Information}, 3(1):42, Oct 2017.

\bibitem{ref:37}
Salvatore~S. Elder, Christopher~S. Wang, Philip Reinhold, Connor~T. Hann,
  Kevin~S. Chou, Brian~J. Lester, Serge Rosenblum, Luigi Frunzio, Liang Jiang,
  and Robert~J. Schoelkopf.
\newblock
  \href{https://link.aps.org/doi/10.1103/PhysRevX.10.011001}{High-Fidelity
  Measurement of Qubits Encoded in Multilevel Superconducting Circuits}.
\newblock {\em Phys. Rev. X}, 10:011001, Jan 2020.

\bibitem{ref:43}
Yu~Chen, C.~Neill, P.~Roushan, N.~Leung, M.~Fang, R.~Barends, J.~Kelly,
  B.~Campbell, Z.~Chen, B.~Chiaro, A.~Dunsworth, E.~Jeffrey, A.~Megrant, J.~Y.
  Mutus, P.~J.~J. O'Malley, C.~M. Quintana, D.~Sank, A.~Vainsencher, J.~Wenner,
  T.~C. White, Michael~R. Geller, A.~N. Cleland, and John~M. Martinis.
\newblock \href{https://link.aps.org/doi/10.1103/PhysRevLett.113.220502}{Qubit
  Architecture with High Coherence and Fast Tunable Coupling}.
\newblock {\em Phys. Rev. Lett.}, 113:220502, Nov 2014.

\bibitem{ref:44}
R.~Barends, J.~Kelly, A.~Megrant, A.~Veitia, D.~Sank, E.~Jeffrey, T.~C. White,
  J.~Mutus, A.~G. Fowler, B.~Campbell, Y.~Chen, Z.~Chen, B.~Chiaro,
  A.~Dunsworth, C.~Neill, P.~O'Malley, P.~Roushan, A.~Vainsencher, J.~Wenner,
  A.~N. Korotkov, A.~N. Cleland, and John~M. Martinis.
\newblock \href{https://doi.org/10.1038/nature13171}{Superconducting quantum
  circuits at the surface code threshold for fault tolerance}.
\newblock {\em Nature}, 508(7497):500--503, Apr 2014.

\bibitem{ref:25}
Morten Kjaergaard, Mollie~E. Schwartz, Jochen Braumüller, Philip Krantz, Joel
  I.-J. Wang, Simon Gustavsson, and William~D. Oliver.
\newblock
  \href{http://dx.doi.org/10.1146/annurev-conmatphys-031119-050605}{Superconducting
  Qubits: Current State of Play}.
\newblock {\em Annual Review of Condensed Matter Physics}, 11(1):369–395, Mar
  2020.

\bibitem{ref:24}
David Zueco, Georg~M. Reuther, Sigmund Kohler, and Peter H\"anggi.
\newblock
  \href{https://link.aps.org/doi/10.1103/PhysRevA.80.033846}{Qubit-oscillator
  dynamics in the dispersive regime: Analytical theory beyond the rotating-wave
  approximation}.
\newblock {\em Phys. Rev. A}, 80:033846, Sep 2009.

\bibitem{ref:22}
M.~D. Jenkins, D.~Zueco, O.~Roubeau, G.~Aromí, J.~Majer, and F.~Luis.
\newblock \href{http://dx.doi.org/10.1039/C6DT02664H}{A scalable architecture
  for quantum computation with molecular nanomagnets}.
\newblock {\em Dalton Trans.}, 45:16682--16693, 2016.

\bibitem{ref:45}
G~Wendin.
\newblock \href{http://dx.doi.org/10.1088/1361-6633/aa7e1a}{Quantum information
  processing with superconducting circuits: a review}.
\newblock {\em Reports on Progress in Physics}, 80(10):106001, Sep 2017.

\bibitem{ref:46}
Jens Koch, Terri~M. Yu, Jay Gambetta, A.~A. Houck, D.~I. Schuster, J.~Majer,
  Alexandre Blais, M.~H. Devoret, S.~M. Girvin, and R.~J. Schoelkopf.
\newblock
  \href{https://link.aps.org/doi/10.1103/PhysRevA.76.042319}{Charge-insensitive
  qubit design derived from the Cooper pair box}.
\newblock {\em Phys. Rev. A}, 76:042319, Oct 2007.

\end{thebibliography}
